\documentclass[lettersize,journal]{IEEEtran}
\usepackage{multirow}
\usepackage{url}
\usepackage{float}
\usepackage{hyperref}
\usepackage{multicol}
\usepackage{colortbl}
\usepackage{dcolumn}
\usepackage{enumitem}
\usepackage{tasks}
\usepackage{algorithmic}
\def\BibTeX{{\rm B\kern-.05em{\sc i\kern-.025em b}\kern-.08em
    T\kern-.1667em\lower.7ex\hbox{E}\kern-.125emX}}
\usepackage{balance}

\NewTasks[style=itemize, label-width=1em]{mylist}[\item](2)
\newlist{circlenum}{enumerate}{1}
\setlist[circlenum]{
  label=\raisebox{-0.3ex}{\textcircled{\scriptsize\arabic*}}, 
  labelwidth=1.2em,
  leftmargin=1.8em,
  itemsep=0pt,
  parsep=2pt
}
\usepackage{amssymb}
\usepackage{amsthm}
\usepackage{amsmath}
\usepackage{multirow}
\usepackage{comment}
\usepackage{tablefootnote}
\usepackage[colorinlistoftodos,prependcaption,textsize=small]{todonotes}
\usepackage{cite}
\usepackage[ruled, linesnumbered, vlined]{algorithm2e}
\usepackage{soul, color, xcolor}
\usepackage{amsmath}
\newtheorem{thm}{Theorem}

\newtheorem{Definition}{Definition}

\newcommand{\setfootnotemark}{%
  \refstepcounter{footnote}%
\footnotemark[\value{footnote}]}

\begin{document}
\title{Technical Report:
Toward Patch Robustness Certification and Detection for Deep Learning Systems 
Beyond Consistent Samples}

 \author{
Qilin Zhou, Zhengyuan Wei, Haipeng Wang, Zhuo Wang, and W.K. Chan
    \thanks{© 20XX IEEE.  Personal use of this material is permitted.  Permission from IEEE must be obtained for all other uses, in any current or future media, including reprinting/republishing this material for advertising or promotional purposes, creating new collective works, for resale or redistribution to servers or lists, or reuse of any copyrighted component of this work in other works.}
    }


\maketitle

\begin{abstract}
Patch robustness certification is an emerging kind of provable defense technique against adversarial patch attacks for deep learning systems.
Certified detection ensures the detection of all patched harmful versions of certified samples, which mitigates the failures of empirical defense techniques that could (easily) be compromised.
However, existing certified detection methods are ineffective in certifying samples that are misclassified or whose mutants are inconsistently predicted to different labels.
This paper proposes HiCert, a novel masking-based certified detection technique.  
By focusing on the problem of mutants predicted with a label different from the true label with our formal analysis, 
HiCert formulates a novel formal relation between harmful samples generated by identified loopholes and their benign counterparts. 
By checking the bound of the maximum confidence among these potentially harmful (i.e., inconsistent) mutants of each benign sample, HiCert ensures that each harmful sample either has the minimum confidence among mutants that are predicted the same as the harmful sample itself below this bound, 
or
has at least one mutant predicted with a label different from the harmful sample itself, 
formulated after two novel insights.
As such, HiCert systematically certifies those inconsistent samples and consistent samples to a large extent.
To our knowledge, HiCert is the \emph{first} work capable of providing such a comprehensive patch robustness certification for certified detection.
Our experiments show the high effectiveness of HiCert with a new state-of-the-art performance:
It certifies significantly more benign samples, including those inconsistent and consistent, and achieves significantly higher accuracy on those samples without warnings and a significantly lower false silent ratio.
Moreover, on actual patch attacks, its defense success ratio is significantly higher than its peers.
\end{abstract}

\section*{Nomenclature}
\begin{mylist}

    \item $ x $: a benign sample
    \item $ x' $: $x$'s harmful sample 
    \item $ \hat{x} $: an arbitrary sample

    \item $f$: a classification model
    \item $v$: a certification function
    \item $w$: a warning function
    \item $\textsc{p}$: a patch region
    \item $\mathbb{P}$: a patch region set
    \item $\textsc{m}$: a mask
    \item $\hat{x}_\textsc{m}$: a mutant of $\hat{x}$ for $\textsc{m}$
    \item $\mathbb{M}_\mathbb{P}$: a covering mask set for $\mathbb{P}$
    \item $\textsc{m}_\textsc{p}$: a mask covering the patch region $\textsc{p}$
    \item $\mathbb{A}_\mathbb{P}(x)$: an attack constraint set for $x$
    \item $D$: a certified detection defender
    \item  \text{$ y_0 $: the true label of a sample}

\end{mylist}

\section{Introduction}
\label{sec:introduction}

\IEEEPARstart{R}{eliability} of safety-critical deep learning (DL) systems, such as autonomous vehicles and robots, is threatened by adversarial attacks 
\cite{hussain2024evaluating, Xu2023ASQ, AlMaliki2023Toward,zhang2024uniada,Huang2024FocusShifting,Qi2022Detection},
particularly those that are physically realizable \cite{hussain2024evaluating} (see Fig.~\ref{fig:moti-attack}).
A major stereotype is patch adversarial attacks \cite{brown2017adversarial, eykholt2018robust,liu2020bias,wei2022adversarial, tao2023hardlabel, wei2023simultaneously}, which is a threat model for a deep learning (DL) system that is tricked into misclassifying an image sample by adding additional content (called a patch) to the sample on an arbitrary region (called a patch region) 
to produce a label differing from the ground truth \cite{kurakin2018adversarial, eykholt2018robust, xiao2018generating} (called harmful), consequently, producing an adversarial example.
Detecting these patched samples is desirable.
Yet, empirical detection and defense techniques \cite{chen2023jujutsu,naseer2019local,hayes2018visible,hussain2024evaluating} often fail against patch attacks unknown to them or even the known ones if their defense strategies are exposed to attackers \cite{chiang2020certified}.

\begin{figure}[b]
\centering
\includegraphics[width=\linewidth]{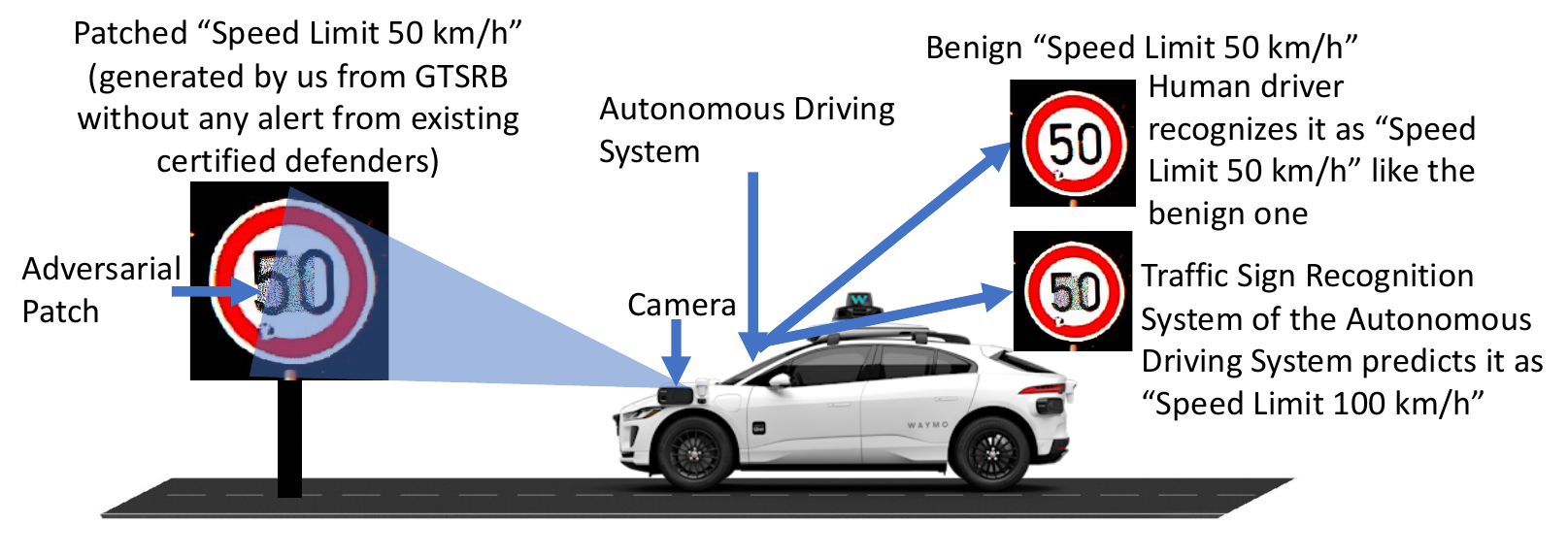}
\caption{One possible patch attack scenario targeting traffic sign recognition systems \cite{hussain2024evaluating}, further threatening the reliability of autonomous driving systems.
}
 \label{fig:moti-attack}
\end{figure} 

Certified detection for patch adversarial attacks \cite{li2022vip, patchcensor,mccoyd2020minority,han2021scalecert,xiang2021patchguard++,zhou2024crosscert} can significantly enhance the security of these systems and is emerging. 
{Its} goal 
is to formulate a provable framework
to cover as many benign samples of a DL system as possible \emph{with the \textbf{deterministic} guarantee of detecting \textbf{all} harmful patched versions of the covered benign samples (called certified samples)}, in the absence of the identity of these benign samples during the detection of the presence of an adversarial patch up to a given size.
Certification provides a formal detection property on these certified samples with \emph{all} their harmful patched versions, unachievable by pure empirical techniques.

To our knowledge, almost all effective certified detection defenders against patch adversarial attacks are masking-based \cite{patchcensor,li2022vip,xiang2021patchguard++,mccoyd2020minority, han2021scalecert}.
As harmful patched versions are considered dangerous (e.g., for safety-critical systems \cite{patchcensor}), 
they are specifically designed to detect all {harmful} patched versions of a certified benign sample meeting their inferable criteria while 
keeping the certified sample itself undetected.
To use the output of defender, if a warned sample is deemed potentially harmful, the system could be switched to a fallback strategy for handling with care, and an unwarned sample could be used as usual by safety-critical downstream tasks \cite{patchcensor}, as depicted in Fig.~\ref{fig:application}.
%
For example, inspections in many critical places, like checkpoints, can enhance the processing capabilities by applying DL models for facial recognition \cite{The_Standard_2024_Facial_recognition} or for item identification in the inspection systems.
 If the results without warning are reliable enough (i.e., lower the risk), for example, a made-up criminal, like patching an adversarial sticker on the face \cite{wei2022adversarial, wei2023simultaneously}, is unable to make oneself recognized as a lawful citizen without warning,
 or a graffitied six-gun is unable to make it recognized as a tin-opener without warning\footnote{Indeed, we did find a six-gun ImageNet sample indexed as \texttt{n04086273/ILSVRC2012\_val\_00000667} with such a threat on peer defenders during our experiments.},
then the checkpoint can significantly reduce the scrutiny on this group of people/items.
The 
systems in Fig.~\ref{fig:moti-attack} can also be equipped with the detection framework in Fig.~\ref{fig:application} to improve the reliability of autonomous driving systems.
Improving the autonomy of the downstream process in Fig.~\ref{fig:application} requires improving 
the overall sample quality in the silent path.

\begin{figure}[t]
\centering
\includegraphics[width=0.8\linewidth]{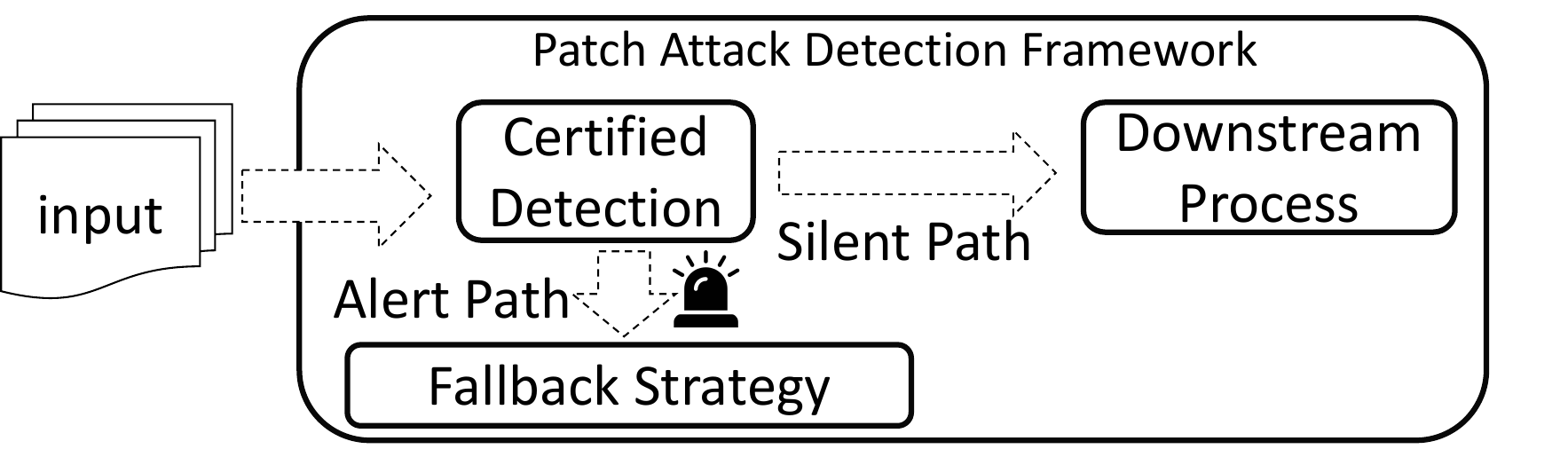}
\caption{
A patch attack detection framework with certified detection.
}\label{fig:application}

\end{figure}

If they have to detect a correctly predicted benign sample as warned, 
they 
cannot
provide any detection guarantee covering all harmful patched versions of this benign sample; thereby, a harmful patched version of the sample may slip through the detection framework to reach downstream operations, which is dangerous (e.g., undermining the ultimate purpose of using such a detection strategy to support fully autonomous safety-critical systems).%
\footnote{Ensuring system safety is typically prioritized over maximizing utility in safety-critical scenarios \cite{leveson2016engineering} (e.g., interruptive maintenance in nuclear power plants \cite{cowing2004dynamic} or high false alerts in earthquake detection systems \cite{minson2019limits}).
}
Furthermore, to our knowledge, effective certification of benign samples with incorrect prediction labels has been overlooked by existing works.
For example, an adversarial example may keep the incorrect prediction label and make the situation even worse (e.g., a cliff sign originally misclassified as a no-limit sign is originally warned, but it may be patched to make the warning silent).
%
This allows the harmful patched sample to bypass the detection guard provided by existing frameworks, 
preventing them from guaranteeing the detection of all harmful versions of the benign sample
(either by excluding such adversarial example counterparts of benign samples from the targets of warning to provide the guarantee, as represented by PatchCensor \cite{patchcensor} and ViP \cite{li2022vip}, 
or by including them as the targets of warning but cannot provide the guarantee, as represented by PG++ \cite{xiang2021patchguard++} and earlier techniques \cite{mccoyd2020minority,han2021scalecert}).
A vast majority of samples were successfully turned into harmful samples and bypassed their ``apparently stringent'' detection guards in our experiment (Section~\ref{sec:attack})
if a defender was unable to certify them. At the same time, state-of-the-art certified detection defenders fail to certify one out of every four ImageNet samples (see Table~\ref{tab:main_eva_results}).
Ideally, all harmful patched versions of benign samples (warned or unwarned, correctly or incorrectly predicted) should be detected by certified detection, even if it reduces the number of benign samples usable for downstream tasks in safety-critical systems for a higher safety-critical standard.%
\footnote{In fact, this is in line with the philosophy of certified accuracy of all certified detection techniques lower than the standard clean accuracy by excluding relatively less reliable benign samples from certification.}

\begin{figure}[tb]
\centering
\includegraphics[width=\linewidth]{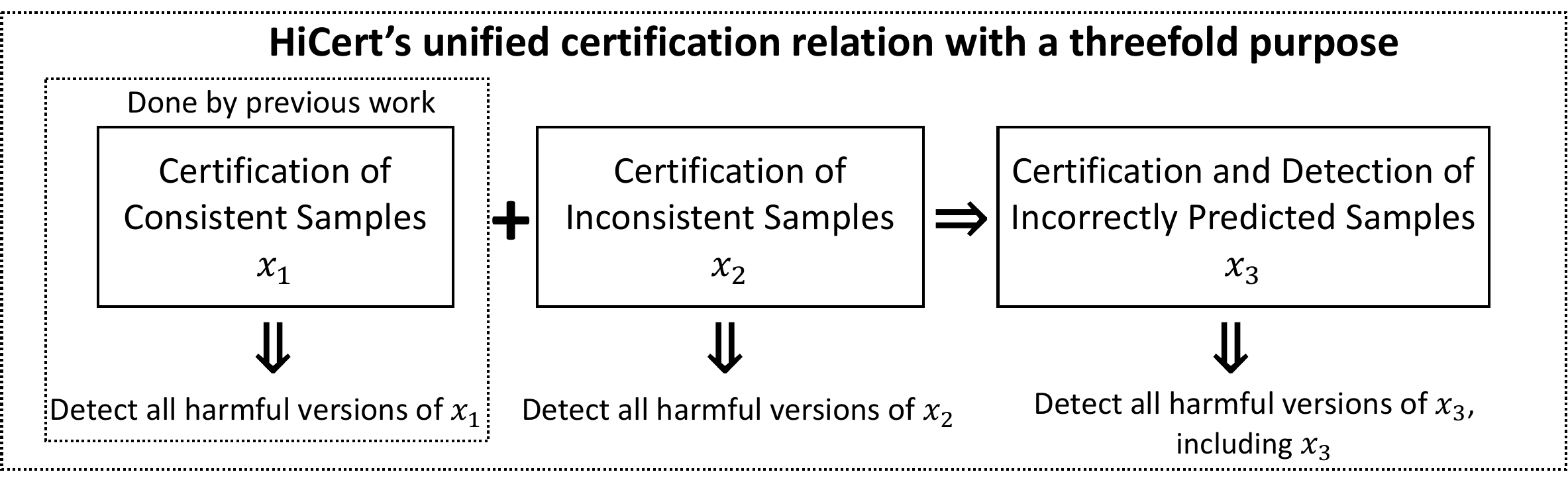}
\caption{HiCert and its three purposes achieved by a unified relation. 
The unified relation is presented as Thm. \ref{thm:Inconsistent-Max-Min}.
}
\label{fig:overall-effect}
\end{figure} 

We first describe some terminology to ease our introduction of this work.
Masking some region of a sample creates a (masked) mutant of the sample.
A mutant is called consistent if a DL model assigns the true label to it; otherwise called inconsistent.
A sample is called consistent if all of its mutants are consistent, otherwise called inconsistent.

In this paper, we propose a novel detection technique HiCert, the effect of which is depicted in Fig. \ref{fig:overall-effect}.
With respect to a given base DL system subject to certified detection's protection,
HiCert is the \emph{first} work to certify consistent and inconsistent samples homogeneously regardless of the correctness of the prediction label and the label distribution of mutants, 
which leads to the simultaneous certification and detection of incorrectly predicted benign samples.
All certified incorrectly predicted samples and all of their patched versions are systematically ruled out from the silent path depicted in Fig. \ref{fig:application}.
HiCert is a masking-based certified detection defender based on a novel checking strategy after a thorough analysis
---
It certifies a sample if either the prediction label of each mutant of the sample is the same as the true label or if all those mutants predicted differently have low confidence. 
It warns a sample if any mutant of the sample is not predicted with the prediction label of the sample or is low in confidence.
Like existing work in the field, 
we prove HiCert's soundness by theorem formulation (Thm. \ref{thm:Inconsistent-Max-Min}).


Our evaluation demonstrates the high effectiveness of HiCert.
For instance, HiCert is able to significantly reduce the gap between clean accuracy and certified accuracy from 9.1\% created by the previous SOTA technique ViP to 0.1\% on ImageNet with the patch region in size of 2\% 
(see Table~\ref{tab:main_eva_results}).
Our detailed analysis on ImageNet further shows that 
HiCert is significantly more effective than the peer techniques in certifying correctly predicted and incorrectly predicted samples in various metrics (see Table~\ref{tab:cases_imagenet_2}).
We perform a real patch adversarial attack to show the significantly higher effectiveness achieved by HiCert empirically compared to its peers in terms of defense success ratios (see Fig.~\ref{fig:real_attack}).
HiCert also shows significantly stronger certification performance when the patch size increases (see Fig.~\ref{fig:vary_patch_size}).

The contribution of this work is threefold.
(1) It proposes a novel patch attack defender, HiCert, which offers comprehensive certification coverage on sample variety in terms of prediction correctness of certified samples as well as the label distribution of mutants. 
(2) It formally proves HiCert's soundness and demonstrates its feasibility through implementation. 
(3) It presents an evaluation showcasing HiCert's high effectiveness and scalability, both empirically and theoretically.

The rest of this paper is organized as follows. \S\ref{sec:Preliminaries} revisits the preliminaries. \S\ref{sec:motiva} discusses the motivation and challenges. \S\ref{sec:proposal} to \S\ref{sec:eva} present HiCert and its evaluation. 
\S\ref{sec:related_work} reviews related work, and \S\ref{sec:Conclusion} concludes the paper.

\section{Preliminaries}\label{sec:Preliminaries}
This section revisits the preliminaries \cite{li2022vip, xiang2021patchguard++}.
%
%
%
%
%
We use \textsc{J} to denote an all-ones matrix 
and $+$, $-$, and $\odot$ to denote element-wise addition, subtraction, and multiplication operators. 



\subsection{Classification Model and Patch Attack}

Given an image sample $\hat{x}$ 
$\in \mathcal{X} \subset \mathbb{R}^{w \times h}$ and its true (class) label $y_0$ $\in \mathcal{Y} = \{0, 1, \cdots, |\mathcal{Y}|-1\}$, 
an image classification (deep learning) model $f: \mathcal{X} \rightarrow \mathcal{Y}$ 
takes $\hat{x}$ as input and produces a class label $f(\hat{x}) \in \mathcal{Y}$ with confidence ${f_{\textit{conf}}}(\hat{x}) \in (0,1)$.
If  $f(\hat{x})$ is the true label $y_0$ for $\hat{x}$, the sample $\hat{x}$ is called \textbf{correctly predicted}, otherwise \textbf{incorrectly predicted}.

Like previous work \cite{levine2020randomized, xiang2021patchguard}, we represent a contiguous square \textbf{patch region} by a binary matrix $\textsc{p} \in \mathbb{P} \subset \{0,1\}^{w \times h}$, where elements within the region are 1, otherwise 0, and
$\mathbb{P}$ represents the \textbf{set of all patch regions} that meet the predefined conditions (e.g., shape, size).
An attack constraint set $\mathbb{A}_\mathbb{P}({x})$  represents all patched versions of a benign sample $x$%
:
$\mathbb{A}_\mathbb{P}({x})$ = $\{{x}'' \mid {x}''=(\textsc{J}-\textsc{p})\odot {x}+\textsc{p}\odot {x}'' \land \textsc{p}\in \mathbb{P} \}$, 
where
${x}''$ is a patched version of $x$ such that
 an attacker can modify any pixels within a patch region \textsc{p} on $x$ ($\textsc{p}\odot {x}''$) while keeping all pixels outside the region unmodified ($(\textsc{J}-\textsc{p})\odot {x}$).
 The attacker can place \textsc{p} anywhere on $x$.
%
A patched version
$x'\in\mathbb{A}_\mathbb{P}({x})$ 
with incorrectly predicted label ($f(x')\neq y_0$) is called a \textbf{harmful sample}.
The patched sample in Fig. \ref{fig:moti-attack} is a harmful sample.

\subsection{Certified Detection}

A \textbf{certified detection defender}
$D = \langle f, w, v\rangle$ 
is a base classification model $f$ plus a warning function $w(\cdot)$ and a certification (verification) function $v(\cdot)$.
Given an input sample $\hat{x}$, 
(1) $f$ outputs a label $f(\hat{x})$, 
(2)
$w(\hat{x})$ returns \textit{True} if it detects $\hat{x}$ as harmful,
otherwise \textit{False}, and
(3)
$v(\hat{x})$ returns \textit{True} if $D$ certifies $\hat{x}$ (see Def.~\ref{def:cert}), otherwise \textit{False}. 
 In pre-deployment, $D$ uses  $v(.)$ to certify a benign sample, while in post-deployment, $D$ detects a sample as harmful by $w(.)$

\textbf{Attacker's objective on detection defenders}: An attacker \cite{tao2023hardlabel} aims to find $x$'s harmful sample $x'$
that can pass a (certified) detection defender $D = \langle f, w, v \rangle $ without being detected, i.e., find $x'\in\mathbb{A}_\mathbb{P}({x})$ such that
$f(x') \neq y_0  \land w(x') = False$. 

%

%
%
%
If a detection scheme guarantees the detection of all harmful samples of $x$ as warned, it is certified detection%
\footnote{There are different schools of thought on the definition of certified detection in the literature.
The other one is to detect all the changes in the prediction labels. See Section \ref{app:certified_detection}.} (Def.~\ref{def:cert}).

\begin{Definition}[Certified Detection]\label{def:cert}
A defender $D = \langle f, w, v\rangle$  certifies  $x$, i.e., $v({x})$ = \textit{True}, if 
 $[\forall {x}' \in \mathbb{A}_\mathbb{P}({x})$, $f({x}')\neq y_0$ implies $w({x}')=\textit{True}]$. 
 $D$ reports $x$ as its \textbf{certified sample}.%

\end{Definition}

If all benign samples can be certified,
then no harmful sample can escape from $D$'s detection guard (which achieves completeness),
a safety-critical downstream process can use the samples (with their predicted labels from $D$'s classifier) without any safety concerns if $D$'s warning function does not issue a warning. 
However, this assumption is challenging to achieve due to the imperfections of deep learning models, unless the detector is trivial (e.g., it certifies and warns for all possible samples).
We settle for the next best option: a higher proportion of benign samples that can be certified by $D$'s certification function indicates that fewer benign samples might lead to harmful samples infiltrating detection (which threatens downstream tasks). Therefore, downstream tasks can use the samples emitted by $D$, which do not trigger $D$ to issue warnings, with greater assurance.


\begin{figure}[tb]
\centering
\includegraphics[width=0.65\linewidth]{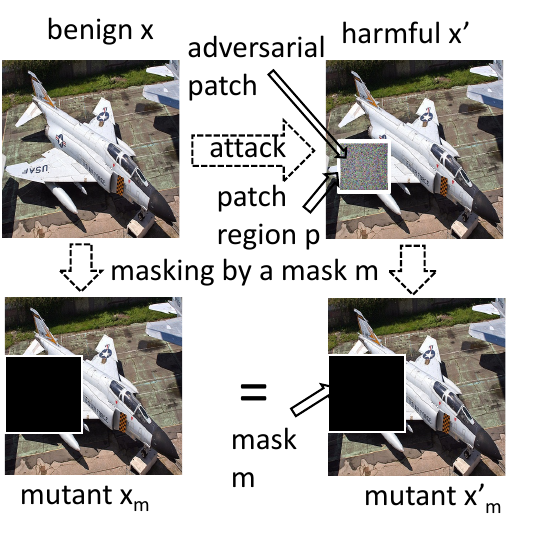}
\caption{ 
Illustration of the concepts of masking.
A military aircraft (from ImageNet \cite{deng2009imagenet}) may be patched to evade DL-based inspection.
 }\label{fig:concept}
\end{figure} 

\subsection{Masking-based Techniques, Types of Samples, and Mutants}\label{sec:OMA_related_work}
In the literature, almost all effective prior art in certified detection \cite{patchcensor,li2022vip,mccoyd2020minority,han2021scalecert,xiang2021patchguard++} against patch attacks are masking-based \cite{zhou2024crosscert}. 
They produce a specific set of mutants from a benign sample subject to certification and certify the sample if all of these mutants meet a certification condition.

\subsubsection{Masking a Sample to Produce Mutants}
A \textbf{mask} is represented by a binary matrix $\textsc{m} \in [0,1]^{w \times h}$, with elements set to 1 within the mask and 0 otherwise. 

We create the \textbf{mutant} $\hat{x}_\textsc{m}$ by removing the content of a sample $\hat{x}$ covered by a mask $\textsc{m}$: $\hat{x}_\textsc{m} = (\textsc{J}-\textsc{m}) \odot \hat{x}$.

%
A mask \textsc{m} \textit{covers} a patch region $\textsc{p}$ if and only if all elements of 1 in $\textsc{p}$
are elements of 1 in \textsc{m}, i.e.,
$\textsc{p}\odot\textsc{m}=\textsc{p}$. 
If 
a mask  $\textsc{m}_\textsc{p}$ 
covers a patch region
$\textsc{p}$, which further covers the difference between $x'$ and $x$,
the two mutants generated by this mask on $x'$ and $x$ 
will be identical due to the difference in content between them removed by the mask,
i.e., $\forall x'\in\{x'\mid{x}'=(\textsc{J}-\textsc{p})\odot {x}+\textsc{p}\odot {x}'\},
\textsc{p}\odot\textsc{m}_\textsc{p}=\textsc{p}\implies x_{\textsc{m}_\textsc{p}}={x}'_{\textsc{m}_\textsc{p}}$.

Fig.~\ref{fig:concept} illustrates these concepts on a military aircraft, which can be patched with paint to evade DL-based inspection.


To ease our presentation, we refer to \emph{placing a patch on a mutant} of a sample as a shorthand description of creating a patched version of the sample, such that the mask that produces the mutant covers the patch (region).

Since the specific patch region used by the attacker is unknown to a defender,
existing masking-based detection defenders 
commonly generate a \textbf{covering mask set} (referred to as window masks in \cite{xiang2021patchguard++}) $\mathbb{M}_\mathbb{P}$, ensuring that every patch region $\textsc{p}\in\mathbb{P}$ is covered by at least one mask in $\mathbb{M}_\mathbb{P}$ (i.e., $\forall \textsc{p}\in\mathbb{P}, \exists \textsc{m}_\textsc{p}\in\mathbb{M}_\mathbb{P}, \textsc{p}\odot{\textsc{m}_\textsc{p}}=\textsc{p}$) \cite{xiang2021patchguard++,patchcensor,li2022vip,xiang2022patchcleanser}.





\subsubsection{Mutation and Sample Types}\label{sec:oma}
A sample $\hat{x}$ satisfies the $D_\text{OMA}$ condition (Def. \ref{def:one-mask}) if all its mutants 
share the same prediction label.

\begin{figure}[t]
\centering
\includegraphics[width=0.99\linewidth]{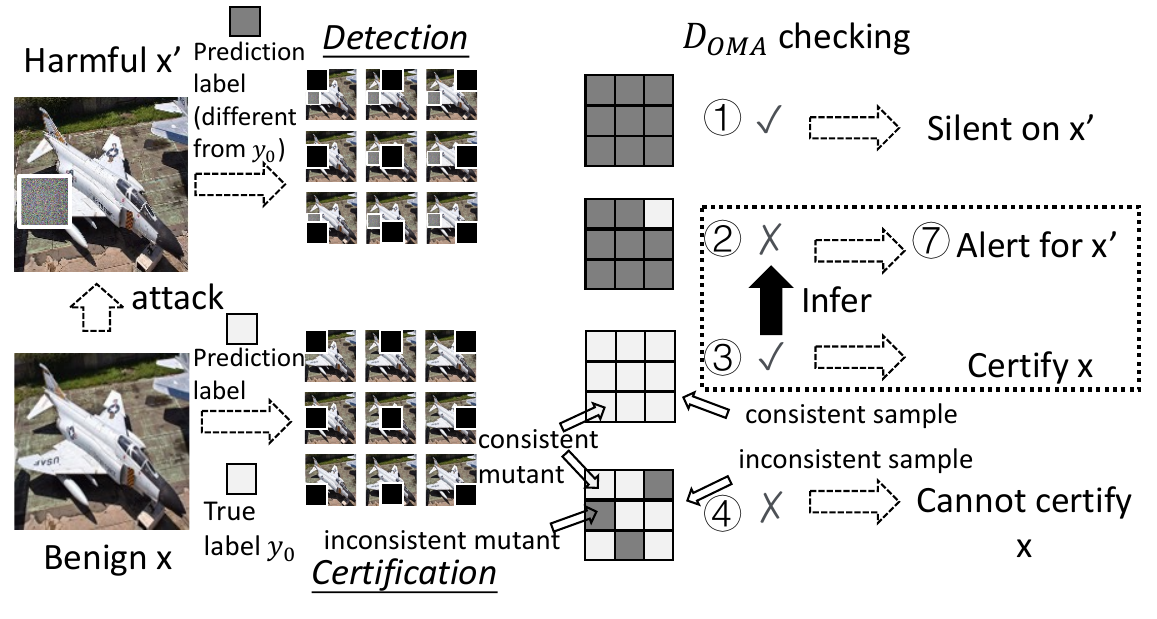}
\caption{
Illustration of the $D_\text{OMA}$ defender. 
}
\label{fig:oma}
\end{figure}


\begin{Definition}[{$D_\text{OMA}$ condition}]
\label{def:one-mask}

The One-Masking-Agreement condition ($D_\text{OMA}$ condition) is defined as
$[\forall \textsc{m} \in\mathbb{M}_\mathbb{P}, f(\hat{x}_\textsc{m}) = y]$ for some label $y$, denoted by OMA($\hat{x},y$).
OMA($\hat{x},y$) is $\textit{True}$ if the condition holds, and otherwise $\textit{False}$. 
\end{Definition}



Specifically, we call a sample $\hat{x}$ has a \textbf{label difference} if $f(\hat{x}_\textsc{m}) \neq f(\hat{x})$ for some mask \textsc{m} in the covering mark set, i.e., there exists mutants not predicted with the prediction label of the sample.
%
Moreover,
we refer to $\hat{x}$ as a \textbf{consistent sample} if \text{OMA}$(\hat{x}, y_0) = \textit{True}$ (all mutants predict the true label of $\hat{x}$), regardless of whether $\hat{x}$ is correctly predicted; otherwise, it is an \textbf{inconsistent sample}. An \textbf{inconsistent mutant} of $\hat{x}$ is a mutant $\hat{x}_m$ not predicted with the true label ($f(\hat{x}_m) \neq y_0$).
In contrast, a \textbf{consistent mutant} is a mutant predicted with the true label ($f(\hat{x}_m) = y_0$).


\begin{figure*}[t]
\centering
\includegraphics[width=1\linewidth]
{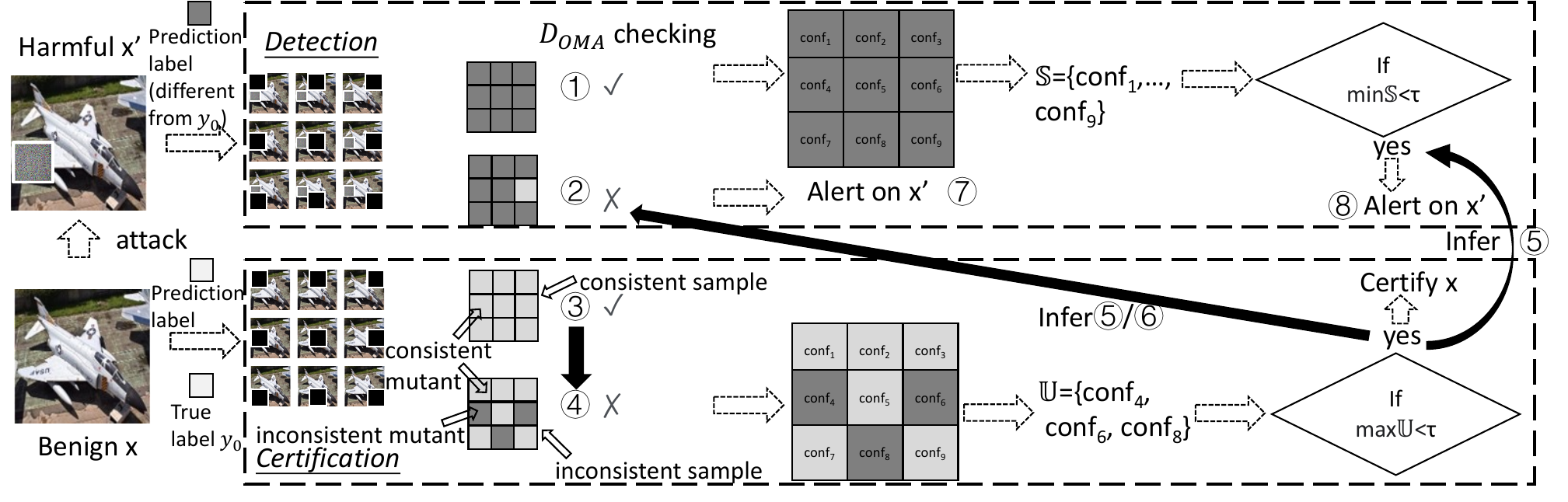}
\caption{
Overview of the design of HiCert on the interplay between the certification of benign samples and the detection of harmful samples.
}
\label{fig:overview}
\end{figure*}

\section{Existing Methods and Motivation}\label{sec:motiva}

\subsection{Categories of Existing Methods and Their Limitations}
Except for CrossCert \cite{zhou2024crosscert},
all existing certified detection defenders are masking-based defenders based on the $D_\text{OMA}$ condition,
which 
can be categorized into two kinds.


\subsubsection{C1: Detecting all the harm}
Minority Report (MR) \cite{mccoyd2020minority} (poor scalability), ScaleCert (SC) \cite{han2021scalecert} (not deterministic guarantee), and  PatchGaurd++ (PG++) \cite{xiang2021patchguard++}
aim to detect all the harmful samples $x'$ of certified benign samples $x$ with its true label $y_0$ (i.e., detecting $x'\in\mathbb{A}_\mathbb{P}(x), f(x')\neq y_0$), which is also adopted by this work.
%
PG++ is designed to \emph{only} warn an input sample if its prediction label differs from any mutant's prediction label with confidence above a threshold $\tau \in [0,1]$.
To achieve this, it certifies a benign sample $x$ with true label $y_0$
by requiring both \text{OMA}$(x, y_0) = \textit{True}$ and the confidence for every mutant of $x$ exceeding $\tau$.
(i.e., If  $[\forall \textsc{m} \in \mathbb{M}_\mathbb{P}, f(x_\textsc{m})=y_0\land f_{\textit{conf}}(x_\textsc{m})>\tau]$ holds, then $[\forall x'\in\mathbb{A}_\mathbb{P}({x}),[ f(x')\neq y_0]\implies[\exists \textsc{m} \in \mathbb{M}_\mathbb{P}, f(x'_\textsc{m})\neq f(x')\land f_{\textit{conf}}(x'_\textsc{m})>\tau]]$ holds \cite{xiang2021patchguard++}).
\textbf{PG++} \cite{xiang2021patchguard++} is formally defined as $\langle f, v, w \rangle$, where
\begin{itemize}
    \item $f$ is a classification model,
    \item $w(\hat{x}) := 
    [\exists \textsc{m} \in \mathbb{M}_\mathbb{P}, f(\hat{x}_\textsc{m})\neq f(\hat{x})\land f_{\textit{conf}}(\hat{x}_\textsc{m})>\tau]$%
    \footnote{
    Note that the condition $w(\hat{x})$ implies $[\text{OMA}(\hat{x}, f(\hat{x}))=\textit{False}]$.}, 
    \item
  $v(x) := 
  [\text{OMA}(x,y_0) = \small{\textit{True}} \land \forall \textsc{m} \in \mathbb{M}_\mathbb{P}, f_{\textit{conf}}(x_\textsc{m})>\tau]$,
  \end{itemize} 
where $\tau \in [0,1]$ is a constant.

If $\tau$ in PG++ 
is set to 0, then the clause $\forall \textsc{m} \in \mathbb{M}_\mathbb{P},f(x_\textsc{m})=y_0\land f_{\textit{conf}}(x_\textsc{m})>\tau$
will be degenerated into \text{OMA}$(x, y_0) = \textit{True}$, and PG++
will be reduced into $D_\text{OMA}$ defender.%
%

$\textbf{\textit{D}}_\text{\textbf OMA}$ is defined as 
$\langle f, v, w \rangle$, where 
\begin{itemize}
\item $f$ is a classification model,
\item $w(\hat{x}) := [\text{OMA}(\hat{x}, f(\hat{x}))=\textit{False}]$, and
\item $v(x) := [\text{OMA}(x,y_0)=\textit{True}]$.
\end{itemize}

Fig.~\ref{fig:oma} illustrates the $D_\text{OMA}$ defender. 
For each input sample, $D_\text{OMA}$ generates a set of its mutants, one for each mask in the covering mask set $\mathbb{M}_\mathbb{P}$ (nine mutants are shown), no matter for certification or for warning.
For a benign sample that is consistent (depicted at the endpoint \textcircled{3}),
all its harmful versions should 
exhibit label differences (depicted at the endpoint \textcircled{2}) and finally being warned at endpoint \textcircled{7}.
%
$D_\text{OMA}$ certifies the benign sample at the endpoint \textcircled{3},
depicted by a single certification-warning path \textcircled{3}-\textcircled{2}-\textcircled{7}
applicable to all consistent samples.
However, $D_\text{OMA}$
fails to certify any inconsistent samples (depicted at the endpoint \textcircled{4}, where some mutants are predicted with a label different from the true label, i.e., an inconsistent mutant) since the label difference may disappear on their harmful version (depicted at the endpoint \textcircled{1}), leaving security risks.

\textbf{Theoretical Limitation:}
Although their certification prerequisites may differ,
{all} defenders in C1 (MR, SC, PG++, and $D_\text{OMA}$) 
commonly require certified samples to be consistent samples (i.e., $\text{OMA}(x, y_0)=\textit{True}$).
They are
theoretically unable to certify \emph{any inconsistent samples}.

\subsubsection{C2: Detecting all the changes}
\label{sec:Detecting all the changes}
ViP \cite{li2022vip} and PatchCensor (PC) \cite{patchcensor} 
present the same variant of the certification theorem of $\textit{D}_\text{OMA}$
as theirs --- If $ [\text{OMA}(x,f(x))=\textit{True}]$ holds, then  $[\forall x'\in\mathbb{A}_\mathbb{P}({x}),[f(x')\neq f(x)\implies\text{OMA}(x', f(x'))=\textit{False}]]$ holds,
i.e., define $v(x) = [\text{OMA}(x,f(x))=\textit{True}]$ instead of $[\text{OMA}(x,y_0)=\textit{True}]$ in $\textit{D}_\text{OMA}$'s certification theorem (mutants are consistently predicted with the same prediction label of the sample) and keep $w(\hat{x}) = [\text{OMA}(\hat{x},f(\hat{x}))=\textit{False}]$.
However, 
the condition 
$f(x')\neq f(x)$ in 
this variant can only ensure that patched samples with prediction labels different from the corresponding ``certified'' benign samples are detected.
We can observe from the above a common strategy for {certification} and warning (referred to as ${\textit{D}}_\textbf{\textit{OMA}}$ \textbf{checking}) adopted by {defenders in C1 and C2}:
The certification function $v()$ always requires that a sample satisfies an $D_\text{OMA}$ condition to be certified; any warning raised by the warning function $w()$ must
associate with a violation of an $D_\text{OMA}$ condition.

\textbf{Theoretical Limitation:}
Defenders in C2 (ViP and PC) are insufficient to
ensure the detection of all harmful samples of a ``certified'' benign sample if a benign sample is incorrectly predicted (i.e., $f(x)\neq y_0$): an attacker can retain this incorrect prediction label but produce harmful samples.
On the other hand, if the benign samples are correctly predicted (i.e., $f(x)= y_0$), defenders in C2 will {incur} 
the same limitation of C1 in consistent samples (i.e., $[\text{OMA}(x,f(x))=\textit{True}]\Leftrightarrow [\text{OMA}(x,y_0)=\textit{True}]$). 

\subsection{When Will Certification Fail?
}
For example, to apply such a defender in autonomous driving scenarios \cite{patchcensor}, a focus is whether the harmful samples (e.g., an adversarially altered traffic sign that originally indicated a low speed but now classified with a high speed) can pass through the detection guard of the traffic sign recognition systems to downstream processes (e.g., increasing the vehicle's travel speed), which depends on whether a defender can ensure the detection of all harmful versions of a benign sample $x$.
\textbf{Inconsistent samples:}
As described in Fig. \ref{fig:moti-attack},
the inconsistent sample $x$, a ``Speed Limit 50 km/h'' sign from GTSRB, is correctly classified. 
However, many of its mutants are predicted with ``Speed Limit 50 km/h'', while the remaining mutants are predicted with ``Speed Limit 100 km/h'', i.e., $x$ cannot be certified by defenders in C1 or C2%
.
At the same time, the presented harmful sample $x'$ in Fig.~\ref{fig:moti-attack} indeed has no label difference.
A previous work
\cite{patchcensor} explained that 
PatchCensor incurred this problem
\textit{``because the [Deep Learning] DL models are imperfect''}, which we agreed.
Still, 
leaving the whole issue unaddressed by a defender is unsatisfactory.
\textbf{Incorrectly predicted samples:}
We performed a case study (see Section \ref{app:motivate_case_study}) and found that only 1 sample out of all 8751 incorrectly predicted samples in ImageNet is a consistent sample.
$D_\text{OMA}$ (defenders in C1) can only certify this single sample in all incorrectly predicted samples, and defenders in C2 can never ensure the detection of all their harmful samples.

\subsection{Certification Consideration}
As a base model $f$ is imperfect, 
many benign samples, especially those incorrectly predicted, have inconsistent mutants.
Without further information, a warning function $w(.)$ that uses label differences (i.e., $w(\hat{x}):= \neg OMA(\hat{x}, f(\hat{x}))$) is theoretically impossible to additionally detect a harmful sample $x'$ that satisfies
OMA$({x'}, f({x}'))$ (i.e., tries to warn more harmful samples to also cover more benign samples for certification)
but not detect benign samples 
that satisfies  
OMA$({x}, f({x}))$
(i.e., now $w(\hat{x})$ is always \textit{True}, which has to alert on \emph{all} input samples).

\section{{Our Proposal: HiCert}}\label{sec:proposal}

\subsection{Overview}\label{sec:design_of_hicert}
We formulate HiCert $D = \langle f, w, v\rangle$ as follows:
\begin{itemize} 
 \item $f$ is a classification model,
 
 \item $w(\hat{x}):=[\{\hat{x}_\textsc{m}\mid \textsc{m}\in\mathbb{M}_\mathbb{P},f(\hat{x}_\textsc{m})\neq f(\hat{x})\}\neq \emptyset] 
 \lor 
 [\min{\{f_{\textit{conf}}(\hat{x}_\textsc{m})\mid \textsc{m}\in\mathbb{M}_\mathbb{P},f(\hat{x}_\textsc{m})=f(\hat{x})\}}<\tau]$, and
 
 \item $v(x):= 
 [\max{\{f_{\textit{conf}}(x_\textsc{m})\mid \textsc{m}\in\mathbb{M}_\mathbb{P},f(x_\textsc{m})\neq y_0\}}<\tau]$,
 \end{itemize}
where $\max \emptyset=-\infty, \min \emptyset=\infty$, and   $\tau \in [0,1]$ is a constant.


Specifically, HiCert is reduced to $D_\text{OMA}$ when $\tau=0$, and reduced to a trivial detection defender when $\tau=1$ (see Section \ref{app:special_case_of_HiCert} for details).

Given an input $\hat{x}$, for certification,
HiCert generates its mutant $\hat{x}_\textsc{m}$ iteratively and checks whether its prediction label is the same as its true label (i.e., $f(\hat{x}_\textsc{m})\neq y_0$); if that is not the case, it checks whether its confidence is lower than the threshold $\tau$ (i.e., $f_{conf}(\hat{x}_\textsc{m})<\tau$). 
A certificate is assigned if all mutants meet one of these two conditions (implementation of the certification function $v()$).
Similarly, for detection,
HiCert generates its mutant $\hat{x}_\textsc{m}$ iteratively and checks whether its prediction label is the same as the prediction label of $x$ (i.e., $f(\hat{x}_\textsc{m})\neq f(\hat{x})$); if that is the case, it checks whether its confidence is lower than the threshold $\tau$ (i.e., $f_{conf}(\hat{x}_\textsc{m})<\tau$).
A warning is raised if any mutant fails to satisfy either condition (implementation of the warning function $w()$).
We also illustrate the flowcharts of the implementation of HiCert in the Section~\ref{sec:flow}.

The design of HiCert on the interplay between the certification of benign samples and the detection of harmful samples is shown in Fig.~\ref{fig:overview}.
A benign sample $x$ with its true label $y_0$ subject to certification is either consistent  (depicted at the endpoint \textcircled{3}) or inconsistent 
(depicted at the end point \textcircled{4}), satisfying OMA$(x,y_0)=\textit{True}$ and OMA$(x,y_0)=False$, respectively. 
Starting from them, there are three certification-warning paths in HiCert with the deterministic guarantee {on different combinations of chosen mutants of benign samples and derived mutants of harmful samples by attackers}:
\textcircled{4}-\textcircled{6}-\textcircled{2}-\textcircled{7}, \textcircled{4}-\textcircled{5}-\textcircled{2}-\textcircled{7},
and \textcircled{4}-\textcircled{5}-\textcircled{1}-\textcircled{8}.
\footnote{Note that the physical meanings of the endpoints \textcircled{5} and \textcircled{6} are that the patch is placed within the mask that generates an inconsistent mutant and a consistent mutant, respectively.
From the endpoint \textcircled{4}, the endpoint \textcircled{1} in the warning process can only be reached by endpoint \textcircled{5}, proven by Thm.~\ref{thm:infeasible}, and the endpoint \textcircled{2} can be reached by endpoint \textcircled{5} or endpoint \textcircled{6}.
}
HiCert first generates one mutant for each mask $\textsc{m}$ in the covering mask set $\mathbb{M}_\mathbb{P}$ (nine mutants are shown).
For an inconsistent sample $x$ at the endpoint \textcircled{4},
HiCert checks if the confidence of every inconsistent mutant of $x$
is below a given threshold $\tau$ (formulated as $[\max \mathbb{U} < \tau]$ where 
$\mathbb{U}={\{f_{\textit{conf}}(x_\textsc{m})\mid \textsc{m}\in\mathbb{M}_\mathbb{P},f(x_\textsc{m})\neq y_0\}}$, equivalent to $v()$).
If true, it ensures that all harmful samples of $x$ are warned by Thm.~\ref{thm:Inconsistent-Max-Min}: by detecting a label difference (the first expression in $w()$) on a harmful sample $x'$ if choosing Path  \textcircled{4}-\textcircled{6}-\textcircled{2}-\textcircled{7} or Path \textcircled{4}-\textcircled{5}-\textcircled{2}-\textcircled{7}, or 
by detecting the confidence of some inconsistent mutant of $x'$ below $\tau$ (formulated as $[\min \mathbb{S} < \tau]$ where 
$\mathbb{S}=\{f_{\textit{conf}}(\hat{x}_\textsc{m})\mid \textsc{m}\in\mathbb{M}_\mathbb{P},f(\hat{x}_\textsc{m})=f(\hat{x})\}$, equivalent to the second expression in $w()$)
if choosing Path \textcircled{4}-\textcircled{5}-\textcircled{1}-\textcircled{8}.
(Note that attackers are responsible for choosing endpoint \textcircled{5} or endpoint \textcircled{6}.)
Specifically,
at the endpoint \textcircled{2} on the warning side, if a difference in the label between a specific mutant and $x'$ appears (where each mutant is generated from each mask in the above set $\mathbb{M}_\mathbb{P}$), HiCert alerts on $x'$ (reaching \textcircled{7}).
If there is no difference in the label between this mutant-sample pair,
HiCert goes to the endpoint \textcircled{1}.
At endpoint \textcircled{1}, 
HiCert checks if the confidence of this specific mutant is below the threshold $\tau$ and alerts on $x'$ (reaching \textcircled{8}) if true.
It repeats this checking process over the two subpaths \textcircled{1}-\textcircled{8} and \textcircled{2}-\textcircled{7} for each mutant of $x'$.
The physical meaning of Path \textcircled{4}-\textcircled{5}-\textcircled{2}-\textcircled{7} is that the attacker is unable (too weak) to make all mutants predicted with the prediction label of the harmful sample even if it is theoretically possible, and in contrast, the physical meaning of Path \textcircled{4}-\textcircled{5}-\textcircled{1}-\textcircled{8} is that the attacker indeed makes all mutants predicted with the prediction label of the harmful sample, 
which is detected by the low confidence criterion.
Note that consistent samples are a special case of inconsistent samples (endpoint \textcircled{3} is a special case of endpoint \textcircled{4}) with $\mathbb{U}=\emptyset$, whose harmful samples should follow the Path \textcircled{4}-\textcircled{6}-\textcircled{2}-\textcircled{7} to be detected, proven by Thm.~\ref{thm:infeasible}. 
On the other hand, $D_\text{OMA}$
only has Path \textcircled{3}-\textcircled{2}-\textcircled{7} (not explicitly shown on the figure but implicitly included in Path \textcircled{4}-\textcircled{6}-\textcircled{2}-\textcircled{7}), which fail to certify any inconsistent samples in \textcircled{4} and fail to warn any harmful samples in \textcircled{1}.
The same evasion of harmful samples will also occur in PatchCensor, ViP, and PG++.




%

\subsection{Inconsistent Mutants + No Label Difference: Key Problem} \label{sec:why_fail}
This section presents our effort to come up with a direction to address the certification problem with inconsistent samples.

Our first insight (\textbf{Insight A}) is non-obvious at first sign: Certifying a consistent benign sample is a special case of certifying an inconsistent benign sample.

An inconsistent sample has two sets of mutants in general: one set for consistent mutants and the other set for inconsistent mutants, and
a special case is that this set for inconsistent mutants is empty.
On the one hand, if placing a patch on a consistent mutant, proven by Thm. \ref{thm:infeasible}, it cannot create a sample without a label difference if it is harmful;
On the other hand, if the attacker
places a patch on an inconsistent mutant (if it exists), it can choose to
produce a harmful patched sample 
that
exhibits no label difference or still keeps exhibiting a label difference.
In other words, consistent samples are a special case of inconsistent samples, placing a patch on consistent mutants is a special case of placing on inconsistent mutants, and the label difference is also a special case for detecting harmful samples of inconsistent samples.
Thm.~\ref{thm:infeasible} formally proves the infeasibility of placing a patch on consistent mutants of a sample without exhibiting a label difference to produce a harmful sample, which is applicable to all samples regardless of whether the sample is incorrectly predicted and whether it is inconsistent.
Intuitively, by relying on consistent mutants rather than consistent samples, Thm.~\ref{thm:infeasible} can eliminate a part of the attack cases on inconsistent samples.

\begin{figure}[!t]
\centering
\includegraphics[width=\linewidth]{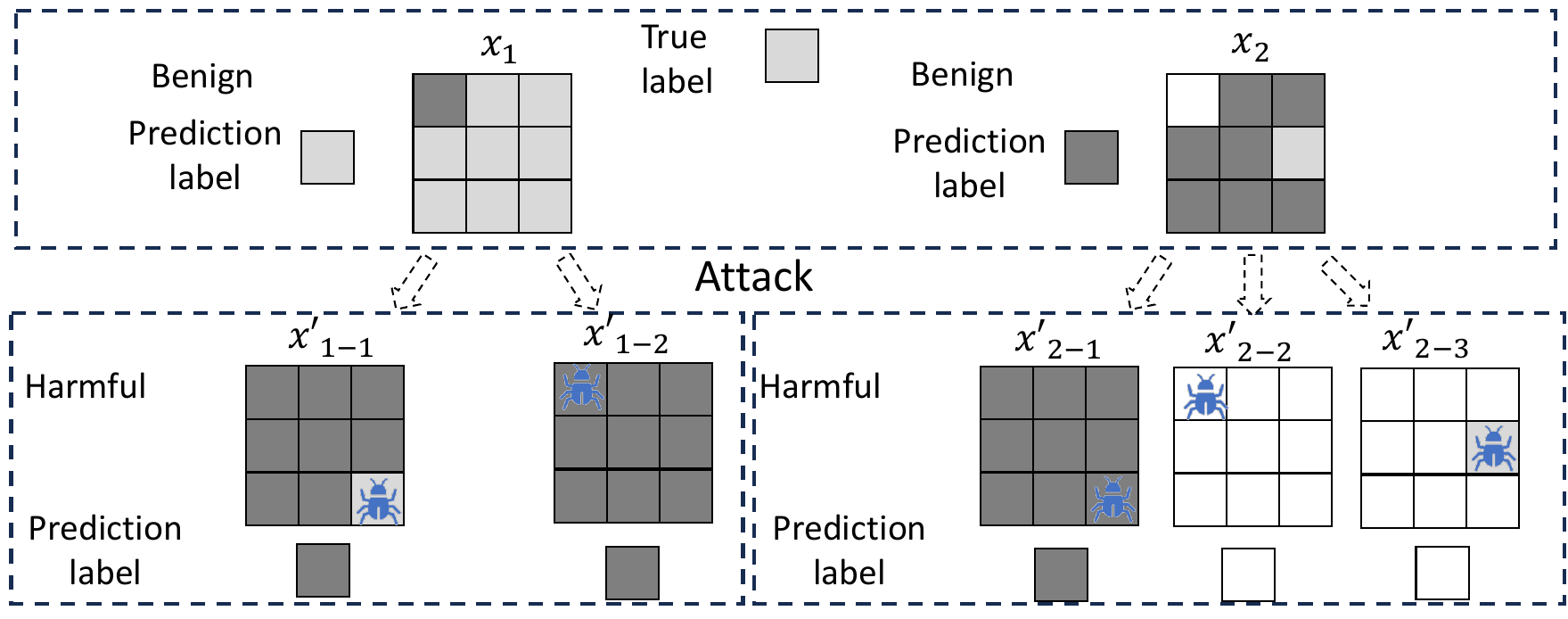}

\caption{
Illustration of (all) five possible attack cases of inconsistent benign samples with the best effort of the attacker.
9 mutants (squares) are generated for each sample, and different colors represent different prediction labels.
}

\label{fig:why_fail}
\end{figure}

\begin{thm}[Consistent mutants are infeasible places for attackers]
\label{thm:infeasible}
If the patch region is covered by a mask whose corresponding mutant's label is the same as the true label, it is infeasible for harmful samples to show no label difference. (i.e.,
if the condition $[\exists \textsc{m}_\textsc{p} \in\mathbb{M}_\mathbb{P},
\textsc{m}_\textsc{p}\odot{\textsc{p}}=\textsc{p}
\land 
f(x_{\textsc{m}_\textsc{p}})=y_0]
$
holds, the condition 
$[\forall x'\in\{x'\mid{x}'=(\textsc{J}-\textsc{p})\odot {x}+\textsc{p}\odot {x}'\},
[f(x')\neq y_0]\implies
[\exists \textsc{m} \in \mathbb{M}_\mathbb{P}, 
f({x}'_{\textsc{m}})  \neq f(x')]]$
holds.)

\end{thm}
The formal proof is in Section \ref{app:proof}. 
Intuitively, with its mask covering the patch, a consistent mutant will always be predicted to the true label, no matter how the attacker attacks. 
With Thm.~\ref{thm:infeasible}, we can turn the focus 
to the focus on inconsistent mutants on a finer granularity.

Fig.~\ref{fig:why_fail} illustrates possible attack cases on inconsistent samples to result in harmful samples, where $x_1$ and $x_2$ are correctly and incorrectly predicted, respectively.
Among all harmful samples shown, $x'_{1-1}$ and $x'_{2-3}$ are the harmful samples generated by placing the patch on the consistent mutant of their benign counterparts ($x_1$ and $x_2$). Based on Thm.~\ref{thm:infeasible}, the labels of these mutants should be different from 
the prediction labels of $x'_{1-1}$ and $x'_{2-3}$ (i.e., the label difference appears), respectively.
However, attackers may alternatively modify $x_1$ and $x_2$ to create other harmful samples by placing patches on their inconsistent mutants, as shown in  Fig.~\ref{fig:why_fail}, which may create harmful samples without any label difference (e.g., $x'_{1-2}$, $x'_{2-1}$ and  $x'_{2-2}$).
\subsection{Low Confidence Across Inconsistent Mutants: Key Solution}
\label{sec:invariant_solution}

Following Insight A stated in Section~\ref{sec:why_fail}, 
we can focus on addressing the certification problem of inconsistent samples by inconsistent mutants while not forgetting those consistent.



Our second sight (\textbf{Insight B}) is:
Creating harmful samples by patching a typical inconsistent sample inevitably leaves traces of prediction-based evidence in every harmful sample, where we identify an effective trace of evidence for certification: \emph{either a label difference is evident, or a low-confidence mutant is retained}.


We observe from our experiment that
most samples with inconsistent mutants of benign samples in the ImageNet dataset have relatively low confidence across all their respective inconsistent mutants empirically
(see Histogram \textcircled{1} in Fig.~\ref{fig:confidence_min_max} for ImageNet+MAE),
where we refer to it as the low confidence property over the set of inconsistent mutants of a sample to ease our reference.
Efforts to mitigate attacks on them should have a higher priority.
Low confidence indicates that the classifier in question has weak support on any labels for such mutants.
Intuitively, the classification results of such mutants (especially those of an incorrectly predicted sample, which is in a messier state as a whole) are more unreliable and easier to manipulate by attackers, which may empirically lead to the disappearance of the label difference in harmful samples.%
\footnote{See the experimental results in Fig.~\ref{fig:real_attack} for the actual attack on $D_\text{OMA}$.
With a patch size of 32 pixels, 
on ImageNet/CIFAR100 (less accurate), almost all inconsistent samples (those samples cannot be certified by $D_\text{OMA}$) can be successfully attacked (i.e., the label difference disappears in harmful samples);
on GTSRB (highly accurate), more than a quarter of inconsistent samples can be successfully attacked.
}
%
However, since the effort of attackers appearing in a harmful sample that results in no label difference will be removed if the specific mask for an inconsistent mutant covers the patch in the harmful sample, this low confidence property of the inconsistent mutants for the sample under attack will always be unveiled by the harmful sample after the patch is covered; in other words, \emph{the reason why a sample with such mutants can be easily attacked successfully will become a trace of harm creation retained in every resulting harmful sample as evidence}.
Detecting a necessary condition of the low confidence property on harmful samples can mitigate the threat.
If the effort of attackers on inconsistent mutants is not strong enough to make a harmful sample not exhibit a label difference, it leaves another trace of evidence.
On the other hand,
despite that most consistent mutants of consistent benign samples are relatively high in confidence (a heuristic used by PG++ and see Histogram \textcircled{4} in Fig. \ref{fig:confidence_min_max}),
by Thm.~\ref{thm:infeasible}, a label difference must appear in every resulting harmful sample 
as long as the attacker attacks a benign sample through its consistent mutants,
which is applicable for both consistent samples and those inconsistent samples with consistent mutants.

Therefore, if we have verified that \emph{all} inconsistent mutants in a sample are relatively low in confidence (which is typical as observed in Histogram \textcircled{1},
modeled as \emph{the confidence of every such mutant below a threshold $\tau$} to define the low confidence property mentioned above), any successful attack on a sample with the low confidence property can \emph{only} lead to one of the two consequences:
Either to produce a harmful sample with a label difference or to produce a harmful sample retaining a mutant with relatively low confidence.

If we design HiCert to warn in both cases, then it should be able to detect all harmful samples the attacker produces from the samples with the low confidence property.
Thm.~\ref{thm:Inconsistent-Max-Min} captures this insight (on top of Insight A and Thm. \ref{thm:infeasible}) to form the certification theorem of HiCert.
Intuitively, Thm.~\ref{thm:Inconsistent-Max-Min} uses low confidence as the trace to strengthen Thm.~\ref{thm:infeasible},
further ensure the infeasibility of those attack cases Thm.~\ref{thm:infeasible} cannot eliminate.


\begin{thm}
[HiCert Certification]
\label{thm:Inconsistent-Max-Min}
If the maximum confidence of inconsistent mutants of a benign sample $x$ is below a threshold $\tau$,
  each harmful sample $x'$ either incurs a label difference or has mutant(s) with minimum confidence below $\tau$ that are predicted with a label the same as $x'$
 --- if the condition 
$[\max{\{f_{\textit{conf}}({x}_\textsc{m})\mid \textsc{m}\in\mathbb{M}_\mathbb{P},f({x}_\textsc{m})\neq y_0\}}<\tau]$~holds, the condition $[\forall x'\in\mathbb{A}_\mathbb{P}({x}),[f(x')\neq y_0]
\implies 
[\{{x}'_\textsc{m}\mid  \textsc{m}\in\mathbb{M}_\mathbb{P},f({x}'_\textsc{m})\neq f({x'})\}\neq \emptyset] 
 \lor 
 [\min{\{f_{\textit{conf}}({x}'_\textsc{m})}$ ${\mid \textsc{m}\in\mathbb{M}_\mathbb{P},f({x}'_\textsc{m})=f({x'})\}}<\tau]]$ holds, 
 which is
 $v(x)\implies [\forall x'\in\mathbb{A}_\mathbb{P}({x}), f(x')\neq y_0 \implies w(x')]$ in HiCert.
\end{thm}
The formal proof is in Section \ref{app:proof}.
Intuitively, with its mask covering the patch, a low confidence inconsistent mutant will always be low confidence, no matter how the attacker attacks. If all inconsistent mutants are in low confidence, then the low confidence mutant must exist even under attacks.
%
%
Thm.~\ref{thm:Inconsistent-Max-Min} has the following three special cases SC1--SC3 corresponding to the three purposes outlined in Fig.~\ref{fig:overall-effect}.
Note that $\max \emptyset=-\infty$ and $\min \emptyset=\infty$.

\begin{description}
[style=nextline]
\item[(SC1) Certifying inconsistent samples:]
Attacks with the patch on inconsistent mutants of inconsistent benign samples, which are identified in the last subsection as the focal problem in certifying inconsistent samples, are tackled by  $ [\min{\{f_{\textit{conf}}({x}'_\textsc{m})}$ ${\mid \textsc{m}\in\mathbb{M}_\mathbb{P},f({x}'_\textsc{m})=f({x'})\}}<\tau]]$ (Path \textcircled{4}-\textcircled{5}-\textcircled{1}-\textcircled{8}) for label difference omission and 
$[\{{x}'_\textsc{m}\mid  \textsc{m}\in\mathbb{M}_\mathbb{P},f({x}'_\textsc{m})\neq f({x'})\}\neq \emptyset]$ (Path \textcircled{4}-\textcircled{5}-\textcircled{2}-\textcircled{7}) for observable label differences; 
attacks with the patch on the other (i.e., consistent) mutants are tackled by $[\{{x}'_\textsc{m}\mid  \textsc{m}\in\mathbb{M}_\mathbb{P},f({x}'_\textsc{m})\neq f({x'})\}\neq \emptyset]$ (Path \textcircled{4}-\textcircled{6}-\textcircled{2}-\textcircled{7}).


    \item[(SC2) Certifying consistent samples:] 
Consistent benign samples are a special case to make the antecedent of the implication relation in Thm.~\ref{thm:Inconsistent-Max-Min} hold (which is $[\max \emptyset < \tau$]), 
where the inconsistent sample degenerates into a consistent sample with an empty set of inconsistent mutants.
Thus, SC2 is actually \emph{a special case} of SC1, but not vice versa, and consistent samples are tackled by $[\{{x}'_\textsc{m}\mid \textsc{m}\in\mathbb{M}_\mathbb{P},f({x}'_\textsc{m})\neq f({x'})\}\neq \emptyset]$ (Path \textcircled{4}-\textcircled{6}-\textcircled{2}-\textcircled{7}).

\item[(SC3) Certifying and detecting incorrectly predicted samples:]
Thm.~\ref{thm:Inconsistent-Max-Min} is applicable to certify incorrectly predicted samples, which ensures the detection of all of its patched versions and the incorrectly predicted sample itself if the sample is certified. Since such a sample may have both consistent and inconsistent mutants in general, certifying it needs both conditions and all three paths to support.
Alerting on a certified incorrectly predicted sample is simple by equating $x$ to $x'$ in the relation.
\end{description}

By designing a warning function that follows Thm.~\ref{thm:Inconsistent-Max-Min},
HiCert places attackers in a \emph{dilemma} if they attempt to create a harmful sample $x'$ of any benign sample $x$ with $v(x)=\textit{True}$ and aim to make HiCert silent on their created harmful samples (see Section \ref{app:proof} for more details).


We have conducted a case study (See Section \ref{app:case_study_on_design} for experimental setting) to show the advancement achieved by HiCert, as shown in Fig.~\ref{fig:confidence_min_max}.
{Histograms \textcircled{1} and \textcircled{2} represent HiCert.}
Histogram \textcircled{1} shows the number of inconsistent samples with the largest confidence  (x-axis) among all inconsistent mutants of the same sample, while Histogram \textcircled{2} shows the number of consistent samples with the smallest confidence (x-axis) among all (consistent) mutants of the same sample. 
{Histograms \textcircled{3} and \textcircled{4} represent  PG++}. 
The bars for samples certified by the corresponding defenders (see the labels for the $y$-axis) are displayed in a solid color; otherwise, they are semi-transparent, where the confidence threshold $\tau$ is set to 0.8 for illustration purposes.
Irrespective of any threshold, the maximum confidence among the confidences of all inconsistent mutants of the same samples for all samples spreads over the range [0,1] without a sharp peak, and their central tendency of these data points shown in the histograms is far from the confidence of 1.
HiCert is able to certify a significant number of inconsistent samples, as shown in sub-figure \textcircled{1}.
As a consequence of extending the scope of certified detection to certify inconsistent samples, it also certifies all consistent samples 
as shown in sub-figure \textcircled{2}.
As a comparison,  PG++ cannot certify any inconsistent samples (sub-figure \textcircled{3}) and can only certify a slice of all consistent samples (sub-figure \textcircled{4}). 
Similarly, we also conducted an ablation study by respectively flipping the inequality symbol for HiCert and PG++ (see Section \ref{app:case_study_on_design}).
Both of them are ineffective in certifying inconsistent samples, which shows that modifying the defenders to have the ability to certify both consistent benign samples and inconsistent benign samples (even in part) effectively is nontrivial.

In short, HiCert tackles the certification problem and formulates a common solution for a suite of closely related problems scenarios that previous works cannot handle them all.
The solution is finer in granularity as well.
HiCert achieves the same time complexity as $D_\text{OMA}$ and is sound but incomplete like $D_\text{OMA}$ (also PatchCensor \cite{patchcensor}, ViP \cite{li2022vip}, and PG++ \cite{xiang2021patchguard++}), which is discussed in detail in Section \ref{app:discussion-on-HiCert-design}.


\begin{figure}[tb] 
\centering
\includegraphics[width=0.8\linewidth]{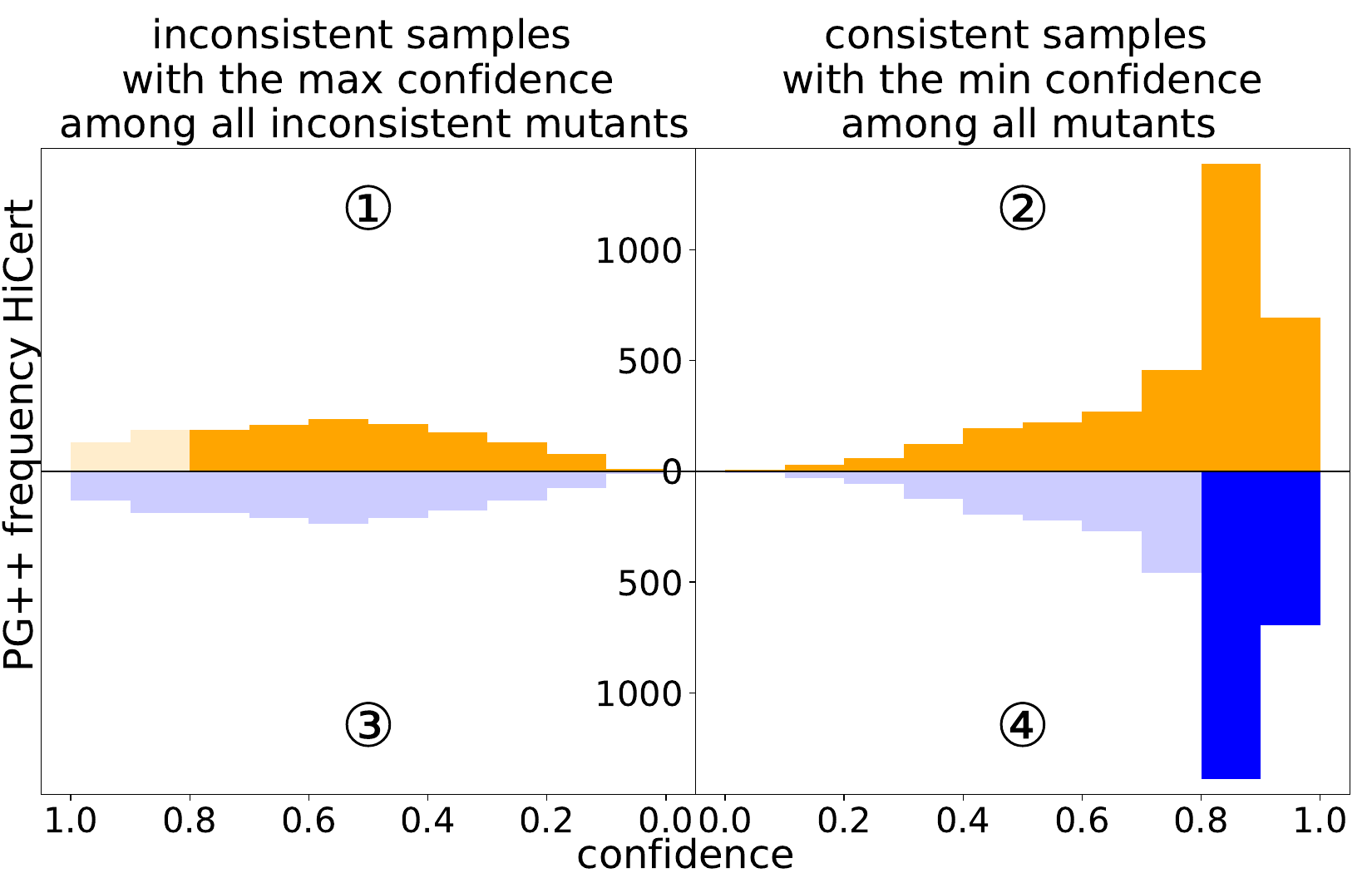}
\caption{
The plots show the maximum and minimum confidences among those mutants of the same sample for those samples out of all samples, as stated in the column headings.
}
\label{fig:confidence_min_max}
\end{figure}

\section{Evaluation}\label{sec:eva}
Our implementation package of HiCert can be found in \cite{HiCert_github}.




\subsection{Research Questions}
We aim to answer the following research questions:
\begin{itemize}
\item[RQ1] How does {HiCert} perform compared to state-of-the-art certified detection defenders against patch attacks?
\item[RQ2] To what extent is {HiCert} effective in defending against a real adversarial patch attack?
\item[RQ3] 
How does {HiCert} perform against stronger attackers in terms of patch sizes?
\end{itemize}
\begin{table}[]
\caption{Clean accuracy of base models for different datasets
}\label{tab:clean_acc}\centering
\begin{tabular}{|c|c|c|c|}
\hline
    & ImageNet & CIFAR100 & GTSRB \\ \hline
MAE (default) & 82.5     & 90.2     & 97.5  \\ \hline
ViT & 81.7     & 92.3     & 98.8  \\ \hline
RN  & 80.4     & 88.8     & 97.9  \\ \hline
\end{tabular}
\end{table}

\newcolumntype{C}[1]{>{\centering\arraybackslash}p{#1}}
\newcolumntype{M}[1]{>{\centering\arraybackslash}m{#1}}

\begin{table*}[]\centering
\caption{Metrics of a defender $D = \langle f, w, v\rangle$ on a dataset $\mathbb{S}$ with each sample $x$ and its true label $y_0$.
Except for $r_{\textit{fa}}$ and $r_{\textit{fs}}$, 
higher values for all other metrics indicate better quality
}\label{tab:metrics}

\begin{tabular}{|M{0.11\linewidth}|M{0.38\linewidth}|M{0.44\linewidth}|}
\hline
Metrics                                    & Formulation & Description \\ \hline
Clean accuracy                             & $acc_{\textit{clean}}=\frac{\mid\{{x}\in\mathbb{S}\mid f({x})=y_0\}\mid}{\mid\mathbb{S}\mid}$        & to evaluate the inherent
classification capability of the base model        \\ \hline \rowcolor{gray!10}\multicolumn{3}{c}{Certification Metrics} \\ \hline
Certified   accuracy                       & $acc_{\textit{cert}}=\frac{\mid\{{x}\in\mathbb{S}\mid f({x})=y_0\land v({x})=\textit{True}\}\mid}{\mid\mathbb{S}\mid}$         & to evaluate the certification ability on
correctly predicted samples        \\ \hline
Certified   ratio                          & $r_{cert}=\frac{\mid\{{x}\in\mathbb{S}\mid v({x})=\textit{True}\}\mid}{\mid\mathbb{S}\mid}$        & to evaluate the certification ability on
all samples        \\ \hline
Certified   ratio for inconsistent samples & $r_{\textit{cert}_{\textit{inc}}}=\frac{\mid\{{x}\in\mathbb{S}\mid v({x})=\textit{True}\land \text{OMA}({x},y_0)=\textit{False}\}\mid}{\mid\{{x}\in\mathbb{S}\mid \text{OMA}({x},y_0)=\textit{False}\}\mid}$        & to evaluate the certification ability on
 inconsistent samples        \\ \hline \rowcolor{gray!10}\multicolumn{3}{c}{Secondary Metrics} \\  \hline
Silent   accuracy                          & $
acc_{\neg w}=\frac{\mid\{{x}\in\mathbb{S}\mid w({x})=\textit{False}\land f({x})=y_0\}\mid}{\mid\{{x}\in\mathbb{S}\mid w({x})=\textit{False}\}\mid}$        & the accuracy on the set of benign samples without warnings triggered (for silent path)       \\ \hline
False alert   ratio                        & $r_{\textit{fa}}=\frac{\mid\{{x}\in\mathbb{S}\mid w({x})=\textit{True}\land f({x})=y_0\}\mid}{\mid\{{x}\in\mathbb{S}\mid f({x})=y_0\}\mid}$         & the fraction of correctly predicted samples for which a defender returns a warning alert         \\ \hline
False silent   ratio                       & $r_{\textit{fs}}=\frac{\mid\{{x}\in\mathbb{S}\mid w({x})=\textit{False}\land f({x})\neq y_0\}\mid}{\mid\{{x}\in\mathbb{S}\mid f({x})\neq y_0\}\mid}$        & the fraction of incorrectly predicted samples for which we do not return an alert        \\ \hline
Defense   success ratio                    & $r_{suc}=\frac{|\{{x}\in\mathbb{S}_{sub}\mid \forall {x}'\in \mathbb{A}^{act}_{\mathbb{P}}({x}), f({x}')\neq y_0\implies w({x}')=\textit{True}\}|}{|\{\mathbb{S}_{sub}\}|}$        &  the proportion of benign samples for which all harmful samples generated by an attacker tool are detected by the defender, where $\mathbb{S}_{sub}$ is a subset of $\mathbb{S}$ used by an actual attacker tool as seed input, $\mathbb{A}^{act}_{\mathbb{P}}({x})$ is a subset of $\mathbb{A}_{\mathbb{P}}({x})$ generated by the actual attacker tool      \\ \hline
\end{tabular}
\end{table*}
\subsection{Experimental Setup}

We adopt ImageNet \cite{deng2009imagenet}, CIFAR100 \cite{krizhevsky2009learning}, and GTSRB \cite{Stallkamp-IJCNN-2011} as our datasets.
We adopt MAE \cite{he2022masked}, Vision Transformer (ViT) \cite{dosovitskiy2021an}, and ResNet (RN) \cite{he2016deep} as the architectures of the base models of defenders (see Table~\ref{tab:clean_acc} for their clean accuracy). We also use the model-agnostic pixel-level strategy to generate the covering mask set following PatchCleanser \cite{xiang2022patchcleanser} and CrossCert \cite{zhou2024crosscert}.
We compare top-performing certified detection defenders implemented in our infrastructure to HiCert (\textbf{HC}):  $\textbf{\textit{D}}_\textit{\text{OMA}}$ (ViP/PatchCensor) and \textbf{PG++} \cite{xiang2021patchguard++}. 
With the same base model and the same masking strategy,
\textbf{ViP} \cite{li2022vip} and PatchCensor (\textbf{PC}) \cite{patchcensor} must share the same certified accuracy and clean accuracy with $D_\text{OMA}$ since each sample $x$ counted by these two metrics satisfies $f(x)= y_0$, and then the condition OMA$(x, f(x)) \land f(x)= y_0$ for ViP and PC is equivalent to the condition OMA$(x, y_0)$ for $D_\text{OMA}$ in certification functions.
We further compare HiCert with more state-of-the-art certified detection defenders, and mark them with the symbol $\star$: ScaleCert (\textbf{SC$_\star$}) \cite{han2021scalecert}, PatchGaurd++ (\textbf{PG++$_\star$}) \cite{xiang2021patchguard++}, Adapted Minority Reports (\textbf{MR+$_\star$}) \cite{patchcensor}, PatchCensor (\textbf{PC$_\star$}) \cite{patchcensor}, ViP (\textbf{ViP$_\star$}) \cite{li2022vip}, and CrossCert (\textbf{CC$_\star$}) \cite{zhou2024crosscert} based on the results reported in the literature. 
%
Our metrics are summarized in Table~\ref{tab:metrics}, and we also collect each combination in certified detection (see Table~\ref{tab:case}) for analysis case by case.

See Section \ref{app:setup} for details of 
datasets, baselines, metrics, and experimental setup for RQs.


\subsection{Experimental Results and Data Analysis}




\subsubsection{Answering RQ1} 
In this section, we compare HC with peer defenders regarding their certification ability.

\paragraph{\textbf{Overall Comparison on Certified Accuracy}} 
Table~\ref{tab:main_eva_results} summarizes the overall results 
in clean accuracy $acc_{\textit{clean}}$ and certified accuracy $acc_{\textit{cert}}$ on
MAE as the base model.

PC$_\star$, ViP$_\star$, and CC$_\star$ were 
the three top-performing defenders in the literature in these two metrics.
They perform comparably with $D_\text{OMA}$/ViP/PC.
PG++$_\star$  used a weaker base model, and if replaced with MAE, it becomes PG++.
PG++ ($\tau = 0.5$) performs comparably with   $D_\text{OMA}$/ViP/PC.

Furthermore, as the threshold $\tau$ increases from 0.5 to 0.9, 
the certified accuracy of HC increases. 
For example, when $\tau = 0.8$, HC's $acc_{\textit{cert}}$ is greater than
$D_\text{OMA}$'s one
on ImageNet, CIFAR100, and GTSRB by 12.7\%--17.4\%,
which are 
0.5\%--5.8\%
lower than the clean accuracy, respectively.
In contrast, the gaps between clean accuracy and certified accuracy {for all other defenders in the table}
are much larger.
The implication is that almost all correctly predicted samples can be certified by HC, lowering the threat of successful adversarial patch attacks to damage the downstream operations.
To avoid overloading readers with repetitive information, 
the results of the defenders in our experiment on all combinations of base models and datasets are summarized in 
Fig.~\ref{fig:vary_patch_size},
where $\tau=0.8$ is used for both PG++ and HC,
and ViP/PC/$D_\text{OMA}$ share the same results.
We will discuss Fig.~\ref{fig:vary_patch_size} in Section~\ref{sec:RQ3}.

\begin{table}[t]
\caption{All eight possible cases based on the combination of three conditions for benign samples in certified detection. $\checkmark$ if the condition is \textit{True}, otherwise \textit{False}.
}\label{tab:case}\centering
\begin{tabular}{|c|c|c|c|c|c|c|c|c|}
\hline
Case   & 1 & 2 & 3 & 4 & 5 & 6 & 7 & 8 \\ \hline
$f(x)=y_0$ & \checkmark & \checkmark                                                & \checkmark & \checkmark &   &                                                  &   &   \\ \hline
$w(x)$   & \checkmark &                                                  & \checkmark &   & \checkmark &                                                  & \checkmark &   \\ \hline
$v(x)$   & \checkmark & \checkmark                                                &   &   & \checkmark & \checkmark                                                &   &   \\ \hline

\end{tabular}
\end{table}

    \begin{table}[]
\caption{The clean accuracy $acc_{clean}$ (Clean) and certified accuracy $acc_{cert}$ (Cert) of certified detection defenders on ImageNet, CIFAR100, and GTSRB, with patch size 32 (2\%), 35 (2.4\%), and 32 (2\%) pixels for the three datasets, respectively.}\label{tab:main_eva_results}\centering
\resizebox{\linewidth}{!}
{
\begin{tabular}{|c|cc|cc|cc|}
\hline
Dataset      & \multicolumn{2}{c|}{ImageNet}                      & \multicolumn{2}{c|}{CIFAR100}                      & \multicolumn{2}{c|}{GTSRB}                         \\ \hline
 Accuracy         & \multicolumn{1}{c|}{Clean}         & Cert          & \multicolumn{1}{c|}{Clean}         & Cert          & \multicolumn{1}{c|}{Clean}         & Cert          \\ \hline
SC$_\star$ \cite{han2021scalecert}           & \multicolumn{1}{c|}{$\oplus$}           & 55.4          & \multicolumn{1}{c|}{$\oplus$}           & $\oplus$           & \multicolumn{1}{c|}{$\oplus$}           & $\oplus$           \\ \hline
PG++$_\star$ ($\tau$ = 0.8)\cite{PatchCensor_arvix}\setfootnotemark\label{beforefirst} & \multicolumn{1}{c|}{62.9}          & 28.0          & \multicolumn{1}{c|}{$\oplus$}           & $\oplus$           & \multicolumn{1}{c|}{$\oplus$}           & $\oplus$           \\ \hline
PG++$_\star$ ($\tau$ = 0.7)\cite{PatchCensor_arvix} & \multicolumn{1}{c|}{62.9}          & 32.0          & \multicolumn{1}{c|}{$\oplus$}           & $\oplus$           & \multicolumn{1}{c|}{$\oplus$}           & $\oplus$           \\ \hline
PG++$_\star$ ($\tau$ = 0.6)\cite{PatchCensor_arvix} & \multicolumn{1}{c|}{62.9}          & 35.5          & \multicolumn{1}{c|}{$\oplus$}           & $\oplus$           & \multicolumn{1}{c|}{$\oplus$}           & $\oplus$           \\ \hline
PG++$_\star$ ($\tau$ = 0.5)\cite{PatchCensor_arvix} & \multicolumn{1}{c|}{62.9}          & 39.0          & \multicolumn{1}{c|}{$\oplus$}           & $\oplus$           & \multicolumn{1}{c|}{$\oplus$}           & $\oplus$           \\ \hline
MR+$_\star$ \cite{PatchCensor_arvix}         & \multicolumn{1}{c|}{75.5}          & 56.3          & \multicolumn{1}{c|}{$\oplus$}           & $\oplus$           & \multicolumn{1}{c|}{96.4}          & 54.7          \\ \hline
PC$_\star$ \cite{PatchCensor_arvix}\setfootnotemark\label{third}           & \multicolumn{1}{c|}{81.8}          & 69.4          & \multicolumn{1}{c|}{$\oplus$}           & $\oplus$           & \multicolumn{1}{c|}{97.1}          & 70.3          \\ \hline
ViP$_\star$ \cite{li2022vip} & \multicolumn{1}{c|}{\textbf{83.7}} & 74.6          & \multicolumn{1}{c|}{$\oplus$}           & $\oplus$           & \multicolumn{1}{c|}{$\oplus$}           & $\oplus$           \\ \hline
CC$_\star$ \cite{zhou2024crosscert}           & \multicolumn{1}{c|}{81.7}          & 64.8          & \multicolumn{1}{c|}{\textbf{92.5}} & 73.2          & \multicolumn{1}{c|}{$\oplus$}           & $\oplus$           \\ \hline
\hline
PG++ ($\tau$ = 0.9) & \multicolumn{1}{c|}{82.5}          & 13.7          & \multicolumn{1}{c|}{90.2}          & 18.8          & \multicolumn{1}{c|}{\textbf{97.5}} & 39.9          \\ \hline
PG++ ($\tau$ = 0.8) & \multicolumn{1}{c|}{82.5}          & 42.4          & \multicolumn{1}{c|}{90.2}          & 52.0          & \multicolumn{1}{c|}{\textbf{97.5}} & 54.1          \\ \hline
PG++ ($\tau$ = 0.7) & \multicolumn{1}{c|}{82.5}          & 51.5          & \multicolumn{1}{c|}{90.2}          & 60.4          & \multicolumn{1}{c|}{\textbf{97.5}} & 61.0          \\ \hline
PG++ ($\tau$ = 0.6) & \multicolumn{1}{c|}{82.5}          & 57.0          & \multicolumn{1}{c|}{90.2}          & 65.4          & \multicolumn{1}{c|}{\textbf{97.5}} & 65.4          \\ \hline
PG++ ($\tau$ = 0.5) & \multicolumn{1}{c|}{82.5}          & 61.4          & \multicolumn{1}{c|}{90.2}          & 69.2          & \multicolumn{1}{c|}{\textbf{97.5}} & 69.0          \\ \hline
$D_\text{OMA}$/PC/ViP          & \multicolumn{1}{c|}{82.5}          & 69.3          & \multicolumn{1}{c|}{90.2}          & 75.0          & \multicolumn{1}{c|}{\textbf{97.5}} & 74.3          \\ \hline
HC ($\tau$ = 0.5)  & \multicolumn{1}{c|}{82.5}          & 77.9          & \multicolumn{1}{c|}{90.2}          & 81.2          & \multicolumn{1}{c|}{\textbf{97.5}} & 80.7          \\ \hline
HC ($\tau$ = 0.6)  & \multicolumn{1}{c|}{82.5}          & 80.0          & \multicolumn{1}{c|}{90.2}          & 84.0          & \multicolumn{1}{c|}{\textbf{97.5}} & 84.2          \\ \hline
HC ($\tau$ = 0.7)  & \multicolumn{1}{c|}{82.5}          & 81.2          & \multicolumn{1}{c|}{90.2}          & 86.6          & \multicolumn{1}{c|}{\textbf{97.5}} & 87.7          \\ \hline
HC ($\tau$ = 0.8)  & \multicolumn{1}{c|}{82.5}          & 82.0          & \multicolumn{1}{c|}{90.2}          & 88.4          & \multicolumn{1}{c|}{\textbf{97.5}} & 91.7          \\ \hline
HC ($\tau$ = 0.9)  & \multicolumn{1}{c|}{82.5}          & \textbf{82.4} & \multicolumn{1}{c|}{90.2}          & \textbf{89.9} & \multicolumn{1}{c|}{\textbf{97.5}} & \textbf{96.6} \\ \hline
\multicolumn{4}{l}{Note: $\oplus$=No data is provided in the literature. }
\end{tabular}
}
\end{table}

\begin{table*}[]
\caption{Certification and secondary metric results on ImageNet with breakdown analysis (see Table~\ref{tab:case}) for patch size of 32 pixels (2\%)
}
\label{tab:cases_imagenet_2}\centering

\begin{tabular}{|cc||D{.}{.}{1}|D{.}{.}{1}|D{.}{.}{1}||D{.}{.}{1}|D{.}{.}{1}|D{.}{.}{1}||D{.}{.}{1}|D{.}{.}{1}|D{.}{.}{1}|D{.}{.}{1}|D{.}{.}{1}|D{.}{.}{1}|D{.}{.}{1}|D{.}{.}{1}|}
\hline
\multicolumn{2}{|c||}{\multirow{2}{*}{Defender}} & \multicolumn{3}{c||}{Certification} & \multicolumn{3}{c||}{Secondary Metrics} & \multicolumn{8}{c|}{Case (in \%)} \\
\cline{3-16}
& & \multicolumn{1}{c|}{$acc_{cert}$} & \multicolumn{1}{c|}{$r_{cert}$} & \multicolumn{1}{c||}{$r_{\textit{cert}_{\textit{inc}}}$} & \multicolumn{1}{c|}{$acc_{\neg w}$} & \multicolumn{1}{c|}{$r_{\textit{fa}}$} & \multicolumn{1}{c||}{$r_{\textit{fs}}$} & \multicolumn{1}{c|}{1} & \multicolumn{1}{c|}{2} & \multicolumn{1}{c|}{3} & \multicolumn{1}{c|}{4} & \multicolumn{1}{c|}{5} & \multicolumn{1}{c|}{6} & \multicolumn{1}{c|}{7} & \multicolumn{1}{c|}{8} \\
\hline
\multicolumn{1}{|c|}{\multirow{5}{*}{PG++}} & $\tau$ = 0.9 & 13.7 & 13.7 & 0.0 & 82.6 & 0.2 & 99.7 & 0.0 & 13.7 & 0.1 & 68.7 & 0.0 & 0.0 & 0.1 & 17.4 \\
\cline{2-16}
\multicolumn{1}{|c|}{} & $\tau$ = 0.8 & 42.4 & 42.4 & 0.0 & 82.7 & 0.7 & 98.3 & 0.0 & 42.4 & 0.5 & 39.6 & 0.0 & 0.0 & 0.3 & 17.2 \\
\cline{2-16}
\multicolumn{1}{|c|}{} & $\tau$ = 0.7 & 51.5 & 51.5 & 0.0 & 82.9 & 1.5 & 95.5 & 0.0 & 51.5 & 1.3 & 29.7 & 0.0 & 0.0 & 0.8 & 54.3 \\
\cline{2-16}
\multicolumn{1}{|c|}{} & $\tau$ = 0.6 & 57.0 & 57.0 & 0.0 & 83.6 & 3.1 & 89.8 & 0.0 & 57.0 & 2.5 & 22.9 & 0.0 & 0.0 & 1.8 & 15.7 \\
\cline{2-16}
\multicolumn{1}{|c|}{} & $\tau$ = 0.5 & 61.4 & 61.4 & 0.0 & 84.7 & 5.6 & 80.3 & 0.0 & 61.4 & 4.6 & 16.5 & 0.0 & 0.0 & 3.4 & 14.1 \\
\hline
\rowcolor{gray!10}
\multicolumn{2}{|l||}{$D_\text{OMA}$} & 69.3 & 69.3 & 0.0 & 91.3 & 16.0 & 37.9 & 0.0 & 69.3 & 13.2 & 0.0 & 0.0 & 0.0 & 10.9 & 6.6 \\
\hline
\multicolumn{1}{|c|}{\multirow{5}{*}{HC}} & $\tau$ = 0.5 & 77.9 & 81.6 & 39.9 & 94.3 & 25.6 & 21.3 & 16.5 & 61.4 & 4.6 & 0.0 & 3.7 & 0.0 & 10.1 & 3.7 \\
\cline{2-16}
\multicolumn{1}{|c|}{} & $\tau$ = 0.6 & 80.0 & 86.0 & 54.3 & 95.5 & 30.9 & 15.5 & 22.9 & 57.0 & 2.5 & 0.0 & 6.0 & 0.0 & 8.8 & 2.7 \\
\cline{2-16}
\multicolumn{1}{|c|}{} & $\tau$ = 0.7 & 81.2 & 90.1 & 67.8 & 96.6 & 37.6 & 10.3 & 29.7 & 51.5 & 1.3 & 0.0 & 8.9 & 0.0 & 6.8 & 1.8 \\
\cline{2-16}
\multicolumn{1}{|c|}{} & $\tau$ = 0.8 & 82.0 & 93.8 & 79.8 & 97.5 & 48.6 & 6.1 & 39.6 & 42.4 & 0.5 & 0.0 & 11.9 & 0.0 & 4.6 & 1.1 \\
\cline{2-16}
\multicolumn{1}{|c|}{} & $\tau$ = 0.9 & 82.4 & 97.5 & 92.0 & 98.6 & 83.4 & 1.3 & 68.7 & 13.7 & 0.1 & 0.0 & 15.2 & 0.0 & 2.1 & 0.2 \\
\hline
\end{tabular}

\end{table*}

To facilitate further investigation, we performed a detailed analysis on ImageNet samples with the patch size 32 pixels, summarized in Table~\ref{tab:cases_imagenet_2}.
\paragraph{\textbf{Overall Precision in Certification}}
The certification column in Table~\ref{tab:cases_imagenet_2} shows that HC achieves higher certified accuracy $acc_{cert}$, certified ratio $r_{cert}$, and certified ratio for inconsistent samples $r_{cert_{inc}}$ than the peer defenders.
$D_\text{OMA}$ and PG++ are 0 in $r_{cert_{inc}}$ since they cannot certify any inconsistent samples by their theories. 
HC has a strong ability to certify these samples with $r_{cert_{inc}}$ in 39.9\%--92.0\%.
Also, the results in $acc_{cert}$ and $r_{cert}$ for each of PG++ and $D_\text{OMA}$ (PC, ViP) in the table are identical (no difference in number of samples) because these defenders are highly ineffective on incorrectly predicted samples.
We find that all
incorrectly predicted samples under the setting in Table~\ref{tab:cases_imagenet_2} are inconsistent samples, which is a major certification target by the design of HC.
HC demonstrates a growing difference between $r_{cert}$ and $acc_{cert}$ (from 3.7\% to 15.1\%) as $\tau$ increases, where $acc_{cert}$ itself is increasing at a slower pace than $r_{cert}$, highlighting HC's effectiveness in certifying incorrectly predicted samples and further in inconsistent samples.

\paragraph{\textbf{Analysis with Secondary Metrics}}
HC shows higher silent accuracy $acc_{\neg w}$ (94.3\%--98.6\%) compared to PG++ (82.6\%--84.7\%) and $D_\text{OMA}$ (91.3\%), indicating higher reliability 
when a detector keeps silent in Table~\ref{tab:cases_imagenet_2}. 
This aligns with HC's design rationale of offering correctness and robustness for the samples passing through undetected: incorrectly predicted samples are effectively filtered out by issuing a warning. 

Readers may wonder whether HC relies on producing a relatively excessive number of indiscriminate warnings to reject harmful samples. 
On the one hand, PG++ abstains from the warning on many ambiguous samples, decreasing the false alert ratio $r_{\textit{fa}}$ from 16.0\% to 0.2\%--5.6\%;
however, its certified ratio $r_{cert}$ is suppressed at a low level, as stated above, and its false silent ratio $r_{\textit{fs}}$ becomes pretty high (80.3\%--99.7\%).
Indeed, tightening the warning criterion produces more false alerts:
As shown in Table~\ref{tab:cases_imagenet_2}, for example, when the hyperparameter $\tau$ = 0.8, HC maintains 61.2\% ($=42.4/69.3$) of all samples that $D_\text{OMA}$ can make certified and silent (grouped as Case 2 by $D_\text{OMA}$), where $r_{\textit{fa}}$ increases from 16.0\% to 48.6\%.
At the same time, HC produces a very small proportion of samples are incorrectly predicted, uncertified, and silent (Case 8), only 16.7\% ($=1.1/6.6$) of those offered by $D_\text{OMA}$, where $r_{\textit{fs}}$ also decreases from 37.9\% to 6.1\%, which makes HC a much more robust choice than $D_\text{OMA}$ in addressing the adversarial patch issue in safety-critical scenarios.

\paragraph{\textbf{Breakdown of Cases on ImageNet}}
%
To understand the tradeoff that HC made and its implications, we analyze the proportions of samples in all eight cases (see Table~\ref{tab:case}, where $y_0$ is the true label for input $x$, $w()$ and $v()$ are respectively the warning function and the certification function) shown in the case column in Table~\ref{tab:cases_imagenet_2}.
The detailed breakdowns of samples in PG++ and HC demonstrate how their philosophies differ and how HC achieves state-of-the-art performance in detail.
    
First, when their $\tau$ values are equal, their sample proportions are the same in two cases (Cases 2 and 3).
This is because when $f(x)=y_0$, the condition of Case 2 ($v(x)\land \neg w(x)$) for PG++ and HC is reduced to the same condition, so does $\neg v(x)\land w(x)$ (Case 3).
{They are also the targets of $D_\text{OMA}$.}
\footnotetext[\getrefnumber{beforefirst}]{The main tables (Table~1 and Table~2) of the published version of PatchCensor \cite{patchcensor} have results in different metrics but marked with the same names, which are rectified in their updated white paper \cite{PatchCensor_arvix}. Therefore, we adopt the updated ones from \cite{PatchCensor_arvix} in this work. }
\footnotetext[\getrefnumber{third}]{Results on GTSRB of PC$_\star$ from \cite{PatchCensor_arvix} in their main table (Table 2) violate that in their detailed result table (Table 7). We adopt the latter one.}

Second, even though Case 1 for HC and Case 4 for PG++ refer to the same set of samples (i.e., the condition of $v(x)\land w(x)$ in HC is equivalent to the condition of $\neg v(x)\land \neg w(x)$ in PG++),
HC and PG++ offer drastically different consequences:
HC warns and certifies these samples (guiding these samples in Case 1 and all their harmful samples to the alert path in Fig. \ref{fig:application}), but
PG++ abstains from providing both warning and certification (guiding these samples in Case 4 to the silent path in Fig. \ref{fig:application} and has no guiding effects on harmful samples).
%

Last but not least, HC and PG++ behave differently on incorrectly predicted samples. 
HC proactively warns samples (in Cases 5 and 7) and further certifies some (these in Case 5) with a warning guarantee on their harmful samples.
PG++ frequently abstains (Case 8) with the ratio of Case 8 to the sum of Case 7 and Case 8 ranging from 80.6\% to 99.4\%  for the five $\tau$ values.
Both get zero in Case 6 ---
incorrectly predicted ImageNet samples with strong support for the true label are extremely difficult to find.
$D_\text{OMA}$ 
always abstains from certifying any benign samples with inconsistent mutants,
including those correctly predicted (some in Case 3) and especially those incorrectly predicted (all in Cases 7 and 8), unlike HC.
Regarding these samples in Case 3 of $D_\text{OMA}$,
HC gradually transfers them to Case 1 by checking the confidence of mutants (e.g., only 0.5\% samples left in Case 3 when $\tau=0.8$), which is certified and warned.
{As two sides of the same coin, HC also moves certified samples in Case 2 of $D_\text{OMA}$ to its Case 1,
which increases $r_{\textit{fa}}$, reducing the number of samples originally passed to the downstream operations but protecting them from the harmful samples of those in Case 3 of $D_\text{OMA}$.}

For inconsistent samples in Cases 7 and 8,
HC gradually transfers them to Case 5 
(e.g., 67.6\% incorrectly predicted samples are certified and warned when $\tau=0.8$, which is 0 in $D_\text{OMA}$);
note that these samples are incorrectly predicted, and those in Case 8 are without warning.
Doing so, HC not only certifies more samples but also increases the silent accuracy $acc_{\neg w}$ and decreases the false silent ratio $r_{\textit{fs}}$.

We also summarize the results on ImageNet with ViT and RN for the same patch size 32 pixels (2\%) in Section \ref{app:other_result}, whose observations are similar to those in MAE.

\textbf{Summary}: HiCert achieves a new state-of-the-art performance in terms of certified accuracy and certified ratio (both in general and for inconsistent samples) on all three datasets.
\begin{figure*}[t] 
\centering
\includegraphics[width=0.75\linewidth]{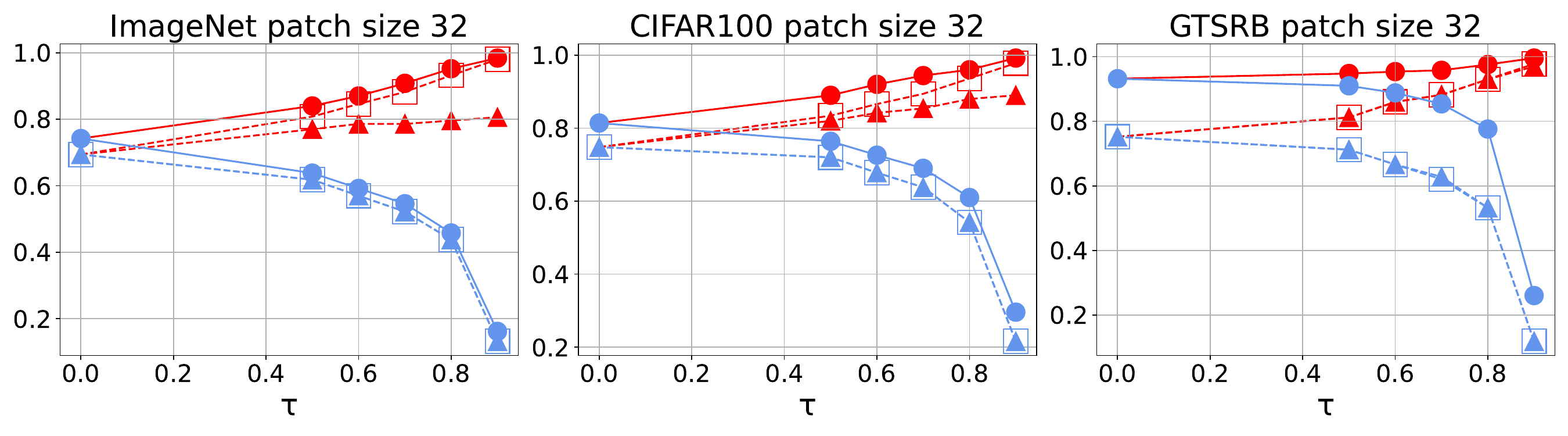}
\includegraphics[width=0.75\linewidth]{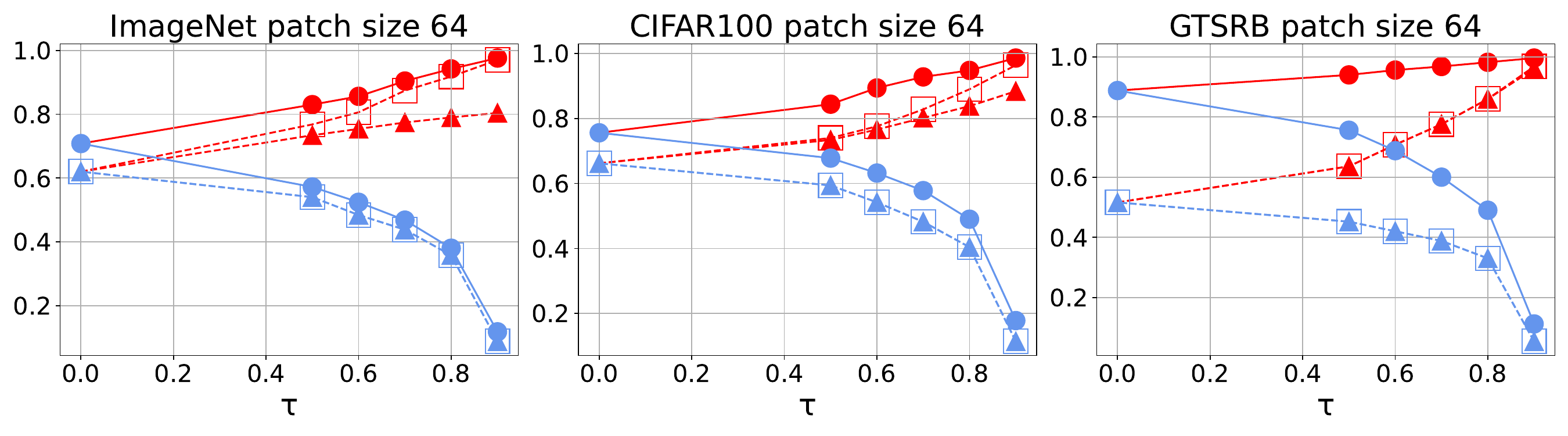}
\includegraphics[width=0.75\linewidth]{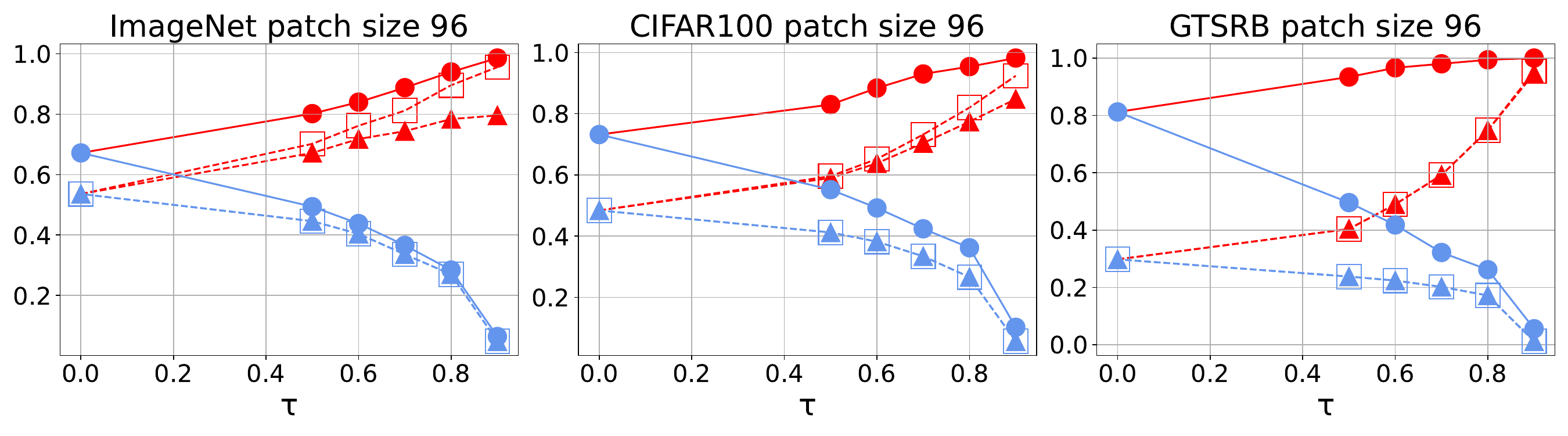}
\includegraphics[width=0.9\linewidth]{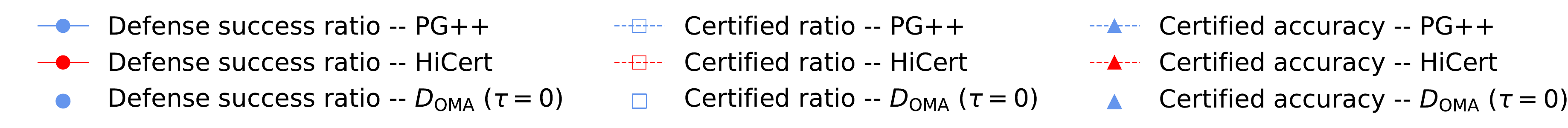}
\caption{
Actual attacks vs. theoretical guarantee.
Each plot shows six lines: three for PG++ and three for HC, representing 
defense success ratio (solid line), certified ratio (dashed line with hollow markers), and certified accuracy (dashed line with solid markers).
Each line has six markers, where the one shared by HC and PG++ at $\tau=0$ is the result of $D_\text{OMA}$, and the other five are the corresponding results for HC/PG++ with $\tau \in [0.5, 0.9]$ ($x$-axis).
}

\label{fig:real_attack}
\end{figure*}

\subsubsection{Answering RQ2}\label{sec:attack}
In this section, we study the defense ability of defenders through conducting an actual attack.

\paragraph{\textbf{Overall Comparison on Defense Success Ratio}}
Fig.~\ref{fig:real_attack} summarizes the defense success ratio $r_{suc}$ against the actual adversarial patch attack (empirical) and the certified accuracy $acc_{cert}$ and certified ratio $r_{cert}$ (theoretical) of different defenders with MAE.
To our knowledge, this paper is the first work to perform a real adversarial patch attack on certified detection defenders.
Every $r_{suc}$ is larger than the corresponding $r_{cert}$ of each defender, as expected.

The overall results in Fig.~\ref{fig:real_attack} show the strong empirical defense ability of HC and a significantly larger security risk of $D_\text{OMA}$ (including PC and ViP) and PG++.
On ImageNet with all three patch sizes, the attacker tool can successfully attack $D_\text{OMA}$ with
more than a quarter of benign samples  (where the gap between its blue circle markers and the ceiling (the line of $y=1$) is still observably large), which appears dangerous for safety-critical downstream operations.
HC largely mitigates this problem by significantly reducing the gap (e.g., the proportion of benign samples that can be successfully attacked on HC is always less than a fifth as shown in all sub-plots of Fig.~\ref{fig:real_attack}) even if the patch size is larger --- for example, HC with $\tau=0.8$ is 94\% in the defense success ratio, while $D_\text{OMA}$ is only 67.2\% and PG++ with $\tau=0.8$ are 6.4\%--49.4\%.
We also observe a similar trend on CIFAR100 and GTSRB, where the proportion of samples that failed to be defended is 18.6\%--27.7\% on CIFAR100 and 6.8\%--18.8\% on GTSRB
for $D_\text{OMA}$ in all patch sizes, and the gap between the red {dot markers of HC and the ceiling is significantly smaller than the other two defenders in all cases.}

\begin{figure*}[]
\caption{
The table in the figure shows  $acc_{cert}$ when the patch size is 16. The nine plots show $acc_{cert}$ of HC ($\tau=0.8$), PG++ ($\tau=0.8$), and $D_\text{OMA}$ on ImageNet, CIFAR100, and GTSRB on the base models MAE, ViT, and RN, respectively (from top to bottom and from left to right) relative to the certified accuracy  $acc_{cert}$ of the same defender on the same dataset when the patch size is 16.
}
{
\vspace{3ex}
\setlength{\tabcolsep}{1pt}
\begin{minipage}{\linewidth}
\resizebox{\linewidth}{!}
{
\begin{tabular}{|c|ccccccccc|ccccccccc|ccccccccc|}
\hline
Model    & \multicolumn{9}{c|}{MAE}                                                                                                                                                                                                             & \multicolumn{9}{c|}{ViT}                                                                                                                                                                                                             & \multicolumn{9}{c|}{RN}                                                                                                                                                                                                              \\ \hline
Dataset  & \multicolumn{3}{c|}{ImageNet}                                                     & \multicolumn{3}{c|}{CIFAR100}                                                     & \multicolumn{3}{c|}{GTSRB}                                   & \multicolumn{3}{c|}{ImageNet}                                                     & \multicolumn{3}{c|}{CIFAR100}                                                     & \multicolumn{3}{c|}{GTSRB}                                   & \multicolumn{3}{c|}{ImageNet}                                                     & \multicolumn{3}{c|}{CIFAR100}                                                     & \multicolumn{3}{c|}{GTSRB}                                   \\ \hline
Defender & \multicolumn{1}{c|}{HC}   & \multicolumn{1}{c|}{$D_\text{OMA}$}  & \multicolumn{1}{c|}{PG++} & \multicolumn{1}{c|}{HC}   & \multicolumn{1}{c|}{$D_\text{OMA}$}  & \multicolumn{1}{c|}{PG++} & \multicolumn{1}{c|}{HC}   & \multicolumn{1}{c|}{$D_\text{OMA}$}  & PG++ & \multicolumn{1}{c|}{HC}   & \multicolumn{1}{c|}{$D_\text{OMA}$}  & \multicolumn{1}{c|}{PG++} & \multicolumn{1}{c|}{HC}   & \multicolumn{1}{c|}{$D_\text{OMA}$}  & \multicolumn{1}{c|}{PG++} & \multicolumn{1}{c|}{HC}   & \multicolumn{1}{c|}{$D_\text{OMA}$}  & PG++ & \multicolumn{1}{c|}{HC}   & \multicolumn{1}{c|}{$D_\text{OMA}$}  & \multicolumn{1}{c|}{PG++} & \multicolumn{1}{c|}{HC}   & \multicolumn{1}{c|}{$D_\text{OMA}$}  & \multicolumn{1}{c|}{PG++} & \multicolumn{1}{c|}{HC}   & \multicolumn{1}{c|}{$D_\text{OMA}$}  & PG++ \\ \hline
$acc_{cert}$     & \multicolumn{1}{c|}{82.1} & \multicolumn{1}{c|}{71.7} & \multicolumn{1}{c|}{45.1} & \multicolumn{1}{c|}{89.4} & \multicolumn{1}{c|}{78.8} & \multicolumn{1}{c|}{57.1} & \multicolumn{1}{c|}{93.0} & \multicolumn{1}{c|}{82.2} & 66.2 & \multicolumn{1}{c|}{78.1} & \multicolumn{1}{c|}{68.9} & \multicolumn{1}{c|}{62.1} & \multicolumn{1}{c|}{87.4} & \multicolumn{1}{c|}{81.7} & \multicolumn{1}{c|}{76}   & \multicolumn{1}{c|}{80.7} & \multicolumn{1}{c|}{59.1} & 47.5 & \multicolumn{1}{c|}{76.7} & \multicolumn{1}{c|}{61.1} & \multicolumn{1}{c|}{51.1} & \multicolumn{1}{c|}{80.5} & \multicolumn{1}{c|}{62.9} & \multicolumn{1}{c|}{49.8} & \multicolumn{1}{c|}{80.6} & \multicolumn{1}{c|}{42.8} & 29.8 \\ \hline
\end{tabular}
}
\end{minipage}
}
\begin{minipage}{\linewidth}
\label{fig:vary_patch_size}
\vspace{2ex}
\centering
\includegraphics[width=0.32\linewidth]{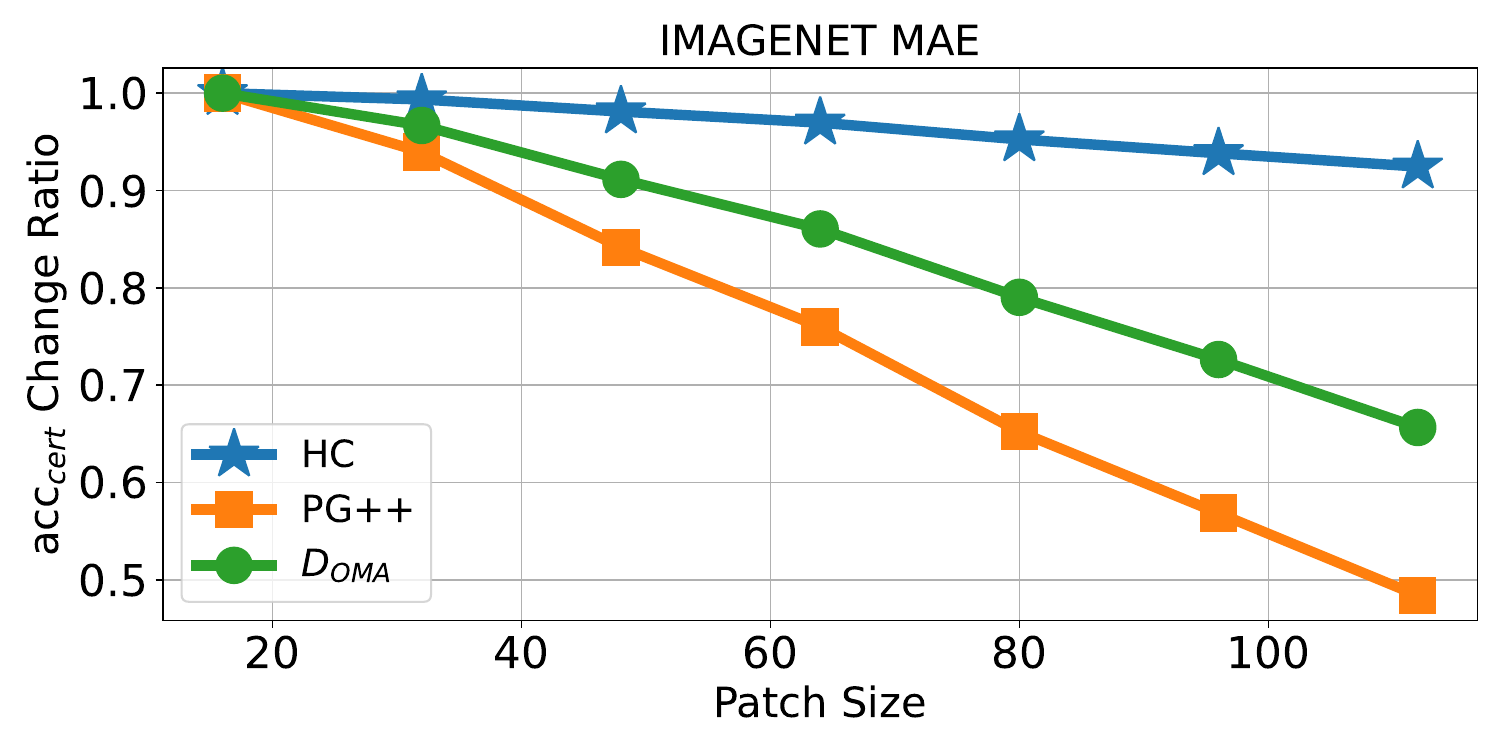}
\includegraphics[width=0.32\linewidth]{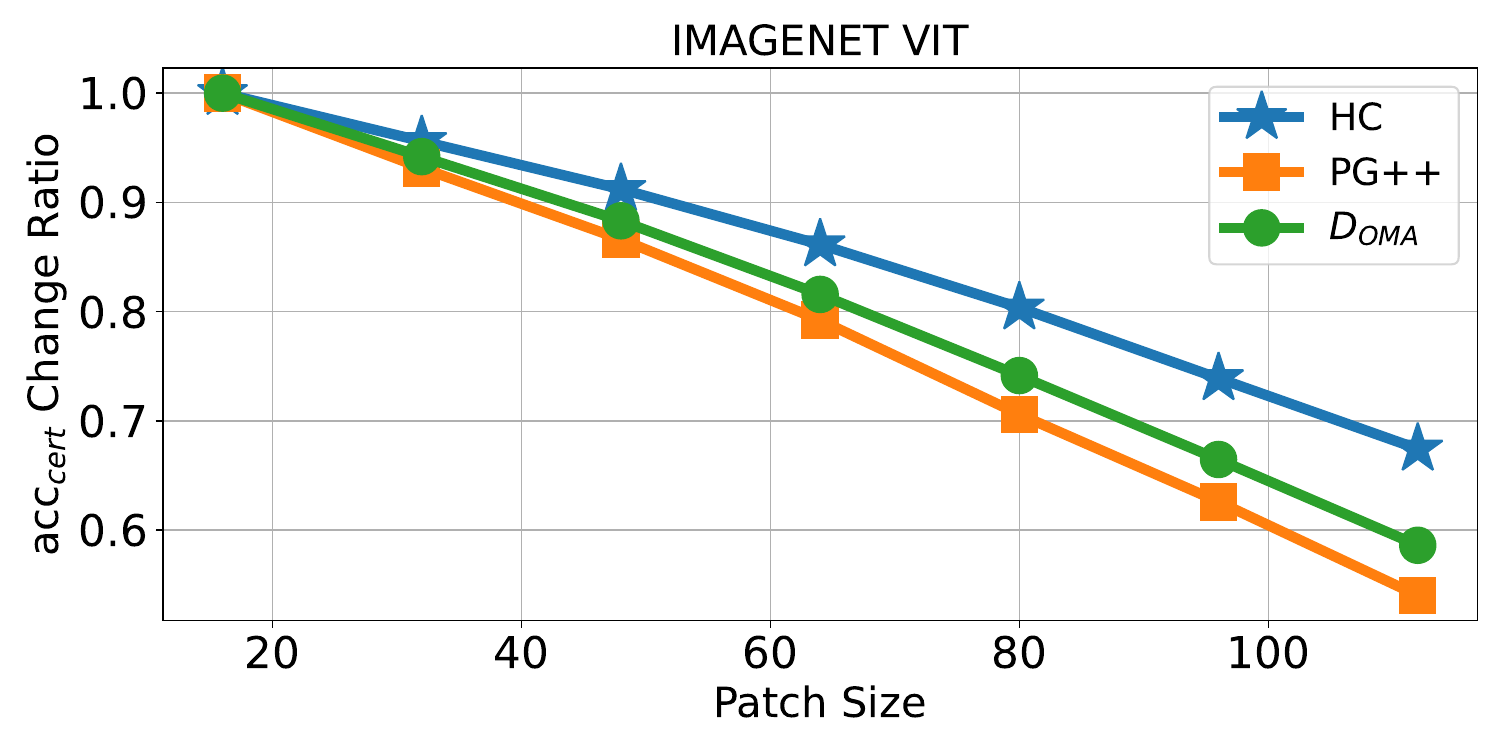}
\includegraphics[width=0.32\linewidth]{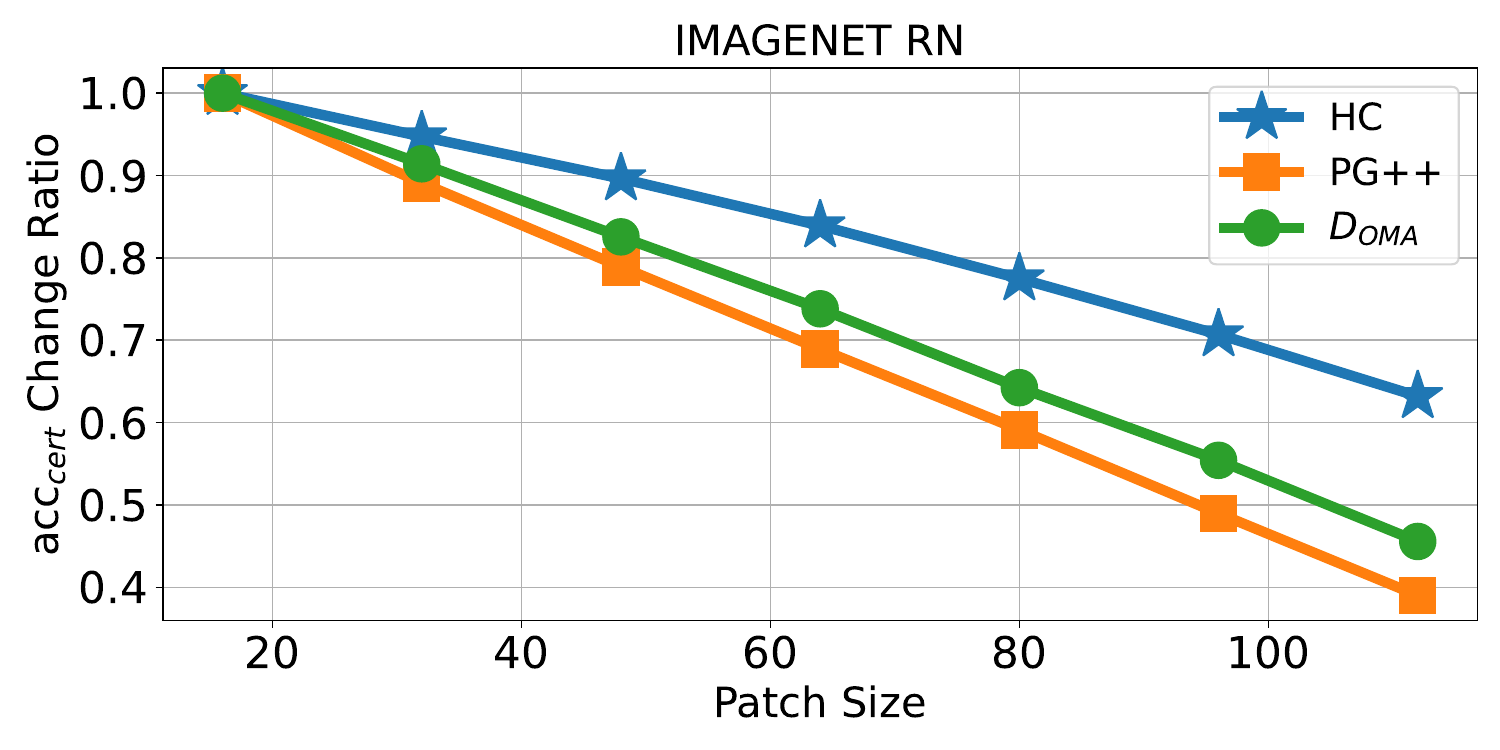}
\includegraphics[width=0.32\linewidth]{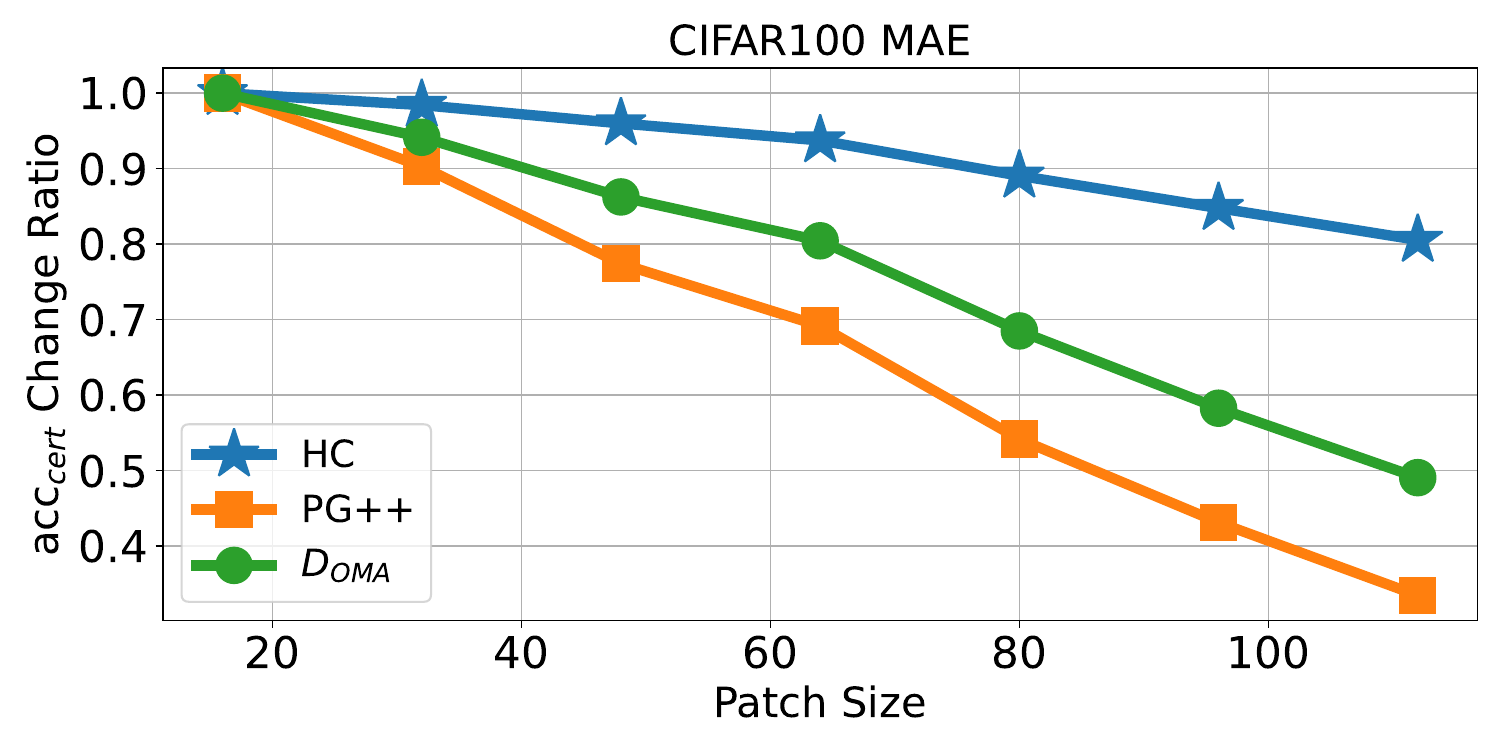}
\includegraphics[width=0.32\linewidth]{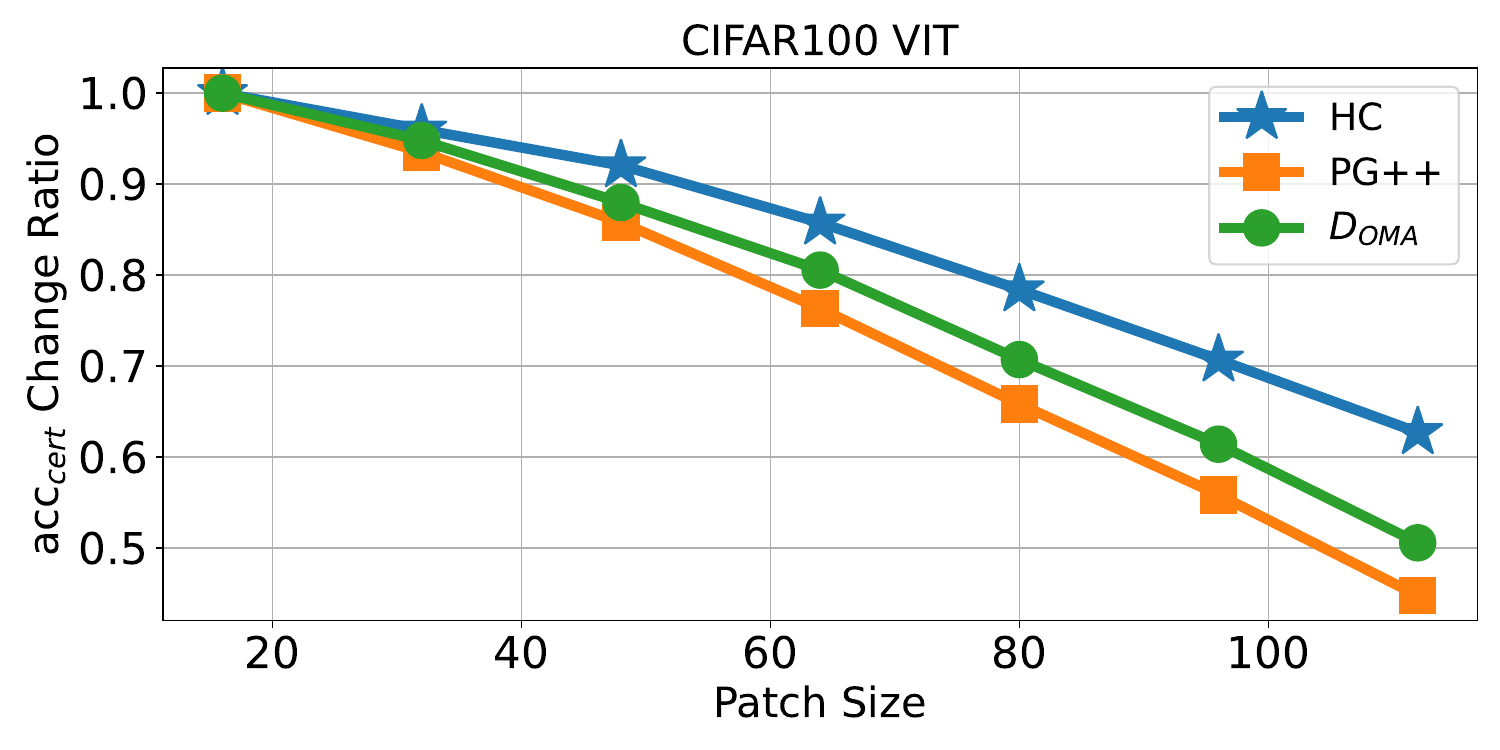}
\includegraphics[width=0.32\linewidth]{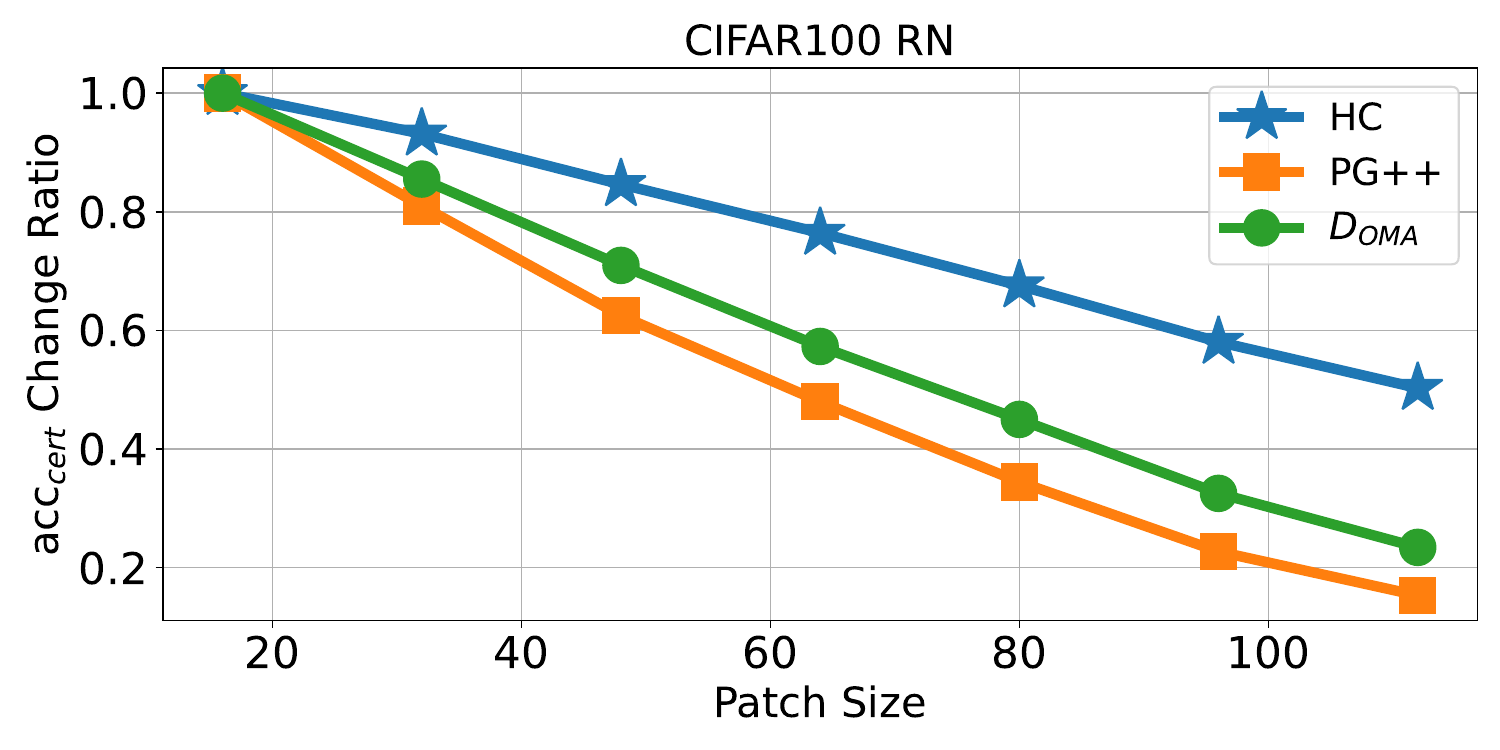}
\includegraphics[width=0.32\linewidth]{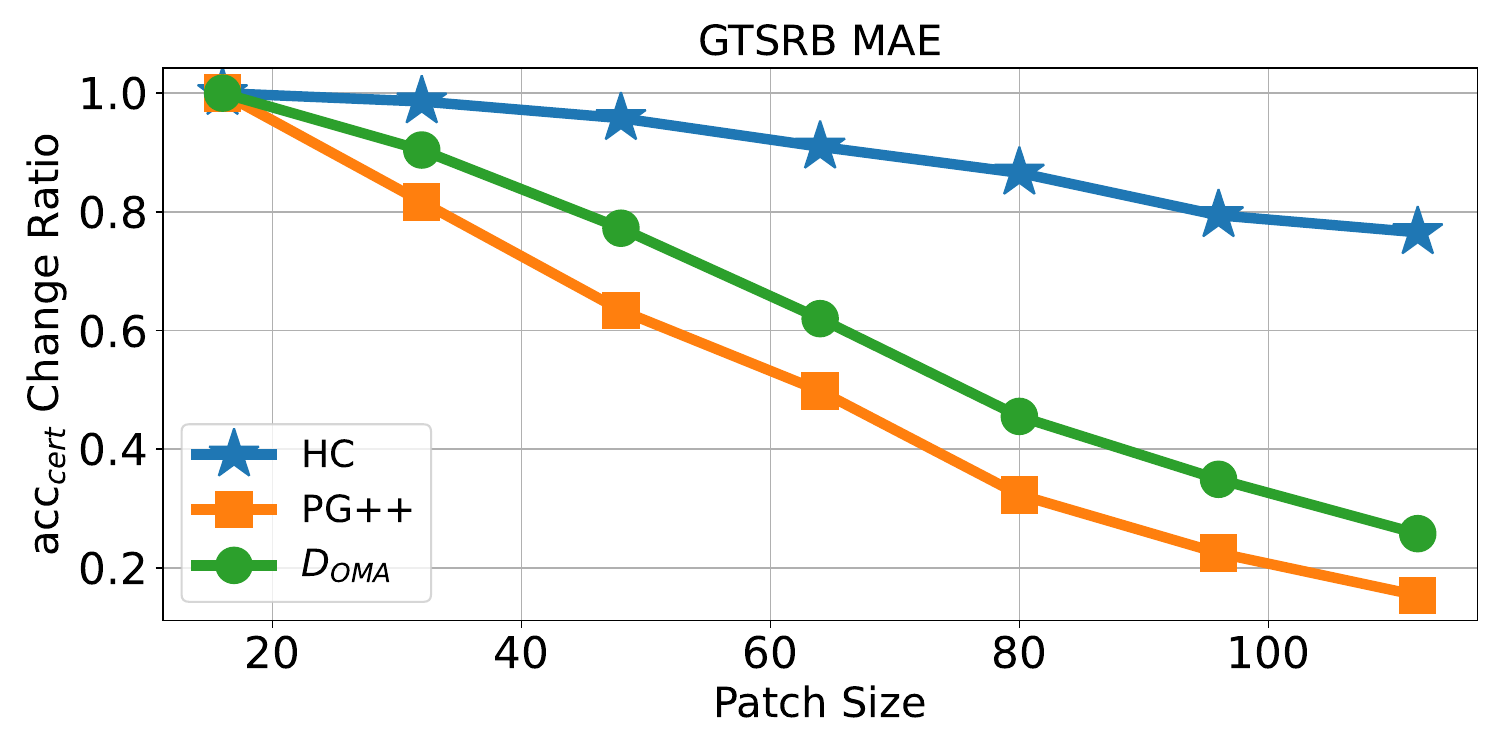}
\includegraphics[width=0.32\linewidth]{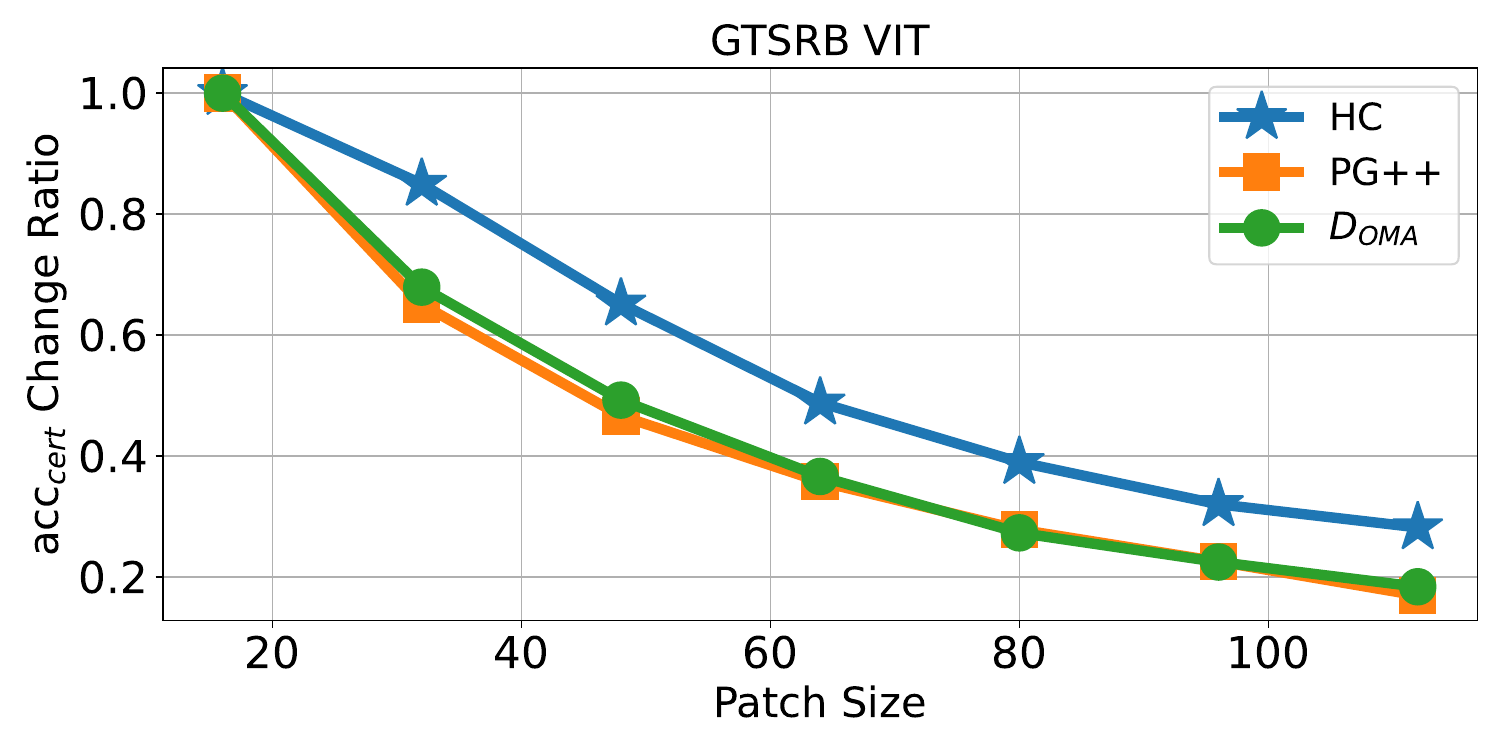}
\includegraphics[width=0.32\linewidth]{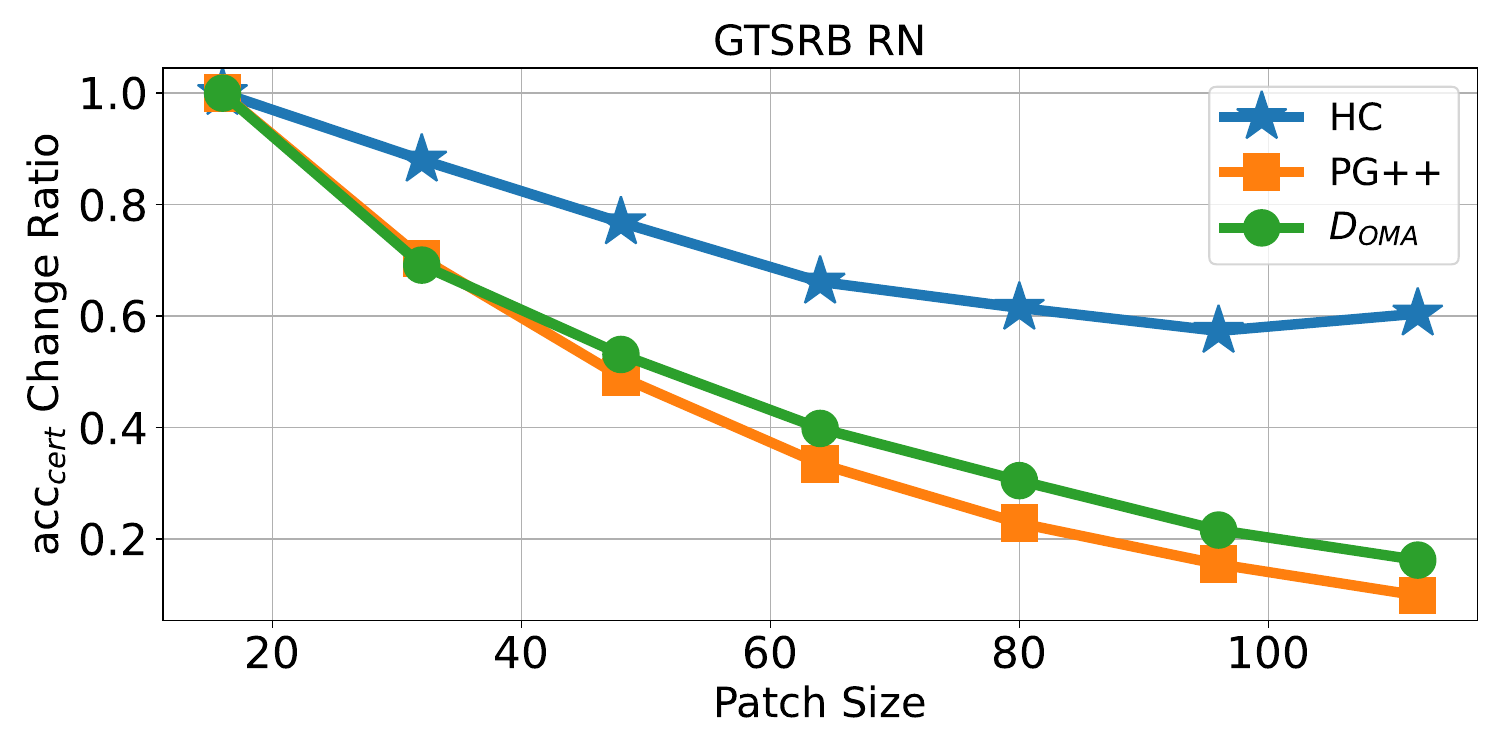}

\label{example}
\end{minipage}  
\end{figure*}

\paragraph{\textbf{Certified Ratio vs. Defense Success Ratio}}
As datasets vary from GTSRB (simple) to CIFAR100 and ImageNet (complex), the curves for certified ratio $r_{cert}$ increasingly approximate those for defense success ratio $r_{suc}$ for PG++, $D_\text{OMA}$, and HC, and
 the gap between each corresponding pair of the curves is much smaller for the more complex ImageNet dataset, suggesting that the theory behind these masking-based certified detection defenders is more applicable on more complex and real-life datasets like ImageNet. This trend is also observed as the patch size decreases from 96 to 32, indicating tighter approximations with smaller patch sizes.



\paragraph{\textbf{Comparison on Accurate Base Models}}
The base models are accurate on GTSRB, correctly predicting almost all benign samples (see Table \ref{tab:clean_acc}).
This allows indirect reviews on the effect of ${D_\text{OMA}}$ checking (designed to certify such samples).
However, the certified accuracy $acc_{cert}$ for $D_\text{OMA}$ deteriorates significantly as the patch size increases.
HC inherits this limitation.
Its ability to certify inconsistent samples improves the situation as $\tau$ increases, and its gap between certified accuracy $acc_{cert}$ and certified ratio $r_{suc}$ remains noticeable even at $\tau = 0.8$ with a patch size of 96. 
PG++'s gap also decreases but at the cost of a much lower defense success ratio.
The result indicates that ${D_\text{OMA}}$ checking for correctly predicted samples on an accurate base model incurs non-trivial issues compared to a less accurate base model for real-life datasets like ImageNet.
{Besides, the clean accuracy for ${D_\text{OMA}}$ generally decreases from GTSRB (simpler dataset with a quite accurate base model) to CIFAR and ImageNet (more complex datasets with less accurate base models), but, the respective certified accuracy increases from GTSRB to CIFAR and ImageNet across different patch sizes. It seems that the difference is narrowing (instead of widening) as the model becomes much less accurate and the dataset becomes more real.}
These observations challenge the conventional wisdom for $D_{\textit{OMA}}$ checking that certification could focus on correctly predicted consistent samples {(evident by the pivot role of the $D_{\textit{OMA}}$ condition in their certification-warning designs and the emphasis on measuring on correctly predicted benign samples in their published evaluations) and leave the base model to improve its clean accuracy to address other issues (hopefully improving the proportion of benign samples they can certify)}, aligning with HiCert's motivation.

\textbf{Summary}:
HiCert is significantly more effective than peer techniques (PG++ and $D_{\textit{OMA}}$).
Its defense success ratios are higher than 80\% in defending against actual patch attacks across all combinations of dataset and patch size combinations.

\subsubsection{Answering RQ3}
\label{sec:RQ3}
In this section, we study the sensitivity of defenders' robustness to changes in patch size.
Fig.~\ref{fig:vary_patch_size} shows the certified accuracy results for HC, $D_\text{OMA}$, and PG++ with all base models on all datasets at $\tau = 0.8$. HC consistently achieves the highest certified accuracy across datasets with the patch size of 16 pixels shown in the table in Fig.~\ref{fig:vary_patch_size}, similar to the results in Tables \ref{tab:main_eva_results} and \ref{tab:cases_imagenet_2}. 
The results in the table are then normalized to 1 to be used as the starting point in all nine plots in Fig.~\ref{fig:vary_patch_size}, respectively, for analyzing the sensitivity under different patch sizes.


In Fig.~\ref{fig:vary_patch_size}, across all plots, the certified accuracy of HC decreases less with increasing patch size compared to those of $D_\text{OMA}$ and PG++. 
For example, on CIFAR100 with ViT, when the patch size is 32, the certified accuracy change ratios (the change ratio of a defender at the patch size $n$ is defined as certified accuracy $acc_{cert}$ at patch size $n$ divided by $acc_{cert}$ at patch size 16) for $D_\text{OMA}$, PG++, and HC are all around 95\%.
At a patch size of 112 (note: 112/32 = 3.5), $D_\text{OMA}$ and PG++ drop to 50.5\% and 44.8\%, while HC attains 62.7\%. 
The primary reason is that HC can certify inconsistent samples, and there are increasingly more inconsistent samples with the increasing patch size and the corresponding increasing size of the mask.
The effectiveness of MAE is better than ViT and RN for all defenders across all datasets.
Additionally, in the plots, HC drops more gently on ImageNet with MAE, maintaining higher than 90\% {in certified accuracy}, while PG++ and $D_\text{OMA}$ drop to 48.3\% and 65.6\%, respectively.
The result shows that the overall effectiveness of HC in certified accuracy is less sensitive to the change in patch size than its peer defenders, which stands out more in a tougher patch adversary scenario.

\textbf{Summary}:
HiCert consistently surpasses peers in certified accuracy and change ratio across diverse patch sizes.

\subsection{Evaluation Summary}
HC is significantly more effective in protecting safety-critical downstream processes by reducing the potential number of harmful samples for their operations.
The certification effect is evaluated on both small (RQ1) and larger patches (RQ3) and validated empirically against actual attacks (RQ2).

\subsection{Further Discussion}
We further investigate HiCert's performance on other patch configurations, its trade-off between performance and time cost, as well as between certified ratio and false alert ratio, 
and hard samples that HiCert fails to certify in a high threshold $\tau$.

HiCert can defend against two adversarial patches by applying two masks on a single mutant, as well as against one patch of an arbitrary rectangular shape using a general rectangle covering mask set, following the strategy of~\cite{xiang2022patchcleanser}. 
More patches or multiple rectangular patches can also be covered by this methodology.
Our defense succeeds as long as at least one mutant includes mask(s) that cover all patch(es), thereby fully eliminating the attacker's influence.
We conduct a case study to evaluate HiCert under these two additional patch configurations.
The results show that HiCert can largely preserve certification performance (with variations within 2\% in certified accuracy, 4\% in certified ratio, and 13\% in certified ratio in inconsistent samples), 
at the cost of a modest increase in false alerts.
The details are shown in Section~\ref{sec:shape/number}.
We note that this drop in performance is a common limitation in masking-based certified detection techniques \cite{li2022vip}, as larger or more numerous masks generally increase the total masked area. Furthermore, over-approximation can help mitigate the shared challenge that some information about the patch attacker may be unknown. 
Notably, the results in RQ2 suggest that HiCert equipped with MAE remains robust even when mutants employ large mask areas.

With the same masking strategy to generate mutants, the three defenders $D_\text{OMA}$, PG++, and HiCert share the same time cost for mutant generation since the confidence value of each mutant is also obtained when $D_\text{OMA}$ predicts their labels.
Regarding the time cost for inference,  the computation time on the certification results of HC/$D_\text{OMA}$/PG++ for all test samples in one configuration is all within 1 second. 
Note that all $D_\text{OMA}$/PG++/HiCert adopt the same 36 masks in the covering mask set in our experiments, and the number of masks in a covering mask set can be reduced by increasing the mask size, which further decreases the number of mutants and the time cost of HiCert.
The trade-off among certified ratio $r_{cert}$/false alert ratio $f_{fa}$ and inference time of HiCert is shown in Section~\ref{sec:time}, demonstrating that HiCert can shorten the inference time cost (e.g., even reaching 0.02s per sample) without significantly sacrificing performance.
Note that Table 1 in \cite{li2022vip} reports the original ViP and PC have a per-sample runtime of 0.92 s, higher than 0.16 s by HiCert in the main setting.

Like PG++ and MR, HiCert is also parametric. 
We visualize HiCert's trade-off between $r_{{cert}}$ and $r_{{fa}}$ by varying $\tau$ from 0 to 1 in steps of 0.01, with the result shown in Section~\ref{sec:visual}. From the result, we observe that both $r_{{cert}}$ and $r_{{fa}}$ are relatively insensitive for $\tau \in [0, 0.4]$ and increase steadily as $\tau$ approaches 0.8. Notably, $r_{{fa}}$ rises sharply once $\tau$ exceeds 0.8.
HiCert users may choose a larger $\tau$ (e.g., $\tau$ = 0.8) to protect more benign samples in safety-critical applications or a smaller $\tau$ to reduce false alerts.
We leave the investigation into nonparametric HiCert as future work.

We also perform a qualitative analysis on hard samples that cannot be certified by HiCert, even when the confidence reaches 0.9 (see Section~\ref{sec:hard sample}). 
Upon manual inspection, we find that these hard samples fall into two main categories:
(1) inputs containing two (or more) items from different classes, where masking one item causes the mutant to be predicted as the other class; and
(2) inputs containing a single item, where the mask changes its semantics, leading to misclassifications of the masked mutants.
Promising research directions include adapting HiCert to be a certification method for multi-label classification and incorporating reliable contextual information from the masked regions in mutants.

\subsection{Threats to Validity}
We evaluate the top-performing and closely related peer defenders in the experiments and follow common practices in recent studies \cite{patchcensor,li2022vip, zhou2024crosscert} on certified detection to compare the reported results of more defenders published in the literature due to our limited GPU resource. 
The main experiments use patch regions in one square, like most certified detection defenders \cite{patchcensor,zhou2024crosscert,xiang2021patchguard++,han2021scalecert, mccoyd2020minority}.
We perform actual adversarial patch attacks on a subset of configurations like \cite{levine2020randomized}, due to high costs.
We use 
false alert ratio, false silent ratio, and silent accuracy as secondary metrics and avoid composite metrics following \cite{yatsura2023certified} due to their lack of theoretical guarantees. 

We do not include CrossCert \cite{zhou2024crosscert} and the original PG++ \cite{xiang2021patchguard++}, PC \cite{patchcensor}, and ViP\cite{li2022vip} in our infrastructure for evaluation.
CrossCert requires a pair of different types of certified recovery defenders as its base defenders and
has lower certified accuracy than ViP/PC.
The original PG++/PC/ViP is deeply coupled with their specific base models; for fair comparison, we unify PC/ViP as $D_\text{OMA}$ and PG++ in our infrastructure and achieve a comparable performance to their reported result.
Unlike traditional program analysis techniques, it is generally impractical to determine whether a benign sample is certifiable, irrespective of the defender.
We are also not aware of the presence of such a dataset.
%
%
To our knowledge, all existing works on certified detection rely on theorem development to ensure complete detection coverage of all harmful versions of certified benign samples, and we follow this practice.
We additionally apply a weaker criterion for evaluation: generating numerous patched versions for each sample and empirically checking whether any generated harmful version bypasses the detection. 
Note that although the attack tool we adopt is a traditional one, recent attack techniques primarily target more difficult scenarios in which much of the deep learning model's internal information is hidden (such as treating the model as a black box \cite{wei2023simultaneously} or further limiting the prediction label query budget \cite{tao2023hardlabel}). In contrast, our chosen tool has full access to gradients, confidence scores, and predicted labels of the base deep learning model, enabling a more effective evaluation.

The experiment has only evaluated certified detection defenders in certain combinations of datasets, patch sizes, base models, $\tau$ values, covering mask sets, a masking strategy, and an attacker tool. 
A larger experiment would certainly enhance the generalization of the current results.

In our experiment, we consistently configure all techniques with the same base model for comparison 
because there would be different results if using different base models for the same defender to certify the same sample, which is a 
common property for these certification techniques \cite{mccoyd2020minority, han2021scalecert, patchcensor, li2022vip, xiang2021patchguard++}.
These techniques will also incur higher computation overhead if a larger base model or a larger dataset is used for evaluation. 
We leave the research on certified defenses alleviating the impact from the base model as future work.
We also observe that many samples in the same class 
are very similar,
while HiCert provides different results of warning; so one possible practical way to mitigate the problem of noisy warnings is to rapid-fire a few shots on an item to make the final decision.

Our implementation may contain bugs. We have validated it through testing and removed all bugs we know.


\section{Related Work}\label{sec:related_work}
There has been a growing focus on the reliability of deep learning systems against adversarial attacks \cite{hussain2024evaluating, Xu2023ASQ, AlMaliki2023Toward}.
Patch adversarial attacks are a typical adversarial paradigm that is physically realizable \cite{hussain2024evaluating}.
Brown et al. 
\cite{brown2017adversarial} first propose patch attacks
as a threat model for physically adversarial attacks. In line with it, various approaches \cite{eykholt2018robust,liu2020bias,wei2022adversarial, tao2023hardlabel, wei2023simultaneously} have been proposed for more advanced adversarial patch attacks against DL models.
%
Several empirical defenses \cite{chen2023jujutsu,naseer2019local,hayes2018visible,hussain2024evaluating} have been introduced, while they are known to be easily compromised \cite{tramer2020adaptive, li2023sok} 
if an attacker is aware of the defense strategies, which are also applicable for defenses against patch attacks demonstrated by Chiang et al. \cite{chiang2020certified}.
Therefore, to develop a provable deterministic defense for such scenarios, a category of robustness certification defenders against patch attacks emerges.

Certified recovery, including those voting-based \cite{levine2020randomized,salman2022certified,xiang2021patchguard, metzen2021efficient, zhou2023majority, li2022vip} and masking-based \cite{xiang2022patchcleanser, xiang2024patchcure}, aims to correctly predict the label of samples even if an adversarial patch is present.
On the other hand, certified detection \cite{patchcensor,li2022vip,mccoyd2020minority,han2021scalecert,xiang2021patchguard++,zhou2023majority,zhou2024crosscert} is dedicated to issuing warnings on patched harmful samples to achieve a higher certification coverage on benign samples.
Minority Report \cite{mccoyd2020minority} initially presents a mask-sliding mechanism for detection certification, but suffers from inadequate performance and poor scalability \cite{han2021scalecert}. 
Based on Minority Report, PatchGuard++ \cite{xiang2021patchguard++} upgraded the empirical masking strategy in PatchGuard \cite{xiang2021patchguard} as the guaranteed feature masking strategy to enhance the scalability. 
For the same purpose, ScaleCert \cite{han2021scalecert} introduces a neural network compression strategy, while it is still possible for its strategy to be broken shown by its paper, making it produce a theoretically weaker certification notion.
PatchCensor and ViP \cite{patchcensor, li2022vip} further refine the masking mechanism and respectively adopt ViT and MAE as their backbones to further improve the performances while they change the target from warning the harmful samples to warning the change of the prediction label from the benign sample. 
We have indirectly compared HiCert with these two defenders when we discussed and evaluated $D_\text{OMA}$ in the previous sections.
Demasked Smoothing \cite{yatsura2023certified} adopts $D_\text{OMA}$ in semantic segmentation tasks.
All of the methods mentioned above are masking-based detection and cannot certify any inconsistent sample; subsequently, a cross-checking-based detection, CrossCert \cite{zhou2024crosscert}, adopts two different types of certified recovery defenders and cross-checks the recovered labels for detection.
It offers unwavering certification, a novel type of guarantee
(which other certified detection defenders cannot offer).
However, it does not exceed PatchCensor and ViP in certified accuracy for detection in its experiment, and requiring to execute two recovery defenders makes CrossCert much slower than its peers.

To our knowledge, HiCert is the first certified defense framework that effectively mitigates the risk on inconsistent samples, provably detecting all their harmful samples. 
%
%
Incorrectly predicted samples are inherently inconsistent. 
In this connection, existing certified detection defenders \cite{li2022vip,patchcensor,han2021scalecert,zhou2024crosscert,mccoyd2020minority,xiang2021patchguard++} also face challenges in certifying incorrectly predicted benign samples in general.

\section{Conclusion}\label{sec:Conclusion}
We have introduced a novel framework, HiCert, capable of certifying both consistent and inconsistent samples, regardless of prediction correctness and the label distribution of mutants, with the guarantee of detecting all their harmful samples deterministically. 
We have presented theorems to prove the soundness of the framework. The design of HiCert can also raise alerts on incorrectly predicted samples if they are certified. 
To our knowledge, HiCert is the first work capable of this comprehensive coverage level for certification. Comprehensive experiments demonstrate its high effectiveness in certifying various kinds of benign samples and detecting all their harmful counterparts with high defense success ratios.

It is interesting to further develop HiCert to handle
real-time video analysis and processing by reducing the time cost and incorporating time-series context for certification and detection, for example. 
Future work also includes extending HiCert to certify samples for
other tasks (e.g., semantic segmentation, text classification) and attack types (e.g., few-pixel attacks), and further developing a comprehensive detection framework. 
It is also interesting to adapt HiCert's philosophy to certifably detect AI-generated fake samples.
Additionally, future research will explore the relationship between mutant confidence values across different patch shapes and sizes, base models, and datasets.
It is also interesting to integrate HiCert into the pipeline of probabilistic certified defense so that it provides a probabilistic guarantee for benign samples rejected by the current deterministic counterpart. 
%
%
Addressing correctly predicted samples with unchanged labels after patch attacks also remains an area for further research.
\bibliographystyle{IEEEtran}
\bibliography{references_short}

\clearpage
\appendices
\renewcommand{\appendixname}{Section}
\thispagestyle{empty}

\section{Definition of certified detection}\label{app:certified_detection}
There are different schools of thought on the definition of certified detection in the literature. 
Some works \cite{xiang2021patchguard++,han2021scalecert, mccoyd2020minority} aim to detect all harmful samples (i.e., inferring a violation of $f(x') = y_0$), while others \cite{patchcensor,zhou2024crosscert, li2022vip} aim to detect the change in the prediction label from the benign sample $x$ (i.e., inferring a violation of $f(x') = f(x)$). 
The latter kind underestimates the attacker's ability (e.g., the attacker may know $y_0$ in producing $x'$ \cite{wei2023simultaneously,wei2022adversarial}), making them insensitive to detecting those harmful samples without changing the prediction label perturbed from incorrectly predicted benign samples,
such that $f(x) = f(x') \neq y_0$.
Some works \cite{xiang2021patchguard++,han2021scalecert} further require a certified detection defender silent on certified samples, but some others \cite{yatsura2023certified} do not, or they \cite{zhou2024crosscert} only require a defender silent on a proper subset of their certified samples with the guarantee by making the defender compatible in semantics with certified recovery.
Our adopted definition goes for worst-case detection (detecting all harmful samples) with the least assumption on benign samples (leaving whether a benign sample should be in a particular detection state unspecified).

\section{A Case Study on the Ineffectiveness of Peers for Incorrectly Predicted Benign Samples}\label{app:motivate_case_study}

For incorrectly predicted ImageNet samples, we performed a case study 
with MAE and three different sizes of the patch (32, 64, 96 pixels) under the same experimental settings as the experiment for answering RQ1 in Section~\ref{sec:eva}.
PG++ with all five values (from 0.5 to 0.9) of $\tau$ in the experiment cannot certify any sample out of 
all 8751 incorrectly predicted samples in the ImageNet dataset,
and
$D_\text{OMA}$ can certify
only 1 sample (the file index is \texttt{n01751748/ILSVRC2012\_val\_00002154}).
\section{Special case of HiCert}\label{app:special_case_of_HiCert}

\begin{figure}[]
\centering
\includegraphics[width=\linewidth]
{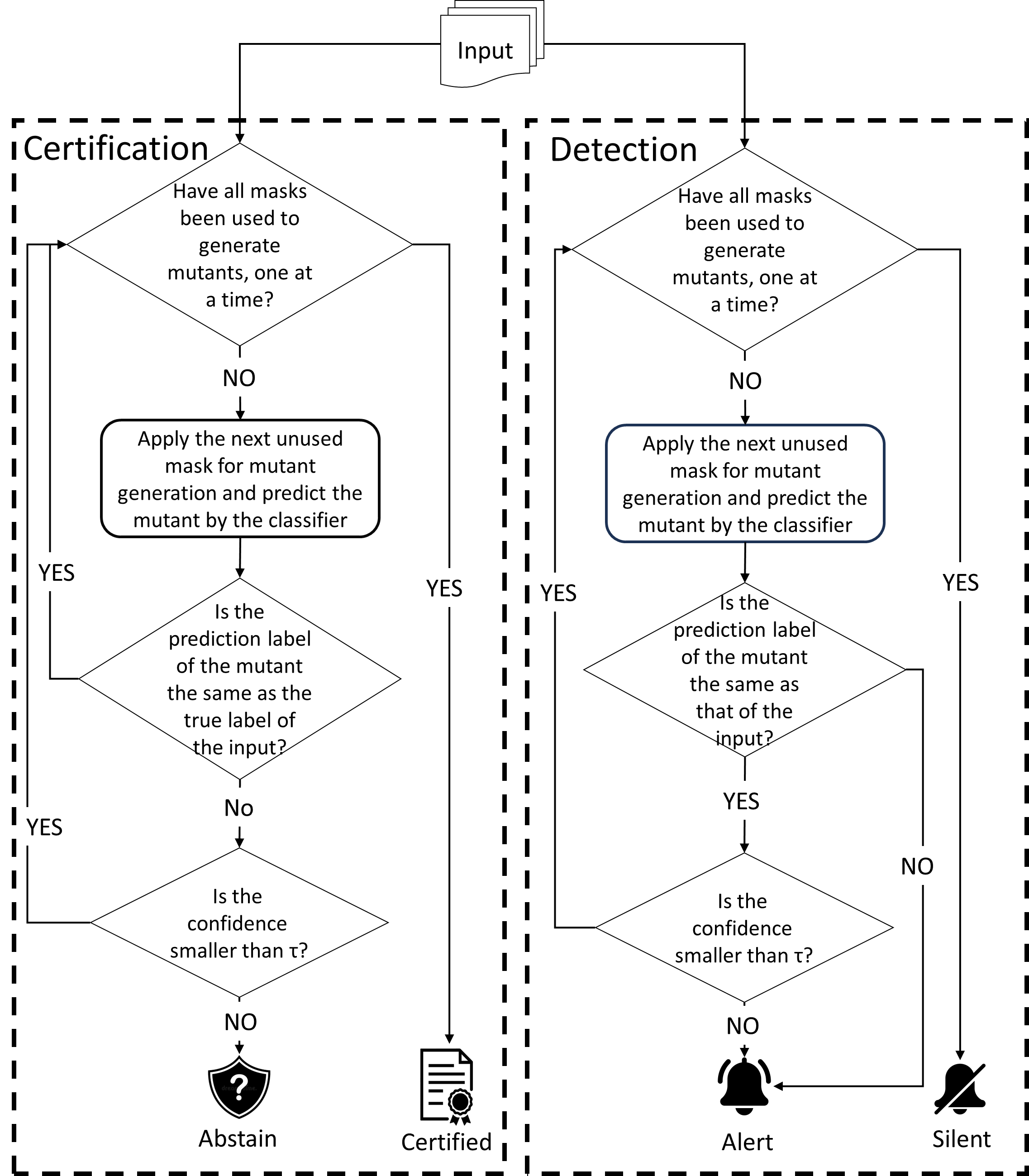}
\caption{The flowcharts of HiCert on certification and detection processes.
}\label{fig:flow}
\end{figure} 

When $\tau=0$, HiCert is reduced to $D_\text{OMA}$.
This is because  the set $\{f_{\textit{conf}}(x_\textsc{m})\mid \textsc{m}\in\mathbb{M}_\mathbb{P},f(x_\textsc{m})\neq y_0\}$ becomes empty ($\emptyset$)  if a given sample $x$ with the true label $y_0$ is consistent (i.e., $[\text{OMA}(x, y_0)=\textit{True}]$), thereby obtaining $v(x)=[\max \emptyset < 0] = [-\infty<0]=\textit{True}$.
Meanwhile, $w(\hat{x})$ is reduced to $[\{\hat{x}_\textsc{m}\mid \textsc{m}\in\mathbb{M}_\mathbb{P},f(\hat{x}_\textsc{m})\neq f(\hat{x})\}\neq \emptyset]$, where the set $\{\hat{x}_\textsc{m}\mid \textsc{m}\in\mathbb{M}_\mathbb{P},f(\hat{x}_\textsc{m})\neq f(\hat{x})\}$ becomes non-empty  only if the input sample $\hat{x}$ has a label difference, i.e., $[\text{OMA}(\hat{x}, f(\hat{x}))=\textit{False}]$.

On the other hand, when $\tau=1$, HiCert is reduced to a trivial detection defender that certifies every benign sample $x$ (i.e., $ [\max{\{f_{\textit{conf}}(x_\textsc{m})\mid \textsc{m}\in\mathbb{M}_\mathbb{P},f(x_\textsc{m})\neq y_0\}}<\tau]$ always holds) and warns every input sample $\hat{x}$ (i.e., $ [\min{\{f_{\textit{conf}}(\hat{x}_\textsc{m})\mid \textsc{m}\in\mathbb{M}_\mathbb{P},f(\hat{x}_\textsc{m})=f(\hat{x})\}}<\tau]$ always holds).

\section{
Flowcharts of HiCert
}~\label{sec:flow}
The flowcharts of HiCert are shown in Fig.~\ref{fig:flow}, illustrating two separate certification and detection processes in HiCert. In both flows, HiCert iteratively generates mutants and checks their prediction labels and confidence. 
For detection, if both conditions on the prediction label and confidence are met for all mutants, HiCert keeps silent on the input; otherwise, if any mutant violates either condition of prediction label or confidence, HiCert raises an alert.
For certification, if either condition on the prediction label and confidence is met for all mutants, HiCert certifies this input; otherwise, if any mutant violates both the conditions of the prediction label and confidence, HiCert abstains from certifying this input.

\section{Theorem, proof, and the dilemma of attackers} \label{app:proof}
\newtheorem*{unthm}{Theorem}

\begin{unthm}[Consistent mutants are infeasible places for attackers]
If the patch region is covered by a mask whose corresponding mutant's label is the same as the true label, it is infeasible for harmful samples to show no label difference. (i.e.,
if the condition $[\exists \textsc{m}_\textsc{p} \in\mathbb{M}_\mathbb{P},
\textsc{m}_\textsc{p}\odot{\textsc{p}}=\textsc{p}
\land 
f(x_{\textsc{m}_\textsc{p}})=y_0]
$
holds, the condition 
$[\forall x'\in\{x'\mid{x}'=(\textsc{J}-\textsc{p})\odot {x}+\textsc{p}\odot {x}'\},
[f(x')\neq y_0]\implies
[\exists \textsc{m} \in \mathbb{M}_\mathbb{P}, 
f({x}'_{\textsc{m}})  \neq f(x')]]$
holds.)

\end{unthm}
\begin{proof}

By 
$[\textsc{m}_\textsc{p}\odot{\textsc{p}}=\textsc{p}]$, we know
$x'_{\textsc{m}_\textsc{p}}
=
((\textsc{J}-\textsc{p})\odot {x}+\textsc{p} \odot{x}')\odot(\textsc{J}-\textsc{m}_\textsc{p})
=
((\textsc{J}-\textsc{p})\odot {x}+\textsc{p} \odot{x}')-((\textsc{J}-\textsc{p})\odot {x}+\textsc{p} \odot{x}')\odot\textsc{m}_\textsc{p}
=
(\textsc{J}-\textsc{p})\odot {x}-(\textsc{J}-\textsc{p})\odot {x}\odot\textsc{m}_\textsc{p}
=
(\textsc{J}-\textsc{p})\odot {x}\odot(\textsc{J}-\textsc{m}_\textsc{p})
=
{x}\odot(\textsc{J}-\textsc{m}_\textsc{p})
=
x_{\textsc{m}_\textsc{p}}
$ (see Fig.~\ref{fig:concept} for illustration).
Here we also know 
$
f(x_{\textsc{m}_\textsc{p}})=y_0
$.
So, we have
$
f(x'_{\textsc{m}_\textsc{p}})=y_0
$.
Further, if 
$f(x')\neq y_0$, we know 
[$\exists \textsc{m} \in \mathbb{M}_\mathbb{P}, 
f({x}'_{\textsc{m}})  \neq f(x')$].%
\end{proof}

\begin{unthm}[HiCert Certification]
If the maximum confidence of inconsistent mutants of a benign sample $x$ is below a threshold $\tau$,
  each harmful sample $x'$ either incurs a label difference or has mutant(s) with minimum confidence below $\tau$ that are predicted with a label the same as $x'$
 --- if the condition 
$[\max{\{f_{\textit{conf}}({x}_\textsc{m})\mid \textsc{m}\in\mathbb{M}_\mathbb{P},f({x}_\textsc{m})\neq y_0\}}<\tau]$~holds, the condition $[\forall x'\in\mathbb{A}_\mathbb{P}({x}),[f(x')\neq y_0]
\implies 
[\{{x}'_\textsc{m}\mid  \textsc{m}\in\mathbb{M}_\mathbb{P},f({x}'_\textsc{m})\neq f({x'})\}\neq \emptyset] 
 \lor 
 [\min{\{f_{\textit{conf}}({x}'_\textsc{m})}$ ${\mid \textsc{m}\in\mathbb{M}_\mathbb{P},f({x}'_\textsc{m})=f({x'})\}}<\tau]]$ holds, 
 which is
 $v(x)\implies [\forall x'\in\mathbb{A}_\mathbb{P}({x}), f(x')\neq y_0 \implies w(x')]$ in HiCert.
\end{unthm}
\begin{proof}
Recall that $\textsc{m}_\textsc{p}$ denotes the mask in the covering mask set $\mathbb{M}$ covering the patch in a harmful version $x'$ of $x$ (i.e., $[\textsc{m}_\textsc{p}\odot{\textsc{p}}=\textsc{p}]$).
We get $x_{\textsc{m}_\textsc{p}} ={x}'_{\textsc{m}_\textsc{p}}$ (see the proof of Thm.~\ref{thm:infeasible} above and Fig.~\ref{fig:concept} for illustration).
Case 1: 
Suppose 
$f(x_{\textsc{m}_\textsc{p}})\neq y_0$.
We know $[\max\{f_{\textit{conf}}({x}_\textsc{m})\mid \textsc{m}\in\mathbb{M}_\mathbb{P},f({x}_\textsc{m})\neq y_0\}<\tau]$, 
which means $f_{\textit{conf}}(x_{\textsc{m}_\textsc{p}})<\tau$
and $f_{\textit{conf}}(x'_{\textsc{m}_\textsc{p}})<\tau$.
Sub-Case 1.1:
Suppose $[\{{x'}_\textsc{m}\mid  \textsc{m}\in\mathbb{M}_\mathbb{P},f({x}'_\textsc{m})\neq f({x'})\}\neq\emptyset]$ does not hold, which means $f({x}'_{\textsc{m}_\textsc{p}})=f({x'})$.
Therefore,
$ [\min\{f_{\textit{conf}}({x}'_\textsc{m})\mid \textsc{m}\in\mathbb{M}_\mathbb{P},f({x}'_\textsc{m})=f({x'})\}<\tau]$ holds.
Sub-Case 1.2: 
$[\{{x'}_\textsc{m}\mid  \textsc{m}\in\mathbb{M}_\mathbb{P},f({x}'_\textsc{m})\neq f({x'})\}\neq\emptyset]$ holds.
Case 2 (Thm.~\ref{thm:infeasible}): 
Suppose $f(x_{\textsc{m}_\textsc{p}})=y_0$. 
We have $f(x'_{\textsc{m}_\textsc{p}})=f(x_{\textsc{m}_\textsc{p}})=y_0$. 
Recalled that $x'$ is harmful, we should have $f(x')\neq y_0$ and then $f(x'_{\textsc{m}_\textsc{p}})\neq f(x')$, thereby
$[\{{x'}_\textsc{m}\mid \textsc{m}\in\mathbb{M}_\mathbb{P},f({x}'_\textsc{m})\neq f({x'})\}\neq \emptyset]$ holds.
\end{proof}

By designing a warning function that follows Thm.~\ref{thm:Inconsistent-Max-Min},
HiCert places attackers in a \emph{dilemma} if they attempt to create a harmful sample $x'$ of any benign sample $x$ with $v(x)=\textit{True}$ and aim to make HiCert silent on their created harmful samples.
All these attempts of the attackers must be failed by HC since it can consistently alert on all these harmful samples.
\begin{figure}[] 
\centering
\includegraphics[width=\linewidth]{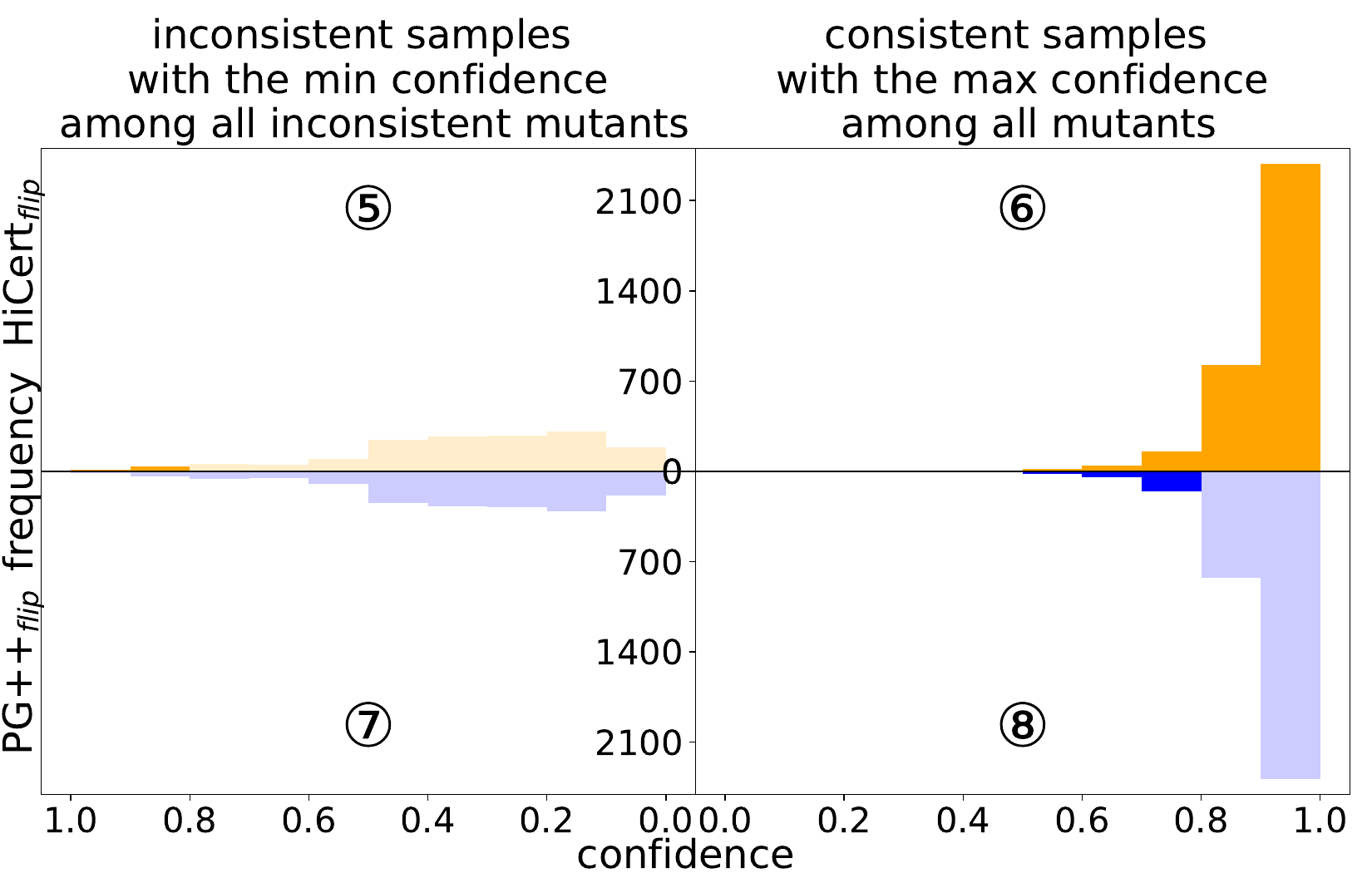}
\caption{
The plots show the maximum and minimum confidences among those mutants of the same sample for those samples out of all samples, as stated in the column headings.
}
\label{fig:confidence_min_max_ablation}

\end{figure}
\begin{description}
 
    \item[\textbf{Case 1 of the Dilemma}:] 
     Suppose that an adversarial patch is placed within a given mask and its corresponding mutant of $x$ (and also $x'$) is inconsistent. Then, the confidence of this mutant of $x$ must be lower than $\tau$.
     If the harmful sample and all its mutants share the same prediction label (following Sub-Case 1.1 in the proof), then this mutant should be predicted with a label the same as the harmful sample, which means the minimum confidence of these mutants is lower than $\tau$ and the warning function of HiCert will return \textit{True}.
     Otherwise, any mutants of the harmful sample that are predicted with a label different from the harmful sample (following Sub-Case 1.2 in the proof) indicate a label difference. So, the warning function of HiCert will also return \textit{True}.
     

\item[\textbf{Case 2 of the Dilemma}:] 
Suppose that an adversarial patch is placed within a given mask and its corresponding mutant of $x$ (and also $x'$) is consistent.
In this case, this mutant should be different from the harmful sample in the prediction label (following Case 2 in the proof).
Like Case 1, the warning function of HiCert will also return \textit{True}.
\end{description}


The two cases in the dilemma correspond to the two main cases in the proof of Thm. \ref{thm:Inconsistent-Max-Min}, respectively, which also correspond to the certification-warning paths depicted in Fig. \ref{fig:overview}:
\textcircled{4}-\textcircled{5}-\textcircled{1}-\textcircled{8} (for Sub-Case 1.1), 
\textcircled{4}-\textcircled{5}-\textcircled{2}-\textcircled{7} (for Sub-Case 1.2),
and 
\textcircled{4}-\textcircled{6}-\textcircled{2}-\textcircled{7} (for Case 2).
Suppose both inconsistent benign samples in Fig. \ref{fig:why_fail} satisfy the antecedent of the implication relation in Thm. \ref{thm:Inconsistent-Max-Min}. In this case, the first ($x'_{1-1}$) and the last ($x'_{2-3}$) harmful samples in the figure will be detected by the label difference condition through the Path \textcircled{4}-\textcircled{6}-\textcircled{2}-\textcircled{7}, and the remaining three in between will be detected by the low confidence property through the Path \textcircled{4}-\textcircled{5}-\textcircled{1}-\textcircled{8}.
\thispagestyle{empty}
\section{A case study on the effectiveness of the design of HiCert}\label{app:case_study_on_design}

We analyze 5000 randomly selected ImageNet samples with their mutants from the experiment with MAE as the base model and the patch size of 32 pixels. Other settings are the same as answering RQ1 in Section \ref{sec:eva}. 
We denote the set of 5000 samples by the set $D$.
We split this set of samples $D$ into two subsets: they contain the samples with and without inconsistent mutants (i.e., inconsistent samples and consistent samples), respectively, and are denoted by sets $D_1$ and $D_2$, respectively.
We compute the maximum and minimum confidences among the confidences of all inconsistent mutants of the same sample for all samples in $D_1$ and these two bounds among the confidences of all (consistent) mutants of the same sample for all samples in $D_2$, to study the two types of confidence bounds on all samples in $D$, denoted by $\max(x)_1$ and $\min(x)_1$ for $x \in D_1$ and $\max(x)_2$ and $\min(x)_2$ for $x \in D_2$, respectively.
The two columns of histograms in Fig.~\ref{fig:confidence_min_max} from left to right correspond to the distributions for $\max(x)_1$, $\min(x)_2$ for all samples $x$ in respective sub-datasets $D_1$ and $D_2$.

Under the same setting, we also conduct an ablation study.
We have further constructed two additional defenders PG++$_{\textit{flip}}$ and HiCert$_{\textit{flip}}$ to pair with PG++ and HiCert, respectively, by replacing the ``$>$'' operator with the ``$<$'' operator in PG++ for PG++$_{\textit{flip}}$ and replacing the ``$<$'' operator with the ``$>$'' operator in HiCert for HiCert$_{\textit{flip}}$,
to demonstrate the inability of PG++$_{\textit{flip}}$ and the ineffectiveness of HiCert$_{\textit{flip}}$ to certify inconsistent samples (the meaning behind the direction of inequality sign of PG++$_{\textit{flip}}$ is to require low confidence on consistent mutants; on the contrary, that of inequality sign of HiCert$_{\textit{flip}}$ is to require high confidence on inconsistent mutants. Both are counterintuitive.). 
For clarification, their certification functions are
$v(x) := 
  [\text{OMA}(x,y_0) = \small{\textit{True}} \land \forall \textsc{m} \in \mathbb{M}_\mathbb{P}, f_{\textit{conf}}(x_\textsc{m})<\tau]$
  for PG++$_{\textit{flip}}$
  and
  $v(x):= 
 [\min{\{f_{\textit{conf}}(x_\textsc{m})\mid \textsc{m}\in\mathbb{M}_\mathbb{P},f(x_\textsc{m})\neq y_0\}}>\tau]$
  for HiCert$_{\textit{flip}}$.
Their results are shown in sub-figures \textcircled{5}--\textcircled{8} in Fig.~\ref{fig:confidence_min_max_ablation}.
The two columns of histograms in Fig.~\ref{fig:confidence_min_max_ablation} from left to right correspond to the distributions for $\min(x)_1$ and $\max(x)_2$ for all samples $x$ in respective sub-datasets $D_1$ and $D_2$.
The bars for samples certified by the corresponding defenders (see the labels for the $y$-axis) are displayed in a solid color; otherwise, they are semi-transparent, where the confidence threshold $\tau$ is set to 0.8 for illustration purposes.
{Histograms \textcircled{5}--\textcircled{6} and \textcircled{7}--\textcircled{8} represent HiCert$_{\textit{flip}}$ and PG++$_{\textit{flip}}$}, respectively.
Although HiCert$_{\textit{flip}}$ can certify all consistent samples (Histogram \textcircled{6}), it needs a small $\tau$ to cover a majority of all inconsistent samples to be effective in certifying these samples.
PG++$_{\textit{flip}}$ cannot certify any inconsistent samples (Histogram \textcircled{7}). 
Note that the ability for HiCert$_{\textit{flip}}$ to certify inconsistent mutants is due to the advancements achieved by HiCert.
The study shows that modifying the defenders to have the ability to certify both consistent benign samples and inconsistent benign samples (even in part) effectively is nontrivial.


\section{Discussion on the design of HiCert}
\label{app:discussion-on-HiCert-design}

\subsubsection{Achieving the same time complexity as $D_\text{OMA}$}
In terms of time complexity, compared to the $D_\text{OMA}$ defender, HiCert only additionally requires a constant-time check on selective mutants' confidence against $\tau$ during certification or warning detection and there are $|\mathbb{M}|$ mutants in total.
Since  $D_\text{OMA}$ already generates these $|\mathbb{M}|$ mutants and uses them for label prediction (and has to assess the confidence for the prediction label of mutants), 
 HiCert is the same as $D_\text{OMA}$ in time complexity, where $|\mathbb{M}|$ is a small constant in practice (e.g., $|\mathbb{M}|$ is 36 in our main experiments).

\subsubsection{Soundness and completeness}
In terms of theoretical soundness and completeness, like all other certified defenders
with deterministic guarantees, HiCert is sound in certifying benign samples without any false positives (i.e., if a sample is reported as certified by HiCert, all of its harmful samples should be warned by HiCert, proven by Thm.~\ref{thm:Inconsistent-Max-Min}).
However, it is incomplete in certifying benign samples, with some false negatives (i.e., all harmful samples of some benign samples will always be detected by HiCert but HiCert doesn't report these benign samples as certified), because the actual situations of the mutants of each harmful sample of a benign sample may not be as bad as the worst-case scenario analyzed by HiCert (e.g., 
{the prediction label of a mutant controlled by attackers may always be unable to be changed as the attackers want}) if HiCert cannot certify it.
This problem is shared by all existing certified defenders (both for recovery and for detection) and
we are also unaware of any certified detection defender (including  \cite{li2022vip,patchcensor,han2021scalecert,zhou2024crosscert,mccoyd2020minority,xiang2021patchguard++}) 
that is complete in certification unless the defender is a trivial one (warning all samples).

\thispagestyle{empty}
\section{Details of the Experimental setup}\label{app:setup}
\subsection{Environment}
We train the base models and generate mutants with their predicted label and confidence on GPU clusters with NVIDIA V100 GPUs. Data analysis is done on an Ubuntu 20.04 machine with Intel Xeon 6136 CPUs and NVIDIA 2080Ti GPUs.

\subsection{Datasets}
We adopt ImageNet \cite{deng2009imagenet} (widely used to evaluate patch robustness certification \cite{xiang2021patchguard,xiang2022patchcleanser, patchcensor,zhou2024crosscert,salman2022certified,li2022vip,han2021scalecert,xiang2021patchguard++}), CIFAR100 \cite{krizhevsky2009learning}, and 
GTSRB \cite{Stallkamp-IJCNN-2011} as our datasets.

 These three datasets cover various applications, complexity, scale, and number of classes.
ImageNet contains 1.3 million training images and 50,000 validation images for 1,000 classes. CIFAR100 contains 50,000 training images and 10,000 test images for 100 classes. GTSRB contains 39,209 training images and 12,630 test images. 

We download ImageNet from \url{image-net.org}, use its entire training set for fine-tuning, and regard its validation set as the test set for evaluation. We download CIFAR100 and GTSRB from torchvision and use their whole training sets for fine-tuning and their test sets for evaluation. All images are resized to $224 \times 224$ in our experiments.


\subsection{Baselines}
\label{sec:baselines}
Unlike previous work focusing on a single model architecture in evaluation (ViP \cite{li2022vip} on MAE, PatchCensor \cite{patchcensor} on ViT, and PatchGuard++ \cite{xiang2021patchguard++} on BagNet), our aim is to achieve technological versatility in evaluation. Therefore, we adopt all Masked Autoencoders \cite{he2022masked} (vit-mae-base with 112M parameters, denoted as \textbf{MAE}), Vision Transformer \cite{dosovitskiy2021an} (vit-b16-224 with 86.6M parameters, denoted as \textbf{ViT}) and ResNet \cite{he2016deep} (resnet-50 with 25.5M parameters, denoted as \textbf{RN}), as the architectures of the base models of defenders. We also use a model-agnostic pixel-level masking strategy following PatchCleanser \cite{xiang2022patchcleanser}
(PatchCleanser \cite{xiang2022patchcleanser} is a certified recovery defender, which is a follow-up work of PatchGuard++ by the same first author)
and CrossCert \cite{zhou2024crosscert}
instead of creating model-specific masking (channel masking for ViT/MAE \cite{li2022vip, patchcensor}, feature masking for BagNet \cite{xiang2021patchguard++}). We follow the principle in PatchCleanser \cite{xiang2022patchcleanser} to generate a covering mask set for each patch size.

We adopt the architectures and pre-trained weights from \url{https://github.com/
facebookresearch/mae} for MAE, \url{https://huggingface.co/timm/vit_base_patch16_224.augreg2_in21k_ft_in1k} for ViT, and \url{https://huggingface.co/timm/resnetv2_50x1_bit.goog_distilled_in1k} for RN.
We fine-tune MAE for each dataset by the original script from \url{https://github.com/
facebookresearch/mae}.
For fine-tuning ViT and RN, we use SGD with a momentum of 0.9, set the batch size to 64, and the number of epochs to 10, reducing the learning rate by a factor of 10 after every 5 epochs. 
Table~\ref{tab:clean_acc} shows the clean accuracy of different base models for the datasets. Since MAE is the state-of-the-art (SOTA) base model \cite{li2022vip}, we use MAE as the default (main) base model in reporting the results of our evaluation.

We compare top-performing certified detection defenders implemented in our infrastructure to HiCert (\textbf{HC}): $\textbf{\textit{D}}_\textit{\text{OMA}}$ and \textbf{PG++} \cite{xiang2021patchguard++}. 
Recall that $D_\text{OMA}$ shares the common $D_\text{OMA}$ checking strategy with ViP \cite{li2022vip} and PatchCensor \cite{patchcensor} but aims to detect $f(x')\neq y_0$ rather than $f(x')\neq f(x)$.
With the same base model and the same masking strategy in our infrastructure,
\textbf{ViP} \cite{li2022vip} and PatchCensor (\textbf{PC}) \cite{patchcensor} must share the same certified accuracy and clean accuracy with $D_\text{OMA}$ since each sample $x$ counted by these two metrics satisfies $f(x)= y_0$, and then the condition OMA$(x, f(x)) \land f(x)= y_0$ for ViP and PC is equivalent to the condition OMA$(x, y_0)$ for $D_\text{OMA}$ in certification functions.
Their certified ratios $r_{cert}$ are the same as their certified accuracy $acc_{cert}$ (which is also shared with 
 $D_\text{OMA}$) because they cannot provide any warning guarantee for those incorrectly predicted benign samples in the situation where $f(x')\neq y_0 \land f(x')=f(x)$. 
The lower section of Table~\ref{tab:main_eva_results} summarizes these results.

We further compare HiCert with more state-of-the-art certified detection defenders, and mark them with the symbol $\star$: ScaleCert (\textbf{SC$_\star$}) \cite{han2021scalecert}, PatchGaurd++ (\textbf{PG++$_\star$}) \cite{xiang2021patchguard++}, Adapted Minority Reports (\textbf{MR+$_\star$}) \cite{patchcensor}, PatchCensor (\textbf{PC$_\star$}) \cite{patchcensor}, ViP (\textbf{ViP$_\star$}) \cite{li2022vip}, and CrossCert (\textbf{CC$_\star$}) \cite{zhou2024crosscert} based on the results reported in the literature. 
Their results are summarized in the upper section of Table~\ref{tab:main_eva_results}.

\subsection{Metrics}
Our evaluation will use certified accuracy, certified ratio, and certified ratio for inconsistent samples as the main metrics.

Suppose $x$ is a benign sample with the true label $y_0$ in a test dataset $\mathbb{S}$ that only contains benign samples.
Previous works use two key metrics, \textbf{clean accuracy}, to evaluate the inherent classification capability of the base model, and \textbf{certified accuracy}, to evaluate the certification ability of a defender on correctly predicted samples, which are defined as 
$acc_{\textit{clean}}=\frac{\mid\{{x}\in\mathbb{S}\mid f({x})=y_0\}\mid}{\mid\mathbb{S}\mid}$
and
$acc_{\textit{cert}}=\frac{\mid\{{x}\in\mathbb{S}\mid f({x})=y_0\land v({x})=\textit{True}\}\mid}{\mid\mathbb{S}\mid}$ \cite{zhou2024crosscert},
despite some work \cite{mccoyd2020minority,xiang2021patchguard++} excluding all benign samples that were warned by the defender concerned as elements in the set in the numerator, and some others (including the present paper) \cite{patchcensor, li2022vip, zhou2024crosscert} including them. 
However, $acc_{\textit{cert}}$ discounts the certification ability of a defender on incorrectly predicted samples and cannot reflect the ability to certify inconsistent samples.
So, we also measure the \textbf{certified ratio} $r_{cert}=\frac{\mid\{{x}\in\mathbb{S}\mid v({x})=\textit{True}\}\mid}{\mid\mathbb{S}\mid}$, which counts all certified samples in $\mathbb{S}$, regardless of correct or incorrect predictions, and the
\textbf{certified ratio for inconsistent samples}
$r_{\textit{cert}_{\textit{inc}}}=\frac{\mid\{{x}\in\mathbb{S}\mid v({x})=\textit{True}\land \text{OMA}({x},y_0)=\textit{False}\}\mid}{\mid\{{x}\in\mathbb{S}\mid \text{OMA}({x},y_0)=\textit{False}\}\mid}$, which counts the proportion of inconsistent samples that are certified.

Table~\ref{tab:case} shows all eight combinations of three conditions on a benign sample: whether the sample is correctly predicted, whether it is warned, and whether it is certified, where a check symbol $\checkmark$ represents the corresponding condition is evaluated as true.
We measure all these combinations on $\mathbb{S}$
to facilitate our detailed analysis case by case.
%

Fig.~\ref{fig:application} has two outgoing paths after certified detection. 
For the silent path, we measure
the \textbf{silent accuracy}
$
acc_{\neg w}=\frac{\mid\{{x}\in\mathbb{S}\mid w({x})=\textit{False}\land f({x})=y_0\}\mid}{\mid\{{x}\in\mathbb{S}\mid w({x})=\textit{False}\}\mid}$, the accuracy on the set of benign samples without warnings triggered,
and for the alert path, we measure
 \textbf{false alert ratio} $r_{\textit{fa}}=\frac{\mid\{{x}\in\mathbb{S}\mid w({x})=\textit{True}\land f({x})=y_0\}\mid}{\mid\{{x}\in\mathbb{S}\mid f({x})=y_0\}\mid}$  \cite{yatsura2023certified}, the fraction of correctly predicted samples for which a defender returns a warning alert, where
having a higher value in $r_{\textit{fa}}$ may make the system waste more additional cost on these correctly predicted samples.
Additionally, we measure the \textbf{false silent ratio} $r_{\textit{fs}}$, 
the fraction of incorrectly predicted samples for which we do not return an alert:
$r_{\textit{fs}}=\frac{\mid\{{x}\in\mathbb{S}\mid w({x})=\textit{False}\land f({x})\neq y_0\}\mid}{\mid\{{x}\in\mathbb{S}\mid f({x})\neq y_0\}\mid}$.
A higher $r_{\textit{fs}}$ value signifies an increased number of incorrectly predicted samples, posing a greater threat to downstream operations.
Note that all these metrics only measure the warning aspect of benign samples for readers to gain a deeper understanding of the warning ability of defenders.
The number of warnings on benign samples cannot represent the warning ability of certified defenders on the whole input domain, which also includes non-benign samples that cannot be exhausted.
However, the application scenario in Fig.~\ref{fig:application} naturally requires a high proportion of samples that are correct and silent,
identifying harmful (incorrectly predicted) samples and minimizing false warnings.
We use them as secondary metrics to supplement the primary ones ($acc_{cert}$, $r_{cert}$ and $r_{suc}$).
To answer RQ2, our experiment will generate actual samples to attack the defender.
We compare the \textbf{defense success ratio} $r_{suc}=\frac{|\{{x}\in\mathbb{S}_{sub}\mid \forall {x}'\in \mathbb{A}^{act}_{\mathbb{P}}({x}), f({x}')\neq y_0\implies w({x}')=\textit{True}\}|}{|\{\mathbb{S}_{sub}\}|}$ between defenders, where $\mathbb{S}_{sub}$ is a subset of $\mathbb{S}$ used by an actual attacker tool as seed input, $\mathbb{A}^{act}_{\mathbb{P}}({x})$ is a subset of $\mathbb{A}_{\mathbb{P}}({x})$ generated by the actual attacker tool. This metric
measures the proportion of benign samples for which all harmful samples generated by an attacker tool are detected by the defender. 
(If not all harmful samples generated by an attacker tool on a benign sample are detected by the defender, the attack is called a \emph{success attack} on the defender.)
Unlike $acc_{cert}$ and $r_{cert}$ for theoretical defense ability, $r_{suc}$ shows empirical defense ability against real adversarial patch attacks.

Except for $r_{\textit{fa}}$ and $r_{\textit{fs}}$, 
higher values for all other metrics indicate better quality.

\thispagestyle{empty}
\subsection{Experimental Setting}\label{sec:experimental_setting}
{In this section, we describe the procedure of the experiments.}


For RQ1, we follow the common practice in the evaluation of certified detection defenders to perform it on the benign samples \cite{xiang2021patchguard++,li2022vip,patchcensor,han2021scalecert,zhou2024crosscert}, with the previously adopted patch size to compared $acc_{clean}$ and $acc_{cert}$: 32 pixels (2\%) in ImageNet \cite{xiang2021patchguard++,li2022vip,patchcensor,han2021scalecert,zhou2024crosscert}, 35 pixels (2.4\%) in CIFAR100 \cite{zhou2024crosscert}, and 32 pixels (2\%) in GTSRB \cite{patchcensor}.
We vary $\tau$ from 0.5 to 0.9 (previous work \cite{xiang2021patchguard++} chooses a similar range).
The results of MAE are shown in Table.~\ref{tab:clean_acc} and more results of other base models are shown in Fig.~\ref{fig:vary_patch_size}.
We also perform a detailed analysis on benign samples of ImageNet with patch size 32 pixels with MAE, to also check $r_{cert}$ and $r_{cert_{inc}}$ for the certification ability, and check $acc_{\neg w}, r_{fa},$ and $f_{fs}$ for warning ability on benign samples as secondary metrics.

For RQ2, we perform an actual adversarial patch attack adopted from \cite{levine2020randomized}, which is gradient-based (using the base model $f$ as the surrogate model to attain the gradient following \cite{levine2020randomized}) and has no knowledge of defenders (HC/$D_\text{OMA}$/PG++) for the fairness.

Specifically, we select the first 500 benign samples from each shuffled dataset for the attack and set 80 random starts, 150 iterations per random start, and a step size of 0.05.
We set patch sizes to 32, 64, and 96 pixels. 
For each patch size, we evaluate the warning function after each iteration for each defender with each $\tau$ from 0.5 to 0.9. 
If any harmful sample of a selected sample passes through undetected, the defender with that $\tau$ value is marked as having failed for the selected sample.
Due to the scale of the experiment, we limit the evaluation of defenders to the most representative base model (MAE).
We note that the defense success ratios are calculated based on the actual outcomes of the warning functions of the defenders.
Since $D_\text{OMA}$, $\text{PC}$, and $\text{ViP}$ share the same warning function, their defense success ratios are the same.

For RQ3, we follow \cite{patchcensor,li2022vip}
to vary the patch size from 16 to 112 pixels, step by 16 pixels. $\tau$ is set to 0.8 in both PG++ and HC, and the results of PC and ViP are the same as those of $D_\text{OMA}$.

\thispagestyle{empty}
\section{Other results in RQ1}\label{app:other_result}
We also summarize the results on ImageNet with ViT and RN for the same patch size (2\%, 32 pixels). 
 The values of $acc_{cert}$ and $r_{cert}$ are almost identical
($D_\text{OMA}$ can only certify  \emph{nine} incorrectly predicted samples based on RN but the number is too small to be discernible by comparing $r_{cert}$ with $acc_{cert}$, and cannot certify \emph{any} such samples in the other three combinations of defenders and base models) for PG++ ($\tau=0.8$) with 45.5 (ViT) and 57.8 (RN) and
for $D_\text{OMA}$, ViP, and PC with 64.9 (ViT) and 55.9 (RN). 
These three defenders are zeros in $r_{cert_{inc}}$ for both ViT and RN.
For HC ($\tau=0.8$), the values of $acc_{cert}$ are 70.1 (ViT) and 73.6 (RN), the values of $r_{cert}$ are 72.7 (ViT) and 74.7 (RN), and the values of $r_{cert_{inc}}$ are 16.8 (ViT) and 9.8 (RN).
Their trends and comparisons are similar to our reported MAE results.
Also, $r_{cert}$ of HiCert is always higher than $D_\text{OMA}$ and PG++ for all combinations of $\tau\in[0.5,0.9]$, the three datasets, and
the three base models.

\begin{table}[]
\caption{Results on ImageNet samples by different modes of patch in total 1\% patch area }\label{tab:different_mode}\centering
\resizebox{\linewidth}{!}
{
\begin{tabular}{|c|ccc|ccc|}
\hline
{Config of Patch} & \multicolumn{3}{c|}{Certification}                                  & \multicolumn{3}{c|}{Secondary   Metrics}                    \\ \cline{2-7} 
                         Area total in 
                         1\% & \multicolumn{1}{c|}{$acc_{cert}$} & \multicolumn{1}{c|}{$r_{cert}$} & {$r_{cert_{inc}}$} & \multicolumn{1}{c|}{$acc_{\neg w}$} & \multicolumn{1}{c|}{$r_{fa}$}  & $r_{fs}$ \\ \hline
one square                    & \multicolumn{1}{c|}{82.0}    & \multicolumn{1}{c|}{94.1} & 69.0     & \multicolumn{1}{c|}{97.5} & \multicolumn{1}{c|}{47.1} & 6.4 \\ \hline
one rectangle                & \multicolumn{1}{c|}{81.1}    & \multicolumn{1}{c|}{92.2} & 62.6     & \multicolumn{1}{c|}{98.2} & \multicolumn{1}{c|}{55.6} & 3.8 \\ \hline
two squares                & \multicolumn{1}{c|}{80.2}    & \multicolumn{1}{c|}{90.1} & 56.0     & \multicolumn{1}{c|}{98.3} & \multicolumn{1}{c|}{61.2} & 3.2 \\ \hline
\end{tabular}
}
\end{table}

\begin{table}[]\caption{Number of mask (mutants) vs Runtime per sample in HiCert}\label{tab:number_of_mask}\centering
\begin{tabular}{cccccc||c}
\hline
\multicolumn{1}{|c|}{Number of mask (mutants)}          & \multicolumn{1}{c|}{$6^2$}  & \multicolumn{1}{c|}{$5^2$}    & \multicolumn{1}{c|}{$4^2$}   & \multicolumn{1}{c|}{$3^2$}  & \multicolumn{1}{c}{$2^2$}  & \multicolumn{1}{||c|}{0}    \\ \hline
\multicolumn{1}{|c|}{runtime per   sample (ms)} & \multicolumn{1}{c|}{155} & \multicolumn{1}{c|}{106} & \multicolumn{1}{c|}{66} & \multicolumn{1}{c|}{39} & \multicolumn{1}{c}{20} & \multicolumn{1}{||c|}{4} \\ \hline
\end{tabular}
\end{table}

\begin{figure}[] 
\centering
\includegraphics[width=\linewidth]{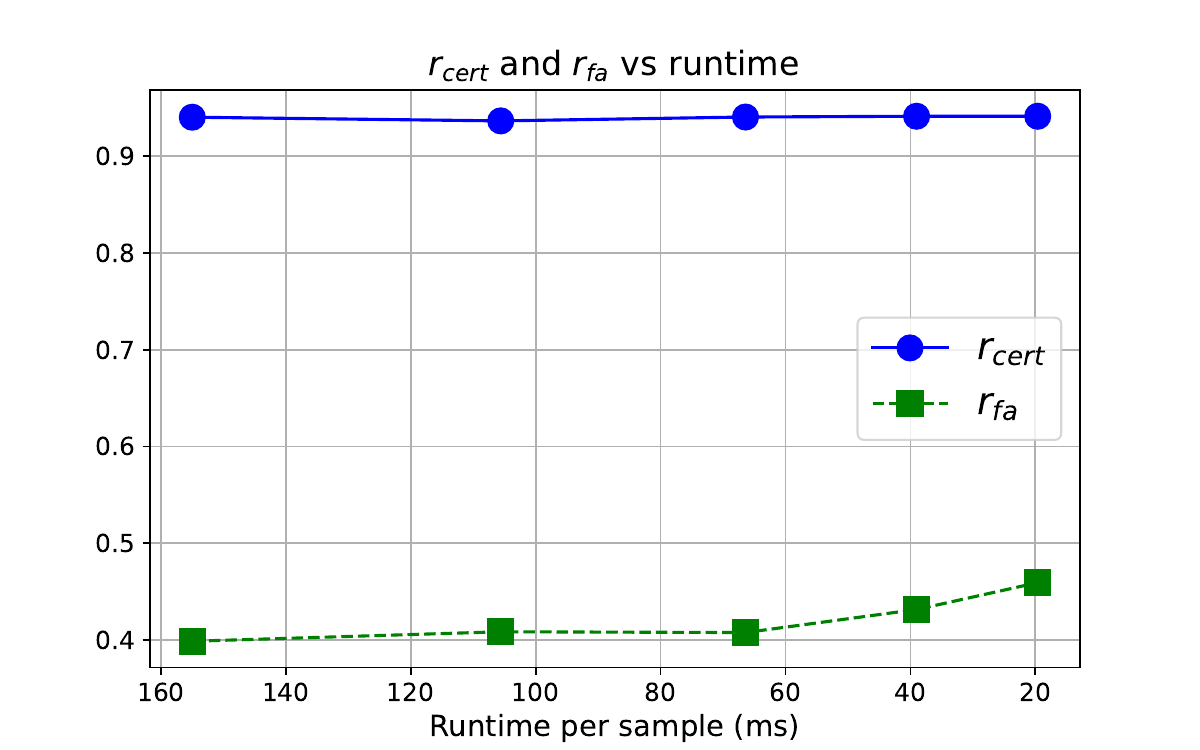}

\caption{Trade-off between $r_{cert}$/$r_{fa}$ and runtime by varying the size of the covering mask set.}\label{fig:runtime}
\end{figure}

\begin{figure}[] 
\centering
\includegraphics[width=\linewidth]{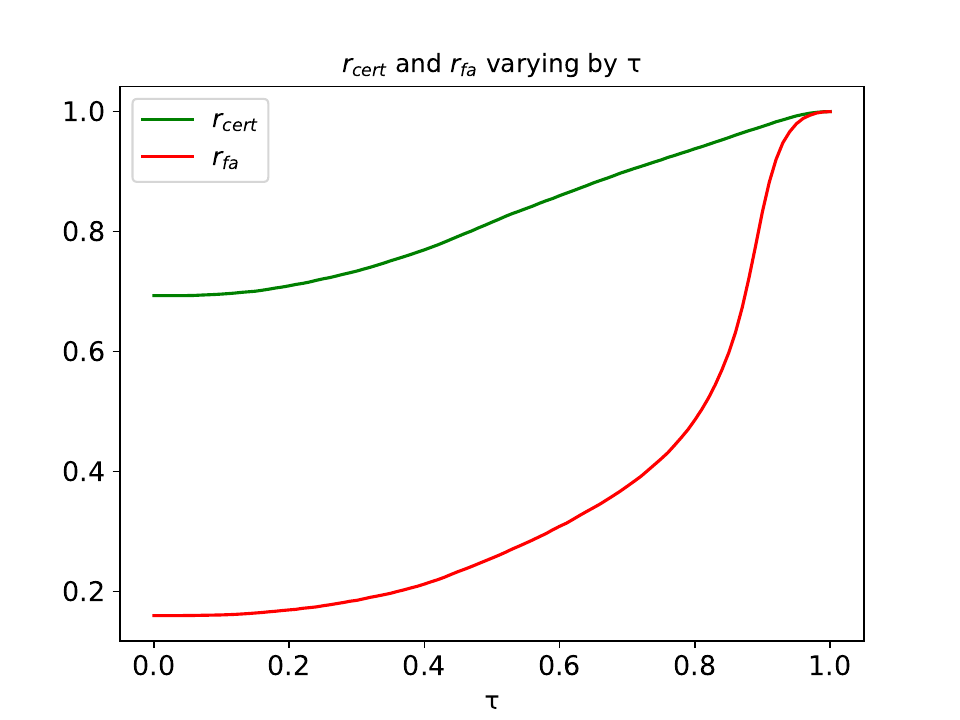}
\caption{Trade-off between $r_{cert}$ and $r_{fa}$ by varying $\tau$,
}\label{fig:visual_t}
\end{figure}

 \begin{figure*}[] 
\centering
\includegraphics[width=0.32\linewidth]{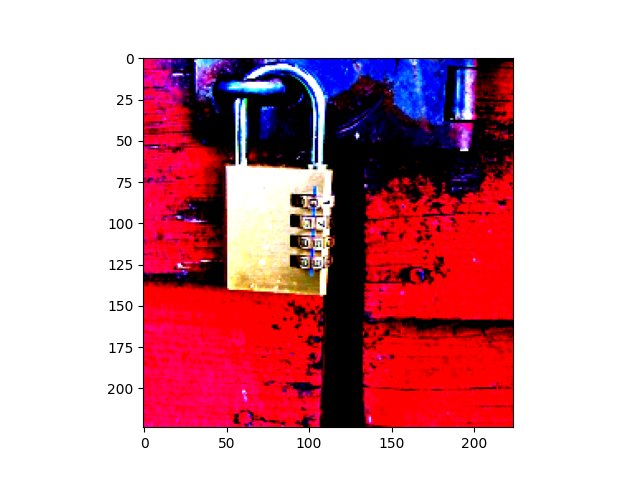}
\includegraphics[width=0.32\linewidth]{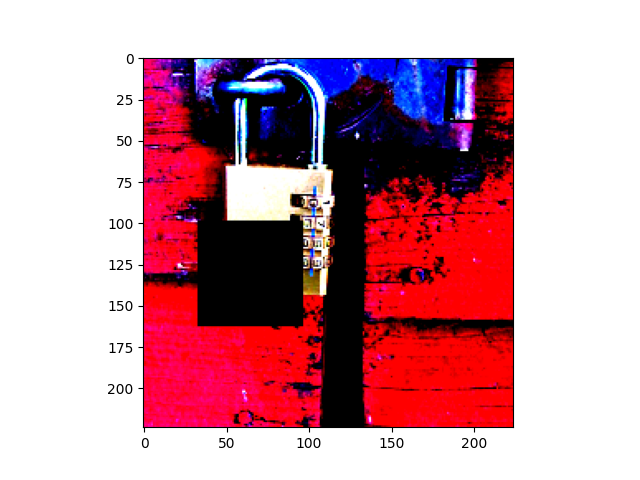}
\includegraphics[width=0.32\linewidth]{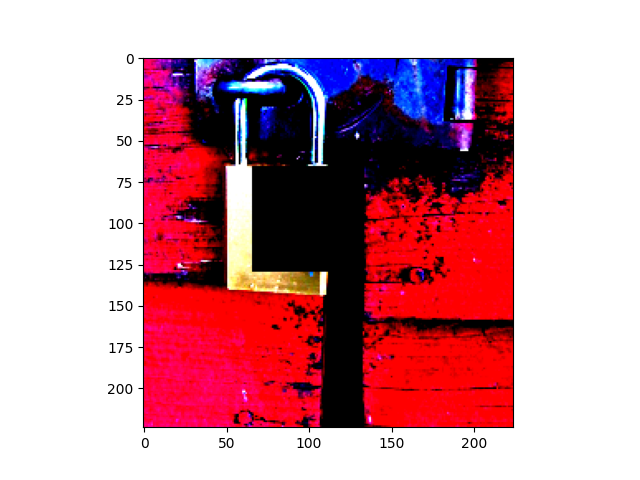}

\includegraphics[width=0.32\linewidth]{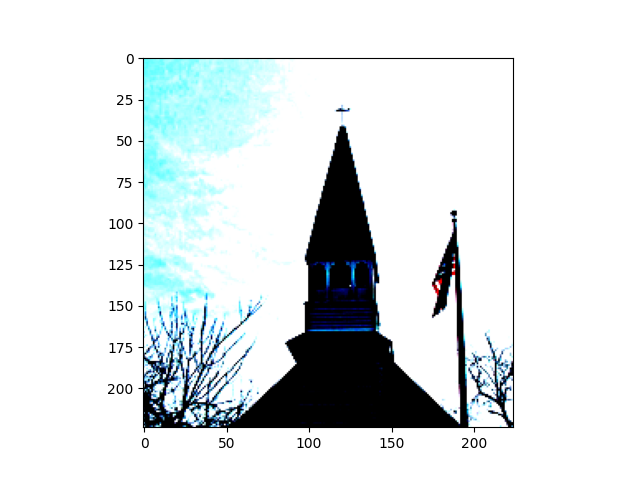}
\includegraphics[width=0.32\linewidth]{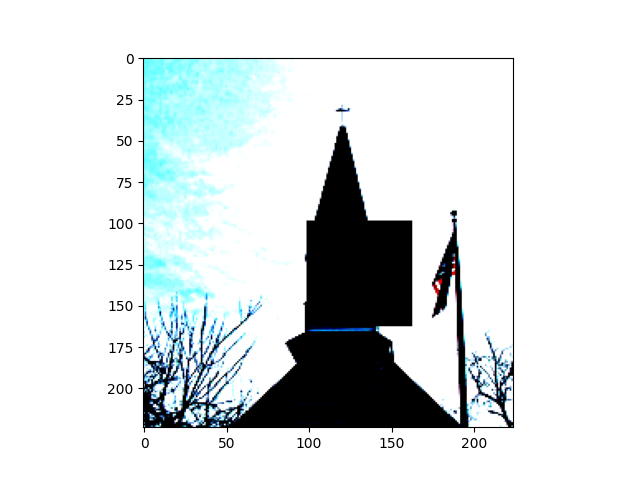}
\includegraphics[width=0.32\linewidth]{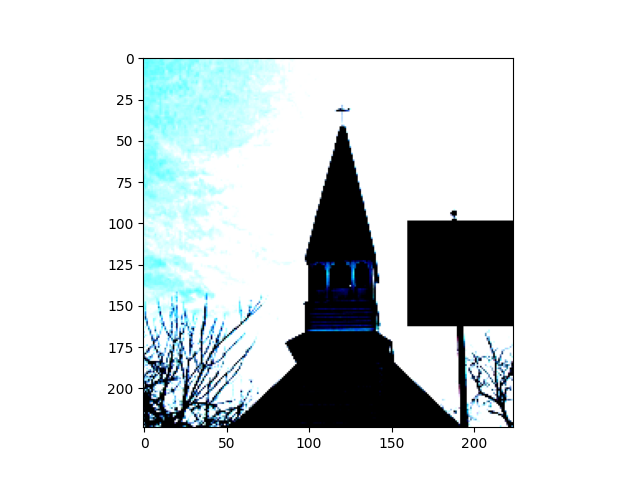}
\caption{
Examples of hard samples that HiCert finds hard to certify in a high $\tau$.
}\label{fig:hard_sample}

\end{figure*}

\thispagestyle{empty}
\section{A Case Study on Shapes/numbers of Patches In HiCert}\label{sec:shape/number}
The covering mask set can be adjusted to handle multiple/rectangular patches based on the analysis in Section 5.1 of \cite{xiang2022patchcleanser}.
Multiple patches can use multiple masks on one mutant for covering, and a patch in an arbitrary rectangle can be covered by a general set of rectangle covering masks.
We demonstrate the performance of HiCert in handling two square patches and one rectangle patch with 5000 random ImageNet samples, with other experimental settings the same as those in RQ1 for ImageNet with MAE. 
We adopt 1\% as the total patch area in this case study since PatchCleanser \cite{xiang2022patchcleanser} formally proves the correctness of the common rectangular covering mask set (which we adopt) for all possible rectangle shapes that consist of 1\% image pixels.
The results are in Table~\ref{tab:different_mode}.
From Table~\ref{tab:different_mode}, we observe a slight decrease in $acc_{cert}$ and $r_{cert}$ for the rectangle and two-square modes, with reductions less than 2\% and 4\%, respectively.
The drop of $r_{cert_{inc}}$ is respectively by 6.4\% and 13.0\% for the rectangle and two-square configurations. 
For secondary metrics, $acc_{\neg w}$ is steady within 1\% for both configurations. 
$r_{fa}$ increase by 8.5\% and 14.1\%, and $r_{fs}$ decrease by 2.6\% and 3.2\%, respectively for configurations of one rectangle and two squares.
Overall, the effect of multiple patches (i.e., two squares) is larger than a patch in a different shape (i.e., an arbitrary rectangle), 
however, 
HiCert can largely preserve certification performance, 
at the cost of a modest increase in false alerts.

\section{A Case study for trade-off between performance and time cost}\label{sec:time}
We conduct a case study on the trade-off between the performance of HiCert (in terms of $r_{cert}$ and $f_{ra}$) and time cost by 5000 random ImageNet samples with one patch in patch size 2\%.
We adopt the method for varying the number of masks in the range [$6^2, 5^2, 4^2, 3^2, 2^2$] in a covering mask set in Section 3.4 of \cite{xiang2022patchcleanser}, where a larger mask area for a single mask results in fewer masks being included in the covering mask set.
%
Table~\ref{tab:number_of_mask} illustrates the relationship between per-sample runtime and the number of masks/mutants in the covering mask set. Notably, reducing the number of mutants from 36 ($=6^2$) to 4 ($=2^2$) shortens the runtime by approximately 87.4\%, with an additional 15.4 ms relative to the runtime of processing the original input alone, without any mutant generation or inference.
The trade-off between $r_{cert}/f_{fa}$ is shown in Fig.~\ref{fig:runtime}.
We can observe that the $r_{cert}$ of HiCert is almost insensitive to the decrease of runtime, which aligns with our experiment shown in Fig.~\ref{fig:vary_patch_size} that a large mask would not largely affect the certification performance
when using MAE as the base model for ImageNet.
On the other hand, 
$r_{fa}$ also remains stable as the runtime decreases from 155 ms to 66 ms, and increases by about 5\% when the runtime is further reduced to 20 ms.
\thispagestyle{empty}

\section{A Case Study of visualization of the trade-off between $r_{cert}$ and $r_{fa}$ as $\tau$ varies}\label{sec:visual}
In Fig.~\ref{fig:visual_t}, we visualize the trade-off between $r_{cert}$ and $r_{fa}$ by varying $\tau$ in [0,1] with each step 0.01 for HiCert, where we adopt the same settings as those for RQ1 on ImageNet with MAE on 2\% patch size. 
We can observe when $\tau$ increase, both $r_{cert}$ and $r_{fa}$ increase.
Both $r_{{cert}}$ and $r_{{fa}}$ are relatively insensitive when $\tau \in [0, 0.4]$ and increase steadily as $\tau$ approaches 0.8. Notably, $r_{{fa}}$ rises sharply once $\tau$ exceeds 0.8.
Users of HiCert may choose a larger $\tau$ (e.g., $\tau$ = 0.8) to protect more benign samples in safety-critical applications or a smaller $\tau$ to reduce false alerts.

\section{Qualitative analysis of hard samples failed to be certified}\label{sec:hard sample}
Upon manual inspection, we find that these hard samples fall into two main categories:
(1) inputs containing two (or more) items from different classes, where masking one item causes the mutant to be predicted as the other class; and
(2) inputs containing a single item, where the mask changes its semantics, leading to misclassifications of the masked mutants.

Fig.~\ref{fig:hard_sample} presents two representative examples, one from each of the two hard sample groups.
The upper three inputs are, respectively, the image of the combination lock and its two mutants, from left to right.
The original input with the label ``combination lock'' can be correctly predicted by the classifier.
When the mask of the mutant covers the position that does not cover the combination dial (e.g., the mutant shown in the middle), the classifier still predicts the mutant as a combination lock.
However, when the combination dial is masked, shown in the mutant on the right, the semantics of ``combination lock are lost and changed into a ``padlock'' without the combination dial in this image, and the classifier inevitably predicts this mutant as ``padlock'' with high confidence (0.95), failing to be certified by HiCert.
To handle this kind of hard samples, a promising future direction would be to make use of the content under the mask. Note that all the existing masking strategies in masking-based detection, to our knowledge, including the one used in HiCert, unavoidably make the mask larger than the patch to decrease the computation cost, which means there is still some original content under the mask even under attacks. 
Leveraging this information makes it possible to address cases where the patch actually fails to alter the semantics, yet the mask does.

The lower three inputs, from left to right, depict the original image containing both a church and a flagpole (labeled as ``flagpole"), followed by two of its mutants. Although the presence of the church introduces noise for the ``flagpole" label, the classifier still correctly predicts the original input as ``flagpole''. 
When the church is masked out in the middle mutant, the classifier continues to predict ``flagpole". However, when the flag is masked in the final mutant, the classifier instead predicts ``church'' with high confidence (0.96), causing HiCert to fail in certifying the original input.
This category of hard samples highlights the need for future research on certification methods adapted to multi-label classification. 
While prior work has addressed certified robustness against various types of attacks, to the best of our knowledge, no existing studies have specifically focused on certified detection against patch attacks. 
We believe this direction holds significant promise, as real-world inputs may comprise a mixture of single-class and multi-class content.

\renewcommand{\thefootnote}{}
\renewcommand{\thefootnote}{\arabic{footnote}} 
\thispagestyle{empty}

\end{document}